\documentclass[11pt,a4paper]{article}
\usepackage{fullpage}
\usepackage{amsmath, amsthm, amsfonts,  amssymb, latexsym}
\usepackage[colorlinks, citecolor=Red, linkcolor=Blue, urlcolor=blue]{hyperref}
\usepackage{caption}
\usepackage[dvipsnames]{xcolor}
\usepackage{tikz}
\usepackage{mcite}
\usepackage{enumerate}

\usepackage{authblk}

\newtheorem{proposition}{Proposition}[section]
\newtheorem{theorem}[proposition]{Theorem}
\newtheorem{lemma}[proposition]{Lemma}

\newtheorem{corollary}[proposition]{Corollary}
\newtheorem{definition}[proposition]{Definition}

\newtheorem{remark}[proposition]{Remark}

\newcommand{\rhs}{r.h.s.\ }
\newcommand{\lhs}{l.h.s.\ }
\newcommand{\wrt}{w.r.t.\ }
\newcommand{\cf}{cf.\ }

\makeatletter
\DeclareRobustCommand{\ibar}{\mathord{%
  \text{$\m@th\mkern-2mu\raisebox{-0.7ex}[0pt][0pt]{$\mathchar'26$}\mkern-7mu i$}%
}}
\makeatother

\newcommand{\eti}[1]{e_\otimes^{\ibar #1}}
\newcommand{\et}[1]{e_\otimes^{#1}}

\newcommand{\ud}{\mathrm{d}}
\newcommand{\del}{\partial}

\newcommand{\R}{\mathbb{R}}

\newcommand{\g}{\mathfrak{g}}
\newcommand{\p}{\mathfrak{p}}
\newcommand{\skal}[2]{\langle #1 , #2 \rangle}

\newcommand{\LB}[2]{\lfloor #1 , #2 \rfloor}

\newcommand{\id}{\mathrm{id}}

\newcommand{\sS}{\mathcal{S}}

\newcommand{\cR}{\mathcal{R}}

\newcommand{\cD}{\mathcal{D}}
\newcommand{\cV}{\mathcal{V}}

\newcommand{\cU}{\mathcal{U}}

\newcommand{\cL}{\mathcal{L}}

\newcommand{\fD}{\mathfrak{D}}

\newcommand{\eps}{\varepsilon}

\newcommand{\ia}{{\mathrm{int}}}
\newcommand{\ret}{{\mathrm{r}}}
\newcommand{\adv}{{\mathrm{a}}}
\newcommand{\YM}{{\mathrm{YM}}}
\newcommand{\EH}{{\mathrm{EH}}}
\newcommand{\lin}{{\mathrm{lin}}}

\newcommand{\loc}{{\mathrm{loc}}}

\newcommand{\us}{{\underline s}}

\newcommand{\nn}{\nonumber}
\newcommand{\beq}{\begin{equation}}
\newcommand{\eeq}{\end{equation}}

\newcommand{\defeq}{\mathrel{:=}}
\newcommand{\eqdef}{\mathrel{=:}}
\newcommand{\eqos}{\mathrel{\approx}}
\newcommand{\eqF}{\mathrel{\approx_\bF}}
\newcommand{\Lie}{\mathcal{L}}
\newcommand{\vol}{\mathrm{vol}}
\newcommand{\rsc}{{\mathrm{sc}}}

\DeclareMathOperator{\Ker}{Ker}
\DeclareMathOperator{\Ran}{Im}
\DeclareMathOperator{\Tr}{Tr}
\DeclareMathOperator{\supp}{supp}
\DeclareMathOperator{\Deg}{Deg}
\DeclareMathOperator{\Ad}{Ad}
\DeclareMathOperator{\ad}{ad}

\newcommand{\cO}{{\mathcal{O}}}
\newcommand{\cA}{{\mathcal{A}}}
\newcommand{\bW}{{\mathbf{W}}}
\newcommand{\bJ}{{\mathbf{J}}}
\newcommand{\bF}{{\mathbf{F}}}
\newcommand{\bphi}{{\bar \phi}}
\newcommand{\bvp}{{\bar \varphi}}
\newcommand{\ba}{{\bar a}}

\usepackage[T1]{fontenc}
\usepackage{libertine}

\usepackage[toc,section=section]{glossaries}
\usepackage[nottoc,notlot,notlof]{tocbibind} 

\makenoidxglossaries

  \newglossaryentry{delta-A}
{
        name=\ensuremath{\ba},
        description={Background variation satisfying $\bar P^\lin \ba=0$ (tangent vector fields on $\mathcal{S}_\YM$)}
}

 \newglossaryentry{03barcA}
{
        name=\ensuremath{\bar{\cA}},
        description={Background $G$-connection}
}

 \newglossaryentry{04A}
{
        name=\ensuremath{A},
        description={Dynamical $\g$-valued 1-form}
}

\newglossaryentry{hat-A-n}
{
        name=\ensuremath{A^\ia_n},
        description={Interacting anomaly with $n$ local insertions}
}

 \newglossaryentry{ret-prop}
{
        name=\ensuremath{\Delta^{\ret/\adv}_{\bphi}},
        description={Retarded/advanced propagator of $P_{\bphi}$}
}

 \newglossaryentry{delta-bphi-variation}
{
        name=\ensuremath{\bar{\delta}_{\bvp}},
        description={Variation \wrt background field $\bar{\delta}_{\bvp} = \skal{\frac{\delta}{\delta \bphi} - }{\bvp}$}
}

 \newglossaryentry{delta-phi-variation}
{
        name=\ensuremath{\delta_{\bvp}},
        description={Variation \wrt dynamical field $\delta_{\bvp} = \skal{\frac{\delta}{\delta {\phi}}- }{ \bvp}$}
}

 \newglossaryentry{delta-ret}
{
        name=\ensuremath{\delta^\ret_{\bvp}},
        description={Retarded variation}
}

 \newglossaryentry{cD-phi}
{
        name=\ensuremath{\cD_{\bvp}},
        description={Connection on local functionals $F[\bphi, \phi]$, $\cD_{\bvp} = \bar{\delta}_{\bvp}- \delta_{\bvp}$}
}

 \newglossaryentry{fD-phi}
{
        name=\ensuremath{\fD_{\bvp}},
        description={Connection on sections $\bphi \mapsto T^\ia_\bphi(\eti{F[\bphi, -]})$, $\fD_{\bvp} =\delta^\ret_{\bvp}- \delta_{\bvp}$}
}

 \newglossaryentry{cD-A}
{
        name=\ensuremath{\cD_{\ba}},
        description={Connection on local functionals $F[\bar{\cA}, A]$, $\cD_{\ba} = \bar{\delta}_{\ba}- \delta_{\ba}$}
}

 \newglossaryentry{tilde-cD-A}
{
        name=\ensuremath{\hat{\cD}_{\ba}},
        description={Connection on local functionals $F[\bar{\cA}, A]$, $\hat \cD_{\ba} = {\cD}_{\ba}  - (-,  \cD_{\ba} \Psi)$}
}

 \newglossaryentry{fD-A}
{
        name=\ensuremath{\fD_{\ba}},
        description={Connection on sections $\bar \cA \mapsto T^\ia_{\bar{\cA}}(\eti{F[\bar \cA, -]})$}
}

 \newglossaryentry{check-cD-A}
{
        name=\ensuremath{\check{\cD}_{\ba}},
        description={Connection on local functionals $F[\bar{\cA}, A]$, $\check{\cD}_{\ba} = {\cD}_{\ba}  + (-,  \delta_{\ba} \Psi)$}
}

 \newglossaryentry{Deg}
{
        name=\ensuremath{\Deg},
        description={Grading on functionals, $\Deg = \deg_\phi + 2 \deg_{\hbar}$}
}

 \newglossaryentry{eqos}
{
        name=\ensuremath{\eqos},
        description={Equal modulo free field equations}
}

 \newglossaryentry{eqF}
{
        name=\ensuremath{\eqF},
        description={Equality in $\bF_{\bar \cA}$}
}

 \newglossaryentry{delta-phi}
{
        name=\ensuremath{\bvp},
        description={Background variation satisfying $P_{\bphi} \bvp=0$ (tangent vector fields on $\mathcal{S}_{\Phi^4}$)}
}

\newglossaryentry{01bphi}
{
        name=\ensuremath{\bphi},
        description={Background scalar field}
}

\newglossaryentry{02phi}
{
        name=\ensuremath{\phi},
        description={Dynamical scalar field}
}
 
 \newglossaryentry{barF}
{
        name=\ensuremath{\bar{F}},
        description={curvature of $\bar \cA$}
}

 \newglossaryentry{F-A}
{
        name=\ensuremath{\bF_{\bar{\cA}}},
        description={Local algebra of physical, gauge-invariant interacting quantum fields on $\bar{\cA}$}
}

 \newglossaryentry{W-bundle-coh}
{
        name=\ensuremath{\bF_{\YM}},
        description={Bundle of physical, gauge-invariant interacting quantum fields for Yang-Mills theory}
}

 \newglossaryentry{ibar}
{
        name=\ensuremath{\ibar},
        description={$\frac{i}{\hbar}$}
}

 \newglossaryentry{cJ0}
{
        name=\ensuremath{\bJ_\bphi},
        description={Ideal of free equations of motion in background $\bphi$}
}

 \newglossaryentry{J_-}
{
        name=\ensuremath{J^\pm(\cL)},
        description={Causal future/past of $\cL$}
}

 \newglossaryentry{bar-nabla}
{
        name=\ensuremath{\bar{\nabla}_\mu},
        description={Background covariant derivative}
}

 \newglossaryentry{cO}
{
        name=\ensuremath{\cO},
        description={Field}
}

 \newglossaryentry{cO-phi}
{
        name=\ensuremath{\cO_{\bphi}^\ia(x)},
        description={Renormalized interacting fields at $\bphi$ (sections of $\textbf{W}_{\Phi^4}$)}
}

 \newglossaryentry{Omega}
{
        name=\ensuremath{\Omega^k},
        description={Bundle of $k$ forms on $M$}
}

 \newglossaryentry{Psi}
{
        name=\ensuremath{\Psi},
        description={Gauge-fixing fermion}
}

 \newglossaryentry{p-bundle}
{
        name=\ensuremath{\p},
        description={Associated bundle $\p = P \times_{\ad} \g$}
}

 \newglossaryentry{P-lin-A}
{
        name=\ensuremath{\bar P^\lin},
        description={Linearized Yang-Mills operator around $\bar{\cA}$}
}

 \newglossaryentry{p-phi}
{
        name=\ensuremath{P_{\bphi}},
        description={Linearized $\Phi^4$ differential operator around $\bphi$}
}

 \newglossaryentry{q}
{
        name=\ensuremath{q},
        description={Quantum BV-BRST differential, $q= s + O(\hbar)$}
}

 \newglossaryentry{Q-quant}
{
        name=\ensuremath{Q^\ia_{\bar{\cA}}},
        description={Quantum BRST charge}
}

 \newglossaryentry{cR}
{
        name=\ensuremath{\cR},
        description={Localization region of observables}
}

 \newglossaryentry{R-phi}
{
        name=\ensuremath{R_{\bphi}},
        description={Retarded products valued in $\textbf{W}_{\bphi}$}
}

 \newglossaryentry{S-phi4}
{
        name=\ensuremath{\mathcal{S}_{\Phi^4}},
        description={Manifold of on-shell background $\Phi^4$ configurations}
}

 \newglossaryentry{S-YM}
{
        name=\ensuremath{\mathcal{S}_\YM},
        description={Manifold of on-shell background Yang-Mills connections}
}

 \newglossaryentry{s}
{
        name=\ensuremath{s},
        description={BV-BRST differential}
}

 \newglossaryentry{S-action}
{
        name=\ensuremath{S},
        description={Gauge-fixed action, $s = {(S, - )} $ }
}

 \newglossaryentry{T-phi}
{
        name=\ensuremath{T_{\bphi,n}},
        description={Time-ordered products valued in $\bW_{\bphi}$}
}
 \newglossaryentry{T-A}
{
        name=\ensuremath{T_{\bar \cA,n}},
        description={Time-ordered products valued in $\bW_{\bar{\cA}}$}
}

\newglossaryentry{T-int-phi}
{
        name=\ensuremath{T^\ia_{\bphi}},
        description={Interacting time-ordered products valued in $\bW_{\bphi}$}
}

 \newglossaryentry{W}
{
        name=\ensuremath{\bW_{\bphi}},
        description={Algebra of quantum fields for the background $\bphi$}
}

 \newglossaryentry{W_loc_phi}
{
        name=\ensuremath{\bW^\loc_\bphi},
        description={Local functionals on the background $\bphi$}
}

 \newglossaryentry{W_A}
{
        name=\ensuremath{\bW_{\bar \cA}},
        description={Algebra of quantum fields for the background $\bar \cA$}
}

 \newglossaryentry{W-bundle}
{
        name=\ensuremath{\bW^\ia_{\Phi^4}},
        description={Algebra bundle of $\Phi^4$-theory}
}
 \newglossaryentry{W-bundle-A}
{
        name=\ensuremath{\bW_{\YM}},
        description={Algebra bundle of Yang-Mills theory}
}

 \newglossaryentry{(F1,F2)}
{
        name=\ensuremath{{(F_1, F_2)}},
        description={Anti-bracket}
}

 \newglossaryentry{(O1,O2)-q}
{
        name=\ensuremath{{(\mathcal{O}_1,\mathcal{O}_2)}_\hbar},
        description={Quantum anti-bracket, ${(-,-)}_\hbar= {(-, - )} + O(\hbar)$ }
}

 \newglossaryentry{commutator}
{
        name=\ensuremath{{[\phi(x), \phi(y)]}_{\star}},
        description={Graded commutator of quantum fields}
}

 \newglossaryentry{Lie-bra-vect}
{
        name=\ensuremath{\LB{\bvp}{\bvp'}},
        description={Lie bracket of vector fields on $\mathcal{S}_{\Phi^4}/ \mathcal{S}_\YM$}
}

 \newglossaryentry{Lie-bra-g}
{
        name=\ensuremath{[-,-]_\g},
        description={Lie bracket on the Lie algebra $\g$}
}

\begin{document}

\title{Background independence in gauge theories}

\author[1, 2]{Mojtaba Taslimi Tehrani \thanks{motaslimi@gmail.com}}
\author[2]{Jochen Zahn\thanks{jochen.zahn@itp.uni-leipzig.de}}
\affil[1]{Max-Planck Institute for Mathematics in the Sciences\\
Inselstr.\ 22, 04103 Leipzig, Germany}
\affil[2]{Institut f\"ur Theoretische Physik, Universit\"at Leipzig\\ Br\"uderstr.\ 16, 04103 Leipzig, Germany}

\date{\today}

\maketitle

\begin{abstract}
Classical field theory is insensitive to the split of the field into a background configuration and a dynamical perturbation. In gauge theories, the situation is complicated by the fact that a covariant (\wrt the background field) gauge fixing breaks this split independence of the action. Nevertheless, background independence is preserved on the observables, as defined via the BRST formalism, since the violation term is BRST exact. In quantized gauge theories, however, BRST exactness of the violation term is not sufficient to guarantee background independence, due to potential anomalies.
We define background independent observables in a geometrical formulation as flat sections of the observable algebra bundle over the manifold of background configurations, with respect to a flat connection which implements background variations. A theory is then called background independent if such a flat (Fedosov) connection exists. 
We analyze the obstructions to preserve background independence at the quantum level for pure Yang-Mills theory and for perturbative gravity. 
We find that in the former case all potential obstructions can be removed by finite renormalization.
In the latter case, as a consequence of power-counting non-renormalizability, there are infinitely many non-trivial potential obstructions to background independence.
We leave open the question whether these obstructions actually occur.
\end{abstract}
\newpage
\tableofcontents
\newpage
\section{Introduction}

In Quantum Field Theory (QFT), one frequently considers the quantum fluctuations around classical field configurations. Examples are:
\begin{itemize}
\item Spontaneous symmetry breaking in the Standard Model, where one considers quantum fluctuations around a non-trivial classical configuration of the Higgs field;
\item The background field method, which is an efficient tool, for example, for the computation of the renormalization group flow (see e.g. \cite{Abbott:1981ke});
\item Perturbative Quantum Gravity, where one has to use a non-trivial background metric, providing the necessary structure for the formulation of a QFT \cite{Brunetti:2013maa, Becker:2014qya}.
\end{itemize} 

Hence, the issue of \emph{background independence} seems to be of high conceptual importance. Apart from the discussion in \cite{Brunetti:2013maa}, on which we comment in detail below, there are basically two approaches to deal with it in the literature. One is the Riemannian path integral framework, which faces the problem that, in the presence of non-trivial background fields, the relation between correlation functions on Riemannian spaces, and the QFT on Lorentzian space-time in which one is ultimately interested, is unclear. In particular, in the absence of an Osterwalder-Schrader theorem, it is not clear whether such correlation functions define a QFT in the sense of observables represented by operators on some Hilbert space. The other approach, discussed in more detail at the end of this section, is to treat the background field as an infinitesimal perturbation around a fixed flat reference background. However, for a full proof of background independence, one should treat the background field non-perturbatively. Then one faces the problem that on generic backgrounds there is no unique vacuum state and that the usual renormalization techniques based on momentum space are not available. A further common shortcoming of these approaches is that they are not ``operational'' in the sense that they do not address the following question:

\begin{quote}
Given a background configuration and an observable defined \wrt this background,
what is the same observable on a different background?
\end{quote}

In view of the difficulties mentioned above, we follow the algebraic approach, i.e.,  we directly (perturbatively) construct the algebras of observables for the different background configurations, using locally covariant renormalization techniques developed in the context of QFT on curved space-times \cite{Hollands:2001nf}. Background independence for us then means that we can unambiguously identify observables on different backgrounds (at least for infinitesimally close backgrounds). As suggested in \cite{Hollands:unpublished, *Hollands:talk}, this can be formulated in the spirit of \emph{Fedosov quantization} \cite{fedosov1994}: One considers the bundle of observable algebras over the manifold of background configurations and constructs a flat connection on it. The sections that are flat, i.e., covariantly constant, \wrt this connection provide a consistent assignment of an observable to each background. 
The similarity of background independence and Fedosov's approach has already been noted, in a quantum mechanical framework, in \cite{ReuterMetaplectic}.\footnote{Also in the context of string (field) theory, background independence was studied in terms of a (flat) connection, though not on the observable algebra bundle, see e.g. \cite{Witten:1993ed,Sen:1993kb}.}

\subsubsection*{The scalar field as a toy model}

To motivate our definition of background independence and to introduce some of the relevant concepts, let us first discuss a toy model, namely the self-interacting $\Phi^4$-theory. Consider splitting the basic scalar field
\beq
\label{barphi+phi}
\Phi = \bphi + \phi,
\eeq
into a \emph{background configuration} $\gls{01bphi}$, which is kept classical at the quantum level (i.e.,  it commutes with all quantum fields) and a \emph{dynamical field} $\gls{02phi}$ which is viewed as fluctuations around $\bphi$ and is quantized in perturbation theory. The question of background independence is then the following: Is field theory independent of the splitting of $\Phi$ into a background $\bphi$ and a perturbation $\phi$?
Clearly, the action functional $S[\Phi]$ depends only on the combination $\bphi + \phi$, hence the classical field theory is independent of this split. We say that it exhibits 
\emph{split independence}.
Here we ask whether and in which mathematically rigorous sense this split independence is preserved at the quantum level.

To analyze the issue, we find it convenient to adopt the framework of locally covariant quantum field theory \cite{Hollands:2001nf, Brunetti:2001dx} which has proven to be powerful for QFT in curved space-time or in the presence of non-trivial background gauge connections \cite{Zahn:2012dz}.  In this framework, the covariance with respect to suitable transformations of background data (e.g.\ isometries of the background metric or gauge transformations of the background connection) is manifest by construction. 
The objects of primary interest are renormalized \emph{interacting time-ordered products}, which include \emph{interacting fields}. They are constructed in perturbation theory and generate the non-commutative local algebra of observables.

More concretely, for the example of scalar field theory expanded around a classical solution $\bphi$ of $\Phi^4$-theory, one constructs for each such background $\bphi$, the local algebra $\bW_{\bphi}$.
To each classical local functional $F[\bphi, \phi]$, one associates the generating functional $T^\ia_\bphi(e^{i F})$ of interacting time-ordered products, which is an element of $\bW_{\bphi}$. These elements generate the algebra $\bW_\bphi^\ia$ of interacting observables.
Now consider a local functional $F[\Phi]$. Obviously, it induces local functionals $F[\bphi, \phi] = F[\bphi + \phi]$ for the different backgrounds $\bphi$. Their background independence can be stated via functional derivatives as
\beq
\label{eq:bg_ind_functional}
 \gls{cD-phi} F \defeq ( \gls{delta-bphi-variation} - \gls{delta-phi-variation} ) F \defeq
 \skal{ ( \tfrac{\delta}{\delta \bphi} - \tfrac{\delta}{\delta \phi} ) F}{\bvp} = 0,
\eeq
where $\bvp$ is some variation of the background.
The question is how to implement this on the quantum observables $T^\ia_\bphi(e^{i F})$. While the second derivative (\wrt the dynamical field $\phi$) is well-defined on  $\bW_{\bphi}$, the first derivative (\wrt the background field $\bphi$) has no obvious meaning on $\bW_{\bphi}$,
as one is comparing elements of different algebras.\footnote{Even if one interprets $T^\ia_\bphi(e^{i F[\bphi, -]})$ as an evaluation functional, which one can differentiate \wrt $\bphi$, this operation is not well-defined on the on-shell algebra. Furthermore, we would like a differentiation that respects the algebraic structure, i.e., fulfills the Leibniz rule \wrt the algebra product. This will not be the case for such a naive derivative.}
The way out is to replace this derivative with the \emph{retarded variation} $\delta^\ret$ \cite{hollands2005conservation}, which is the infinitesimal version of the M{\o}ller operator relating the algebras on the different backgrounds \cite{Brennecke:2007uj} (see below). The natural translation of \eqref{eq:bg_ind_functional} to an assignment 
\beq
\label{eq:bphi_to_T_int}
\bphi \mapsto T^\ia_\bphi(e^{i F[\bphi, -]})
\eeq
of interacting fields to different backgrounds is thus
\beq
\label{eq:bg_ind_observable}
\gls{fD-phi} T^\ia_\bphi(e^{i F[\bphi, -]}) \defeq ( \delta^\ret_{\bvp} - \delta_{\bvp} ) T^\ia_\bphi(e^{i F[\bphi, -]}) = 0.
\eeq
It turns out \cite{Hollands:unpublished}, \cf \cite{collini2016fedosov} for details, that in the $\Phi^4$-theory, this is equivalent to \eqref{eq:bg_ind_functional} in the sense that
\beq
\label{eq:relation_fD_cD}
 \fD_{\bvp} T^\ia_\bphi(e^{i F[\bphi, -]}) = i T^\ia_\bphi( \cD_{\bvp} F \otimes e^{i F[\bphi, -]})
\eeq
precisely if \emph{perturbative agreement}\footnote{Perturbative agreement asserts that it should, on the infinitesimal level, not matter whether one puts terms quadratic in $\phi$ into the free part or the interacting part of the action.}
\cite{hollands2005conservation} holds for changes in the (position dependent) mass of the scalar field.\footnote{The free theories linearized around different backgrounds $\bphi$ differ in the mass term, see below.} As a consequence of the flatness of $\cD$, also $\fD$ is then flat. For variations in the mass, perturbative agreement can be fulfilled \cite{collini2016fedosov,Drago2017}, so that \eqref{eq:relation_fD_cD} indeed holds.

It is natural to give this a geometric interpretation along the lines of Fedosov quantization, as suggested in \cite{Hollands:unpublished, *Hollands:talk} (see \cite{collini2016fedosov} for details).
Consider the manifold $\mathcal{S}_{\Phi^4}$ of solutions to the interacting $\Phi^4$ field equations. The tangent space at each $\bphi \in \sS_{\Phi^4}$ is the space of solutions $\bvp$ of the field equations linearized around $\bphi$. We patch all algebras $\bW^\ia_{\bphi}$ together to obtain the algebra bundle
\begin{equation*} 
\gls{W-bundle} = \bigsqcup_{\bphi}\bW^\ia_{\bphi} \to \mathcal{S}_{\Phi^4}.
\end{equation*} 
An assignment $\bphi \mapsto T^\ia_\bphi(e^{i F[\bphi, -]})$ as above is then interpreted as a section of $\bW^\ia_{\Phi^4}$, and $\fD_{\bvp}$ as a covariant derivative (connection) on this bundle in the direction of the vector field $\bvp$. If this connection is flat, we call the QFT \emph{background independent}. 
Flatness ensures that, at least formally, any interacting observable on one background can be uniquely parallel transported to any other background, providing an answer to the question posed at the beginning of this section.
Or, in the spirit of Fedosov quantization: The space of sections of $\bW^\ia_{\Phi^4}$ is much larger than the space of functionals of $\Phi$, i.e., the space of functions on $\sS_{\Phi^4}$. 
However, when restricting to sections that are flat \wrt $\fD$, i.e., fulfill \eqref{eq:bg_ind_observable}, one obtains a one-to-one correspondence between functions on $\sS_{\Phi^4}$ and flat sections of $\bW^\ia_{\Phi^4}$. Again, flatness of $\fD_{\bvp}$ is crucial.

\subsubsection*{Gauge theories}

The main aim of the present work is to analyze the issue of background independence for gauge theories where more complications arise due to gauge-fixing. Let us for definiteness consider the pure Yang-Mills theory which is the theory of a $G$-connection $\cA$ on a principal bundle, subject to the Yang-Mills field equations. We split
\begin{equation}
\label{eq: A=barA+A}
\cA = \bar{\cA} + A,
\end{equation}
into a background connection $\gls{03barcA}$ and a dynamical $\mathfrak{g}$-valued 1-form $\gls{04A}$ (a vector potential) which will be quantized in perturbation theory. $\bar \cA$ is a solution to the Yang-Mills equation.
Similar to the scalar case, the classical Yang-Mills action is independent of this split.  However, for the purpose of perturbative quantization, one has to fix the gauge, which necessarily breaks this split independence if one requires a covariant gauge fixing.
The gauge-fixed action exhibits a residual fermionic symmetry, the BV-BRST symmetry. It acts by a nilpotent operator $s$, and the physical (gauge invariant) observables are obtained as the cohomology of $s$. In fact, the violation of the split independence 
in the gauge-fixed action is $s$-exact. It follows that, classically, split independence holds at the level of gauge invariant observables, i.e., there is a flat connection $\hat \cD_\ba$ on classical local functionals that is well-defined on $s$ cohomology, i.e., $\hat \cD_\ba \circ s = s \circ \hat \cD_\ba$. Here $\ba$ is an infinitesimal variation of the background.

To quantize, one constructs, for each background $\bar \cA$, the (unphysical) algebra $\bW^\ia_{\bar \cA}$. The subalgebra $\bF_{\bar \cA} \subset \bW^\ia_{\bar \cA}$ of physical (gauge invariant) observables is given by the cohomology of $[Q^\ia_{\bar{\cA}},-]_\star$, where $Q^\ia_{\bar{\cA}}$ is the renormalized interacting BRST charge, and the commutator is taken \wrt the algebra $\star$ product.
Therefore, for background independence to hold, the desired connection $\fD_{\ba}$ has to be well-defined on the BRST cohomology, that is, it must satisfy
\beq
\label{eq:fD_well-defined}
\fD_{\ba} \circ [Q^\ia_{\bar{\cA}},-]_\star - [Q^\ia_{\bar{\cA}},-]_\star \circ \fD_{\ba} =0,
\eeq
on-shell. Furthermore, on the kernel of $[Q^\ia_{\bar{\cA}}, -]_\star$, the curvature of $\fD_{\ba}$ has to vanish modulo an element in the image of $ [Q^\ia_{\bar{\cA}},-]_\star$.
If this is the case, background independent observables can be defined as those sections of the observable algebra bundle which are flat \wrt $\fD_{\ba}$ modulo $\Ran [Q^\ia_{\bar{\cA}},-]_\star$.
We find that there are potential obstructions (anomalies) for the construction of such a connection. However, for pure Yang-Mills theory in $D=4$ space-time dimensions, these turn out to be trivial. Power counting renormalizability is a crucial ingredient of our proof. If the relevant anomaly is absent, then an identity analogous to \eqref{eq:relation_fD_cD} holds in $\bF_\YM = \sqcup_{\bar \cA} \bF_{\bar \cA}$, namely
\[
 \fD_\ba T^\ia_{\bar \cA}(e^{i F}) = i T^\ia_{\bar \cA} ( \{ \hat \cD_{\ba} F +\hat A_\ba(e^F) \} \otimes e^{i F}),
\]
where $\hat A$ incorporates quantum corrections. Hence, a classically gauge invariant and background independent local functional does not automatically give rise to a background independent observable at the quantum level, but quantum corrections may be necessary.

We also sketch the application of our framework to perturbative quantum gravity. As in any diffeomorphism-invariant theory, the definition of local observables is a major issue, and we follow recent proposals \cite{Brunetti:2013maa, Khavkine:2015fwa}, based on \cite{BergmannKomar}, for the construction of such (relational) observables employing a set of configuration-dependent covariant coordinates. 
As opposed to the pure Yang-Mills case, our analysis of potential anomalies to background independence shows that for the case of perturbative gravity one can indeed find infinitely many candidates for such anomalies using 
the dimensionful coupling of the theory. 
From this perspective, it seems difficult to prove the absence of anomalies, as they may appear at arbitrarily high order in perturbation theory.\footnote{In \cite{Brunetti:2013maa}, a different conclusion was found. We comment on the approach taken there in Section~\ref{section:pert-Grav}.}

We would like to point out that our work does not yet provide a full Fedosov quantization of Yang-Mills theories. First of all, one should then work on gauge equivalence classes of classical solutions as base space, not on the full space of classical solutions, as we do (but see \cite{Benini:2017zjv} for a different point of view). Second, the set of solutions to the Yang-Mills equation is a manifold only up to singular points corresponding to solutions with symmetries \cite{arms_1981}. We work locally in configuration space, i.e., in a neighborhood of a generic configuration, avoiding these singularities. We should also emphasize that the main focus of our work is algebraic, not (functional) analytic. In particular we do not discuss the analytical aspects of the infinite-dimensional manifolds of solution spaces, and algebra bundles upon these. We refer to \cite{collini2016fedosov} for a thorough discussion.

\subsubsection*{Comparison with the path integral approach}

Let us compare our treatment of background independence with more formal approaches, in particular the path integral formalism. In the case of the scalar field, one defines the generating functional of connected graphs as
\[
 \tilde W[J, \bphi] = - i \log \int D \phi \ e^{i (S[\bphi + \phi] + \int J \phi)},
\]
and the corresponding effective action as
\[
 \tilde \Gamma[\tilde \phi, \bphi] = \tilde W[J, \phi] - \int J \tilde \phi,
\]
with
\[
 \tilde \phi = \frac{\delta \tilde W}{\delta J}.
\]
Assuming that the path integral measure $D \phi$ is shift invariant, one obtains, with the shift $\phi \to \phi - \bphi$, that
\[
 \tilde \Gamma[\tilde \phi, \bphi] = \Gamma[\tilde \phi + \bphi],
\]
with $\Gamma$ the generating functional in the absence of the background field \cite{Abbott:1981ke}. In particular,
\beq
\label{eq:EffActionInvariance}
 \tilde \Gamma[\tilde \phi - \delta \phi, \bphi + \delta \phi] = \tilde \Gamma[\tilde \phi , \bphi],
\eeq
In this sense, background independence holds, provided that shift invariance of the path integral measure is fulfilled. One can thus see perturbative agreement as the rigorous version of the shift invariance of the formal path integral.\footnote{This interpretation was already suggested in \cite{hollands2005conservation}.}

Shift invariance of the path integral measure is also a crucial requirement in the treatment of background independence in gauge theory given in \cite{KlubergStern:1974xv}. However, as described above, this is not sufficient, as the gauge fixed action is not split independent. To deal with this, an extended BRST differential is introduced in \cite{KlubergStern:1974xv}, which also implements a shift between the background and the dynamical vector potential. It is then argued that the corresponding Slavnov identities can be fulfilled. As in our treatment, a crucial ingredient in that proof is power counting renormalizability, which restricts the number of possible counterterms.

Let us summarize two major conceptual differences between our treatment and the path integral approach:
\begin{itemize}
\item Typically, renormalization techniques are employed which require that the propagator is translation invariant. This means that the background is in fact treated perturbatively, i.e., it enters only the vertices, not the propagators. This entails that the background field is a vector potential, not a principal bundle connection and also that shift invariance of the path integral measure is trivially fulfilled. But the perturbative expansion with all the background fields in the vertices is ill-defined, unless the background field is treated as an infinitesimal perturbation, so that one may expand in powers of the background field.
Hence, only an infinitesimal neighborhood of a fixed flat reference connection is actually treated. In contrast, in our approach, the background connection is treated non-perturbatively.

\item A formulation of background independence such as \eqref{eq:EffActionInvariance} does not refer to observables, i.e., it does not address the question posed at the beginning of the introduction. For this, one would need to couple generic observables through source terms to the action and study the background independence of the resulting effective action. To the best of our knowledge, this has not been done in the literature.

\end{itemize}

\subsubsection*{Outline}

The article is structured as follows. To set the stage, we review, in the next section, the case of scalar field theory, in particular the construction of the algebras $\bW_{\bphi}$. Following \cite{Hollands:unpublished, collini2016fedosov}, the relation of background independence and perturbative agreement is discussed. In the main part of this work, Section~\ref{section:YM}, we study the case of Yang-Mills theories. Perturbative Quantum Gravity is treated in Section~\ref{section:pert-Grav}. 
An appendix contains technical lemmata.
For the convenience of the reader we provide a glossary of symbols used.

\section{Background independence for scalar field theory}\label{section: scalar-theory}

\subsection{Perturbative QFT on a background $\bphi$} \label{section:phi-4}

In this section, we review the discussion of background independence for a self-interacting scalar field $\Phi$ \cite{Hollands:unpublished, collini2016fedosov}. 
Throughout this work, we consider globally hyperbolic space-times $(M, g)$ with signature $(-, +, \dots +)$ and compact Cauchy surfaces. $\gls{J_-}$ denotes the causal future/past of a space-time region $\cL \subset M$, c.f., for example, \cite{WaldGR} for a definition.

Due to the time-slice axiom \cite{Chilian:2008ye}, it is sufficient to define the interacting observables localized in a causally closed, compact space-time region $\gls{cR} \subset M$ which contains a Cauchy surface. In particular, we may choose $\cR = J^+(\Sigma_0) \cap J^-(\Sigma_1)$ for two non-intersecting Cauchy surfaces $\Sigma_{0/1}$.
\label{def:cR}
We may thus replace the coupling constant $\lambda_0$ with a smooth compactly supported cutoff function $\lambda(x)$ which equals $\lambda_0$ on a neighborhood of $\cR$. 
For the perturbations $\phi$, we consider the expansion of the action 
\[
S[\Phi] = - \int \left( \tfrac{1}{2} \nabla_\mu \Phi \nabla^\mu \Phi + \tfrac{1}{2} m^2 \Phi^2 + \tfrac{1}{4!}\lambda \Phi^4 \right) \vol,
\]
around a background $\bphi$
\beq
\label{scalar action}
 S[\bphi, \phi] = - \int \tfrac{1}{2} \left( \nabla_\mu \phi \nabla^\mu \phi + ( m^2 + \tfrac{1}{2} \lambda \bphi^2 ) \phi^2 \right) \vol - \int \left( \tfrac{1}{3!} \lambda \bphi \phi^3 + \tfrac{1}{4!} \lambda \phi^4 \right) \vol \eqdef {S_0} + {S_\ia}.
\eeq
Note that the free Lagrangians for different backgrounds $\bphi$ coincide outside of the support of $\lambda$.
This is essential for identifying quantum theories around different backgrounds as discussed in the next section.
Also note that there is no source term in \eqref{scalar action}, i.e., a term linear in $\phi$, since the background configuration
is required to fulfill the interacting equation of motion
\beq
\label{eq:bphi_eom}
(\Box - m^2) \bphi + \tfrac{1}{3!} \lambda \bphi^3 =0.
\eeq

The solutions to \eqref{eq:bphi_eom} form a manifold $\gls{S-phi4}$, with tangent space $T_\bphi \sS_{\Phi^4}$ at $\bphi$ given by the solution space to the linearized equation of motion
\beq
\label{eq:bphi-lin-eom}
 \gls{p-phi} \gls{delta-phi} \defeq \left( \Box - m^2 - \tfrac{1}{2} \lambda \bphi^2 \right) \bvp = 0.
\eeq
This means that given a smooth curve $\{ \bphi_s \}_s$ in $\sS_{\Phi^4}$, i.e., of solutions to \eqref{eq:bphi_eom}, with $\bphi_0 = \bphi$, its derivative
\beq
\label{eq:bg-variation}
 \bvp \defeq \del_s \bphi_s|_{s = 0},
\eeq
is a solution to \eqref{eq:bphi-lin-eom}. We refer to \cite{collini2016fedosov} for details, in particular on the notion of smoothness.

Background independence of the classical scalar field theory now means that it is independent of the arbitrary split \eqref{barphi+phi} into background and dynamical fields. One manifestation of this is the \emph{split independence} of the action in the sense that
\beq
\label{eq:ShiftSymmetryScalar}
 \frac{\delta S}{\delta \bphi(x)} = \frac{\delta {S_\ia}}{\delta \phi(x)},
\eeq
where the interaction part $S_\ia$ of the action was defined in \eqref{scalar action}.

\subsubsection*{The free algebra $\bW_{\bphi}$}

The algebra $\gls{W}$ (also called the \emph{free algebra} in contrast to the interacting one defined below) consists of \emph{evaluation functionals}
\beq
\label{eq:def_functional}
 F[\phi] = \sum_{n=0}^N \int_{M^n} f_n(x_1, \dots, x_n) \phi(x_1) \dots \phi(x_n) \vol(x_1) \dots \vol(x_n),
\eeq
where the singularities of the symmetric distributions $f_n$ on $M^n$ are constrained by a condition on their \emph{wave front set}, \cf \cite{Hollands:2001nf}. We define the \emph{support} of a functional of the form \eqref{eq:def_functional} as
\[
 \supp F = \left\{ x \in M \mid (x, y_1, \dots, y_{n-1}) \in \supp f_n \text{ for some } n, y_i \in M \right\}.
\]

Given a \emph{Hadamard two-point function} $\omega_{\bphi}$ for $P_\bphi$, \cf \cite{Hollands:2001nf} for a definition, one defines a non-commutative $\star$ product
\beq
\label{eq:star-prod-def}
 F \star_{\omega_\bphi} G = \mathfrak{m} \circ \exp( \hbar \Gamma_{\omega_{\bphi}} ) (F \otimes G),
\eeq
where $\mathfrak{m}$ is the point-wise multiplication of functionals, $\mathfrak{m}( F \otimes G)(\phi) = F(\phi) G(\phi)$, and
\[
 \Gamma_{\omega_\bphi} (F \otimes G) = \int_{M^2} \omega_\bphi(x,y) \tfrac{\delta}{\delta \phi(x)} F \otimes \tfrac{\delta}{\delta \phi(y)} G. 
\]
Here the functional derivative $\frac{\delta}{\delta \phi(x)} F$ is interpreted as a $\bW_\bphi$ valued density, whose evaluation on test functions $\varphi$ is defined as
\[
 \skal{\tfrac{\delta}{\delta \phi} F}{\varphi}[\phi] \defeq \tfrac{\ud}{\ud \lambda} F[\phi + \lambda \varphi]|_{\lambda = 0}.
\]
The definition of the $\star$ product, and thus also that of $\bW_\bphi$, depends on the two-point function $\omega_\bphi$. However, it turns out that algebras equipped with $\star$ products defined by different Hadamard two-point functions $\omega_\bphi$, $\omega'_\bphi$ are isomorphic \cite{Hollands:2001nf}, justifying the notation.\footnote{Despite the explicit dependence of the star-product on the background, we drop the subscript $\bphi$ and simply write $\star$. }
Note that, in particular, \eqref{eq:star-prod-def} implies that
\[
\gls{commutator} = i \hbar \Delta_{\bphi}(x,y),
\]
where $\Delta_{\bphi} = \Delta^\adv_\bphi - \Delta^\ret_\bphi$ is the causal propagator of $P_\bphi$, with $\gls{ret-prop}$ denoting the retarded/advanced propagator.

Elements of $\bW_\bphi$ are considered in the sense of formal power series in $\hbar$, i.e.,  $\bW_\bphi$ is considered as a graded vector space with grading provided by $\deg_\hbar$, which counts the number of $\hbar$ factors. 
A further grading is given by
\begin{equation*}
 \gls{Deg} = 2 \deg_\hbar + \deg_\phi,
\end{equation*}
where $\deg_\phi$ counts the number of fields.  For example, for an $F$ of the form \eqref{eq:def_functional}, with $f_n \neq 0$ and $f_m = 0$ for all $m \neq n$, one has $\deg_\phi(F) = n$. It is obvious that the $\star$ product respects the grading, i.e., 
\[
 \Deg(F \star G) = \Deg(F) + \Deg(G).
\]
This grading is in fact the natural grading in the context of Fedosov quantization  \cite{fedosov1994}.

\emph{Local covariance} \cite{Hollands:2001nf, Brunetti:2001dx} is a crucial ingredient of our approach.\footnote{To get the full strength of requirement of local covariance, one should not restrict to space-times with compact Cauchy surfaces from the outset. We thus assume that we have implemented local covariance without this restriction, i.e., constructed algebras $\bW_\bphi$ and time-ordered products $T_{\bphi, n}$, \cf below, and then restricted to space-times with compact Cauchy surfaces.} It is implemented as follows: A \emph{morphism} $\psi: (M', g', \bphi') \to (M, g, \bphi)$ is an isometric embedding $\psi: M' \to M$, i.e., $\psi^*g =g'$, which preserves the causal structure, and such that $\psi^*\bphi =\bphi'$. For each morphism $\psi$, there exists an algebra homomorphism
\beq
\label{eq:alpha_psi}
 \alpha_\psi: \bW_{\bphi'} \to \bW_\bphi,
\eeq
defined by
\[
 (\alpha_\psi F)[\phi] \defeq F[ \psi^* \phi].
\]

To implement the equations of motion, one passes to the \emph{on-shell} algebra. This proceeds by quotienting out the ideal
\beq
\label{eq:ideal-J_0}
\gls{cJ0} := \left\{  F[\phi] = \sum_{n=1}^N \int_{M^n} f_n(x_1, \dots, x_n) \phi(x_1)  \dots P_\bphi \phi(x_n) \vol(x_1) \dots \vol(x_n) \right\} \subset \textbf{W}_{\bphi}
\eeq
 of functionals $F$ that vanish on all solutions $\phi$ of the linearized equations of motion $P_\bphi \phi = 0$.
 
The subspace $\gls{W_loc_phi} \subset \bW_\bphi$ of \emph{local functionals} consists of those $F$ of the form \eqref{eq:def_functional} for which each $f_n$ is supported on the total diagonal of $M^n$. It is generated by smearing \emph{fields} $\cO(x)$ with appropriate test tensors. Fields depend locally and covariantly on $g, \bphi, \phi$, or, abstractly,
\[
\psi^* \cO[g, \bphi, \phi] = \cO[\psi^* g, \psi^* \bphi, \psi^* \phi]
\]
for a morphism $\psi$.
They are of the form
\[
 \gls{cO} [g, \bphi, \phi](x) = P\big( \nabla_{(\alpha)} \phi(x), \nabla_{(\alpha)} \bphi(x), g_{\mu \nu}(x), g^{\mu \nu}(x), \nabla_{(\alpha)} R_{\mu \nu \rho \sigma}(x) \big),
\]
where $P$ is a polynomial, $\alpha$ stands for multi-indices and $R_{\mu \nu \rho \sigma}$ is the Riemannian curvature of $g$. It is sometimes useful to express a local functional in terms of its integral kernel.
 
\subsubsection*{Time-ordered products}
To obtain the interacting renormalized quantum fields, one needs to define \emph{renormalized time-ordered products} (or \emph{renormalization schemes}) on the algebra $\bW_{\bphi}$. These are a collection of symmetric multi-linear maps
\beq
 \label{eq:TOP}
  \gls{T-phi} : (\bW^\loc_\bphi)^{\otimes n} \to \bW_\bphi,
\eeq
which are subject to the axioms (or renormalization conditions) of \cite{Hollands:2001nf, Brunetti:1999jn}, \cf also the reviews \cite{Hollands:2007zg, Hollands:2014eia, MR3469848}. In particular, they fulfill:
\begin{description}
\item[Grading.] Time-ordered products respect the $\Deg$ grading, i.e., 
\begin{equation*}
 \Deg(T_{\bphi, n}(F_1 \otimes \dots \otimes F_n)) = \sum_i \Deg(F_i).
\end{equation*}
\item[Locality and covariance.] Let $\psi: (M', g', \bphi') \to (M, g, \bphi)$ be a morphism, and $\alpha_\psi$ as in \eqref{eq:alpha_psi}. Then
\beq
\label{eq:TOLocalCovariance}
 \alpha_\psi \circ T_{\bphi', n} = T_{\bphi, n} \circ {\alpha_\psi}^{\otimes n}.
\eeq
\item[Scaling.] Each $T_{\bphi, n}$ scales almost homogeneously, \cf \cite{Hollands:2001nf}, under
\beq
\label{eq:Scaling}
 (g_{ab}, \lambda, m, \bphi, \phi) \mapsto (\mu^{-2} g_{ab}, \lambda,\mu m, \mu \bphi, \mu \phi).
\eeq
\item[Causal factorization.] For $\cup_{m=1}^i \supp F_m \cap J^{-}( \cup_{l=i+1}^n \supp F_l) = \emptyset$, it holds
\beq
\label{CFaxiom}
T_{\bphi, n}( F_1 \otimes \dots \otimes F_n) = T_{\bphi, i}( F_1 \otimes \dots \otimes F_i) \star T_{\bphi, n-i}( F_{i+1} \otimes \dots \otimes F_n).
\eeq
\item[Field independence.] Each $T_{\bphi, n}$ is independent of the dynamical field $\phi$, in the sense that
\beq
\label{eq:FieldIndependence}
\tfrac{\delta}{\delta \phi(x)}T_{\bphi, n}( F_1 \otimes \dots \otimes F_n) = \sum_{i=1}^n  T_{\bphi, n}( F_1\otimes \dots \otimes \tfrac{\delta}{\delta \phi(x)} F_i \otimes \dots \otimes F_n).
\eeq
\item[Single field factor.]
A time-ordered product with a single field factor
simplifies as
\begin{multline}
 T_{\bphi, n+1}( \phi(x) \otimes F_1 \otimes \dots \otimes F_n) = \phi(x) \star T_{\bphi, n}(F_1 \otimes \dots \otimes F_n) \\ \label{eq:T10}
  + i\hbar \sum_{j = 1}^n \int \Delta_\bphi^\adv(x,y) T_{\bphi, n}(F_1 \otimes \dots \tfrac{\delta}{\delta \phi(y)} F_j \otimes \dots \otimes F_n).
\end{multline}
\item[Support.] Time-ordered products do not increase the support, i.e.,
\beq
\label{eq:T_support}
 \supp T_{\bphi, n}(F_1 \otimes \dots \otimes F_n) \subset \cup_{i} \supp F_i.
\eeq
\end{description}

For fields, it is more convenient to use the \emph{mass dimension} instead of the \emph{scaling dimension}, defined by the power of $\mu$ in the scaling law \eqref{eq:Scaling}. It is defined as the scaling dimension plus the number of lower indices minus the number of upper indices. It has the advantage that it does not depend on the position of the indices.

As shown in \cite{Brunetti:1999jn, Hollands:2001fb}, time-ordered products exist and are unique up to a well-characterized, local and covariant \emph{renormalization ambiguity} which is described by the \emph{main theorem of renormalization theory}.  
These ambiguities are best expressed in terms of the generating functional for time-ordered products given by
\[
T_{\bphi}(\eti{F}) = \sum_n \frac{\ibar^n}{n!} T_{\bphi,n}(F^{\otimes n}),
\]
where we have introduced the notation
\[
\gls{ibar} \defeq \frac{i}{\hbar} .
\]
In passing, we note that for $F$ a proper interaction, i.e.,  $\deg_\phi(F) \geq 3$, the expression is well-defined \wrt the $\Deg$ grading, i.e.,  at any given grade only a finite number of terms contribute.

Now,  let $T_{\bphi}$ and $T'_{\bphi}$ be two different time-ordered products (renormalization schemes) which satisfy the above axioms. The main theorem of renormalization theory then states that they are related via
\beq
\label{eq:T-T'}
T'_{\bphi}(\eti{F}) = T_{\bphi}(\eti{(F + D(\et{F}))}),
\eeq
with $D(\et{F}) = \sum_{n \geq 1} \frac{1}{n!} D_n(F^{\otimes n})$, where 
\beq
\label{eq:counter-terms}
D_n : (\bW^\loc_\bphi)^{\otimes n} \rightarrow \bW_\bphi^\loc ,
\eeq
correspond to finite local counter terms, characterizing the renormalization ambiguity. They are of order $O(\hbar)$, decrease the total $\Deg$ by $2 (n-1)$, are supported on the total diagonal, i.e., they vanish unless the supports of all arguments overlap, and are locally covariant and field independent, i.e., fulfill \eqref{eq:TOLocalCovariance} and \eqref{eq:FieldIndependence} with $T_n$ replaced by $D_n$.
Furthermore, they scale homogeneously under \eqref{eq:Scaling} and vanish if one of their arguments is a linear field.

The time-ordered products $T_{\bphi, 1}(\cO)$ are usually called \emph{Wick powers} and are constructed by point-splitting \wrt the \emph{Hadamard parametrix} $h$, \cf \cite{dewitt1960radiation, Hollands:2001nf}, which is constructed covariantly from the local geometric data and captures the singularities of Hadamard two-point functions $\omega$, i.e., $\omega - h$ is smooth. Concretely, one defines
\beq
\label{eq:DefWickPower}
 T_1(F)_\omega \defeq \exp(\hbar \tilde \Gamma_{\omega- h}) F, 
\eeq
where
\[
 \tilde \Gamma_f F = \int_{M^2} f(x,y) \tfrac{\delta^2}{\delta \phi(x) \delta \phi(y)} F
\]
and the subscript $\omega$ on the \lhs denotes the two-point function \wrt which the $\star$ product is defined.
Time-ordered products $T_{\bphi, n}$ for $n > 1$ can be constructed recursively using in particular the causal factorization to define the distributions up to the diagonal in $M^n$ and extending them to the diagonal as first proposed by Epstein and Glaser \cite{epstein1973role} (for details see \cite{Hollands:2001nf, Hollands:2001fb, Brunetti:1999jn}).

\subsubsection*{The interacting algebra $\bW_\bphi^\ia$}

Interacting observables are represented in $\bW_\bphi$ via \emph{retarded products}, defined by Bogoliubov's formula
\[
 \gls{R-phi}( \eti{F}; \eti{G}) \defeq T_{\bphi}( \eti{G} )^{-1} \star T_{\bphi}( \eti{F} \otimes \eti{G}).
\]
By causal factorization \eqref{CFaxiom}, retarded products are trivial if the support of second argument does not intersect the past of the support of the first, i.e.,
\[
 R_\bphi ( \eti{F}; \eti{G}) = T_\bphi(\eti{F}) \qquad \supp G \cap J^-(\supp F) = \emptyset.
\]
The generating functional of \emph{interacting time ordered products} is then given by
\[
 \gls{T-int-phi}(\eti{F}) \defeq R_\bphi(\eti{F}; \eti{S_\ia}).
\]
Given a field $\cO$, one thus defines the corresponding \emph{interacting field} as
\beq
\label{eq:InteractingField}
 \cO^\ia_\bphi(x) \defeq T^\ia_\bphi(\cO(x)).
\eeq
As for time-ordered products, interacting time-ordered products fulfil causal factorization, i.e., 
\beq
\label{eq:CausalFactorization_int}
 T^\ia_\bphi(\eti{(F+G)}) = T^\ia_\bphi(\eti{F}) \star T^\ia_\bphi(\eti{G}) \qquad \supp F \cap J^-(\supp G) = \emptyset.
\eeq

The \emph{interacting algebra} $\bW_\bphi^\ia$ is the subalgebra of $\bW_\bphi$ generated by the interacting time ordered products for $\supp F \subset \cR$. The subalgebras $\bW^\ia_\bphi(\cL)$ of observables measurable in compact, causally closed space-time regions $\cL \subset \cR$ are generated by $T^\ia_\bphi(\eti{F})$ with $\supp F \subset \cL$. By \eqref{eq:CausalFactorization_int}, the algebras corresponding to causally disjoint space-time regions commute.

Finally, we also introduce \emph{interacting retarded products} by
\[
 R^\ia_\bphi( \eti{F}; \eti{G}) \defeq T^\ia_{\bphi}( \eti{G} )^{-1} \star T^\ia_{\bphi}( \eti{F} \otimes \eti{G}).
\]
We note that (the equality holds both for usual and interacting time-ordered/retarded products)
\beq
\label{eq:R_1_int}
 R^{(\ia)}_\bphi(\eti{F}; G) = T^{(\ia)}_\bphi(G \otimes \eti{F}) - T^{(\ia)}_\bphi(G) \star T^{(\ia)}(\eti{F}),
\eeq
and that, as a consequence of \eqref{eq:FieldIndependence}, field independence of interacting time-ordered products holds in the sense that
\beq
 \label{eq:T_int_field_indep}
 \tfrac{\delta}{\delta \phi(x)} T^\ia_{\bphi}( \eti{F} ) = \ibar T^\ia_\bphi ( \tfrac{\delta}{\delta \phi(x)} F  \otimes \eti{F}) + \ibar R^\ia_\bphi( \eti{F}; \tfrac{\delta}{\delta \phi(x)} S_\ia).
\eeq

\subsection{Background independence of renormalized scalar field theory} \label{section:BI-scalar}
As discussed in the introduction, the naive derivative $\bar \delta_{\bvp} \defeq \skal{ \tfrac{\delta}{\delta \bphi} -}{\bvp}$
in \eqref{eq:bg_ind_functional} \wrt the background field is not properly defined on the algebra bundle $\bW_{\Phi^4} = \sqcup_{\bphi}\bW_{\bphi}\to \sS_{\Phi^4}$.
The natural replacement is the retarded variation $\delta^\ret_{\bvp}$ defined as follows. Given two backgrounds $\bphi$ and $\bphi'$, one defines the \emph{retarded M{\o}ller operator} \cite{Brennecke:2007uj}, \cf also \cite{Hollands:2001nf} for an on-shell version,
as an algebra isomorphism \cite{Hollands:2001nf, Zahn:2013ywa}
\begin{equation*}
\tau^\ret_{\bphi, \bphi'}: \bW_{\bphi '} \rightarrow \bW_{\bphi},
\end{equation*}
by its action on functionals as
\beq
\label{eq:def_MollerMap}
 (\tau^\ret_{\bphi, \bphi'} F)_{\omega_\bphi} [\phi] \defeq F_{\omega_{\bphi'}}[r_{\bphi', \bphi} \phi].
\eeq
Here $r_{\bphi', \bphi}$ is the \emph{retarded wave operator} 
\beq
\label{eq:retardedWaveOperator}
 r_{\bphi', \bphi} \phi \defeq \phi + \Delta^\ret_{\bphi'} \left( (P_{\bphi} - P_{\bphi'}) \phi \right),
\eeq
mapping solutions of $P_{\bphi} \phi = 0$ to solutions of $P_{\bphi'} \phi = 0$ which coincide outside of $J^+(\supp (\bphi -\bphi'))$.\footnote{It is well-defined as $P_{\bphi} - P_{\bphi'}= - \tfrac{1}{2} \lambda (\bphi^2 - \bar{\phi'}^2 )$ is compactly supported.} In \eqref{eq:def_MollerMap}, the subscript $\omega$ denotes a two-point function \wrt which the $\star$ product on $\bW_{\bphi}$ is defined, and $\omega_{\bphi'}$ is obtained by acting with $r_{\bphi', \bphi}$ on both variables of $\omega_\bphi$. 
Given an infinitesimal background variation $\bvp$, as in \eqref{eq:bg-variation}, and a family $\{ F_s \}_{s \in \R}$ of functionals, $F_s \in \bW_{\bphi_s}$,\footnote{A typical family of such functionals would be given by the assignment $\bphi \mapsto \cO^\ia_\bphi$ of an interacting observable to each background, given a field $\cO$.} one defines the \emph{retarded variation}
\beq
\label{eq:retarded-var-def}
 \gls{delta-ret} F \defeq \del_s  {\left(\tau^\ret_{\bphi, \bphi_s} F_s\right)\Big|}_{s = 0}.
\eeq

A key identity on which our discussion of background independence is based, is the so-called \emph{perturbative agreement} formulated in \cite{hollands2005conservation}. It is derived from the requirement that it should not matter whether one includes terms quadratic in the fields into the free or the interacting part of the action. The comparison between the two theories thus defined is performed by the retarded M{\o}ller operator, or, infinitesimally, by the retarded variation. This implies a further renormalization condition, supplementing those mentioned in the previous section:
\begin{description}
\item[Background variation.] For an infinitesimal variation $\bvp$ of the background $\bphi$, we have
\beq
\label{eq:PPA}
\delta^\ret_{\bvp} T_{\bphi}( \eti{F} ) = \ibar T_{\bphi}( \bar \delta_{\bvp} F \otimes \eti{F}) + \ibar R_{\bphi}( \eti{F} ; \bar \delta_{\bvp}S_0).
\eeq
\end{description}
As shown in \cite{Drago2017, collini2016fedosov}, this condition can indeed be implemented.
In the following, we thus assume that \eqref{eq:PPA} holds. In particular, we then have the following version of perturbative agreement on interacting time-ordered products.
\begin{lemma}
On interacting time-ordered products, perturbative agreement implies
\beq
\label{eq:delta_ret_T_int}
 \delta^\ret_{\bvp} T^\ia_\bphi(\eti{F}) = \ibar T^\ia_{\bphi}( \bar \delta_{\bvp} F \otimes \eti{F}) + \ibar R^\ia_{\bphi}( \eti{F} ; \bar \delta_{\bvp} S).
\eeq
\end{lemma}
\begin{proof}
We compute
\begin{align*}
 \delta^\ret_{\bvp} T^\ia_\bphi(\eti{F}) & = T_\bphi(\eti{S_\ia})^{-1} \star \delta^\ret_{\bvp} T_\bphi(\eti{F} \otimes \eti{S_\ia}) - T_\bphi(\eti{S_\ia})^{-1} \star \delta^\ret_{\bvp} T_\bphi(\eti{S_\ia}) \star T^\ia_\bphi(\eti{F}) \\
 & = \ibar T^\ia_{\bphi}( \bar \delta_{\bvp} (F + S_\ia) \otimes \eti{F}) + \ibar T_\bphi(\eti{S_\ia})^{-1} \star R_{\bphi}( \eti{F} \otimes \eti{S_\ia} ; \bar \delta_{\bvp}S_0) \\
 & \quad - \ibar T^\ia_{\bphi}( \bar \delta_{\bvp} S_\ia ) \star T^\ia_\bphi(\eti{F}) - \ibar T_\bphi(\eti{S_\ia})^{-1} \star R_{\bphi}( \eti{S_\ia} ; \bar \delta_{\bvp} S_0) \star T^\ia_\bphi(\eti{F}).
\end{align*}
The claim then follows from 
\begin{align*}
 & R_{\bphi}( \eti{F} \otimes \eti{S_\ia} ; \bar \delta_{\bvp}S_0) - R_{\bphi}( \eti{S_\ia} ; \bar \delta_{\bvp} S_0) \star T^\ia_\bphi(\eti{F}) \\
 & = T_\bphi(\eti{F} \otimes \bar \delta_{\bvp} S_0 \otimes \eti{S_\ia}) - T_\bphi(\bar \delta_{\bvp} S_0 \otimes \eti{S_\ia}) \star T^\ia_\bphi(\eti{F}) \\
 & = T_\bphi(\eti{S_\ia}) \star R^\ia_\bphi(\eti{F}; \bar \delta_{\bvp} S_0),
\end{align*}
which is a consequence of \eqref{eq:R_1_int}.
\end{proof}

Corresponding to the subalgebras $\bW^\ia_\bphi(\cL)$ for observables localized in the space-time region $\cL$, we may introduce the subbundles $\bW^\ia_{\Phi^4}(\cL)$. The space of sections $\Gamma(\bW^\ia_{\Phi^4})$ of the algebra bundle $\bW^\ia_{\Phi^4}$ is an algebra in itself, with the product being fiber-wise given by $\star$. With a slight abuse of notation, we denote the resulting product again by $\star$. One may define the subalgebra $\Gamma^\infty(\bW^\ia_{\Phi^4})$ of smooth sections, \cf \cite{collini2016fedosov} for details. For our purposes, it is sufficient to think of it as generated by sections \eqref{eq:bphi_to_T_int} for local functionals $F$ with a smooth dependence on the background $\bphi$. Analogously to the usual definition of connections on vector bundles, we give a tentative definition of a connection on the interacting algebra bundle, with a supplementary space-time localization condition, which seems natural in a quantum field theoretical context. 

\begin{definition}
\label{def:Connection}
A \emph{connection} $\fD$ on $\bW^\ia_{\Phi^4}$ is a map
\[
  \Gamma^\infty (T \sS_{\Phi^4}) \times \Gamma^\infty(\bW^\ia_{\Phi^4}) \ni (\bvp, F) \mapsto \fD_{\bvp} F \in \Gamma^\infty(\bW^\ia_{\Phi^4}),
\]
which is $C^\infty(\sS_{\Phi^4})$ linear in the first and additive in the second argument, reduces to the ordinary derivative on c-number functionals, i.e.,
\[
 \fD_{\bvp} F_0 = \bar \delta_{\bvp} F_0, \qquad \forall F_0 \text{ s.t. } [F_0, -]_\star = 0,
\]
is a derivation, i.e., fulfilling
\beq
\label{eq:Leibniz}
 \fD_{\bvp} ( F \star G ) = \fD_{\bvp} F \star G + F \star \fD_{\bvp} G,
\eeq
and respects space-time localization, in the sense that
\beq
\label{eq:ConnectionLocalization}
 \fD_{\bvp} \Gamma^\infty(\bW^\ia_{\Phi^4}(\cL)) \subset \Gamma^\infty(\bW^\ia_{\Phi^4}(\cL)).
\eeq
\end{definition}

By \eqref{eq:delta_ret_T_int}, due to the second term on the r.h.s., the background variation $\delta^\ret_{\bvp}$ violates the locality requirement \eqref{eq:ConnectionLocalization}.\footnote{It is not even obvious that it is well-defined on $\bW^\ia_{\Phi^4}$, i.e., that $\fD_{\bvp} F \in \Gamma^\infty(\bW^\ia_{\Phi^4})$, as $\bar \delta_\bvp S$ is not supported in $\cR$.} But, as seen in the following proposition, subtracting the derivative \wrt $\phi$ yields a connection. The following propositions, first proven in \cite{Hollands:unpublished}, \cf \cite{collini2016fedosov} for details, summarize background independence for scalar fields.

\begin{proposition}\label{nabla-hbar}
The operator 
\beq\label{eq:fD-sclar}
\fD_{\bvp} \defeq  \delta^\ret_{\bvp} - \delta_{\bvp}
\eeq
defines a connection on $\bW^\ia_{\Phi^4}$, acting as
\beq
\label{eq:nablahbar-nabla}
\fD_{\bvp} T^\ia_{\bphi}( \eti{F} ) = \ibar T^\ia_\bphi(\cD_{\bvp} F \otimes \eti{F}),
\eeq
where $\delta_{\bvp}$ and $\cD_{\bvp}$ are defined in \eqref{eq:bg_ind_functional}.
\end{proposition}

\begin{proof}
That $\fD_{\bvp}$ is a derivation is a consequence of the retarded M{\o}ller operator being an algebra isomorphism and of $\delta_{\bvp}$ being a derivation. The localization requirement \eqref{eq:ConnectionLocalization} is a consequence of \eqref{eq:nablahbar-nabla}. To prove the latter, we note that by \eqref{eq:delta_ret_T_int} and \eqref{eq:T_int_field_indep}, we have
\[
 \fD_{\bvp} T^\ia_{\bphi}( \eti{F} ) = \ibar T^\ia_\bphi(\cD_{\bvp} F \otimes \eti{F}) + \ibar R^\ia_\bphi(\eti{F}; \{ \bar \delta_{\bvp} S + \delta_{\bvp} S_\ia \}).
\]
The claim then follows from \eqref{eq:ShiftSymmetryScalar}.
\end{proof}


\begin{proposition} The connection $\fD_{\bvp} $, defined in \eqref{eq:fD-sclar}, is flat.
\end{proposition}
\begin{proof}
It is straightforward to check that $\cD_{\bvp}$ satisfies
\[
[\cD_{\bvp}, \cD_{\bvp'}] - \cD_{\LB{\bvp}{\bvp'}} =0,
\]
where
\[
\gls{Lie-bra-vect} \defeq \skal{\tfrac{\delta}{\delta \bphi} \bvp'}{ \bvp} - \skal{\tfrac{\delta}{\delta \bphi} \bvp}{\bvp'}
\]
is the Lie bracket of vector fields $\bvp$ and $\bvp'$ on $\mathcal{S}_{\Phi^4}$.
Therefore, using \eqref{eq:nablahbar-nabla}, the curvature of $\fD_{\bvp}$ vanishes:
\[
([\fD_{\bvp}, \fD_{\bvp'}] -  \fD_{\LB{\bvp}{\bvp'}} ) T^\ia_\bphi(\eti{F}) = \ibar T^\ia_\bphi ( ([\cD_{\bvp}, \cD_{\bvp'}] - \cD_{\LB{\bvp}{\bvp'}}) F \otimes \eti{F} ) = 0.
\]
\end{proof}

Hence, defining the background independent observables as sections which are covariantly constant \wrt $\fD_{\bvp}$, \eqref{eq:nablahbar-nabla} implies that background independent interacting fields $\cO^\ia_\bphi$ correspond to classically split independent fields $\cO$, i.e., fulfilling $\cD_{\bvp} \cO=0$. This means that there is a one-to-one correspondence between classical and quantum background independent fields.

\section{Pure Yang-Mills theory} \label{section:YM}

This main part of the article is structured as follows: We begin by setting up Yang-Mills theory on the classical level, culminating in the identification of $\hat \cD_\ba$ as the relevant connection on classical local functionals. In Section~\ref{sec:QuantumYM}, we discuss, following \cite{Hollands:2007zg, Tehrani:2017pdx}, quantization, in particular the occurrence of anomalies. As a crucial ingredient for background independence, we prove a theorem on the background dependence of the anomaly, assuming that perturbative agreement holds. In Section~\ref{subsection:back-ind} we then prove our main result on background independence in Yang-Mills theories.

\subsection{Classical gauge theory}\label{section:YM-classical}

\subsubsection{The basic setting}

Let $P \rightarrow M$ be a $G$ principal fibre bundle over space-time $M$, with $G$ a semi-simple Lie group.
We denote by $\Ad$ the adjoint action of a Lie group $G$ on itself, $\Ad_g h \defeq g h g^{-1}$, and the adjoint action on the corresponding Lie algebra $\g$ by $\ad$. The Lie bracket on $\g$ is denoted by $\gls{Lie-bra-g}$.

The Yang-Mills theory is the dynamical theory of a $G$ connection $\cA$ on $P$ whose dynamics is governed by the Yang-Mills action
\begin{equation}\label{eq:YM-action}
  \int_M \Tr (F \wedge * F),
\end{equation}
where $F$ is the curvature of $\cA$, interpreted as a section of $\p \otimes \Omega^2$, with $\gls{p-bundle} \defeq P \times_{\ad} \g$ and $\gls{Omega}$ the bundle of $k$ forms on $M$. Let $\{T_I\}_I$, be a basis of $\g$, normalized as $\Tr (T_I T_J) = - \frac{1}{2} \delta_{IJ}$. Then, we can write $F = \frac{1}{2} F^I_{\mu \nu} T_I \ud x^\mu \wedge \ud x^\nu$.

Classical solutions will play the role of background configurations, and these will be typically denoted by a bar, i.e.,  we will consider connections $\bar \cA$ which are solutions to the Yang-Mills equation
\beq
\label{eq:YM_eom}
 \bar \nabla_\mu \bar F^{\mu \nu} = 0,
\eeq
where $\gls{barF}$ is the curvature of $\bar{\cA}$ and  $\gls{bar-nabla}$ is the associated covariant derivative on sections of $\p \otimes \Omega^k$.
The Yang-Mills equation is well-posed \cite{ChruscielShatah}, guaranteeing the existence of global solutions. Furthermore, the set $\gls{S-YM}$ of such solutions is a manifold, i.e.,  its tangent space $T_{\bar \cA} \sS_{\YM}$ at a solution $\bar \cA$ is the space of solutions $\gls{delta-A}$ to the Yang-Mills equation linearized around $\bar \cA$,
\beq
\label{eq:lin-YM-eq}
 \gls{P-lin-A} \ba_\mu^I \defeq \bar \nabla^\nu  \left( \bar \nabla_\nu \ba_\mu^I - \bar \nabla_\mu \ba_\nu^I \right) + [\bar F_{\mu \nu}, \ba^\nu]_\g^I =0,
\eeq
except at certain symmetrical background configurations $\bar \cA$, \cf \cite{arms_1981}. At these symmetrical background configurations, there are solutions to \eqref{eq:lin-YM-eq} that are not tangent to $\sS_{\YM}$, i.e.,  do not arise as the derivative of a curve in $\sS_{\YM}$. The presence of such singular points in configuration space $\sS_{\YM}$ does not impart our considerations, as these are local in $\sS_{\YM}$, so that we can restrict to regions not containing such exceptional points. Thus, we will henceforth identify the space of solutions to \eqref{eq:lin-YM-eq} with the tangent space $T_{\bar \cA} \sS_\YM$ of $\sS_\YM$ at $\bar \cA$.

\subsubsection*{Background and dynamical gauge transformations }
We consider the decomposition \eqref{eq: A=barA+A} of $\cA$ into a \emph{background connection} $\bar{\cA}$ and a \emph{dynamical} $\mathfrak{g}$-valued one-form $A$, i.e., a section of $\p \otimes \Omega^1$.
In local coordinates, the corresponding covariant derivative operator $D$ when acting on sections of $\p \otimes \Omega^k$ takes the form
\[
D_{\mu} = \bar{\nabla}_\mu + [A_\mu, -]_\g.
\] 
Then, the curvature two form $F$ in local coordinates is given by
\beq
\label{eq:full-curvature}
F^I_{\mu \nu} = \bar{F}^I_{\mu \nu} + \bar \nabla_\mu A^I_\nu - \bar \nabla_\nu A^I_\mu +  [A_\mu, A_\nu]_\g^I.
\eeq

Gauge transformations are parametrized by smooth sections $g$ of $P \times_{\Ad} G$. On a connection $\cA$, they act as
\[
 \cA \mapsto \cA^g \defeq \ad_{g^{-1}} \circ \cA + g^* \theta,
\]
with $\theta$ the Maurer-Cartan form. For $\cA$ split as in \eqref{eq: A=barA+A}, there are then two natural implementations of this gauge transformation.
A \emph{background gauge transformation} acts as
\[
 \bar \cA \mapsto \bar \cA^g, \qquad A \mapsto \ad_{g^{-1}} A.
\]
The covariance of the quantum theory under such a transformation will be part of the requirement of local (gauge) covariance. On the other hand, one may keep the background fixed and implement the change $\cA \mapsto \cA^g$ by solely changing $A$, i.e., 
\[
\bar{ \cA} \mapsto \bar{\cA}, \qquad A \mapsto \bar \cA^g - \bar \cA + \ad_{g^{-1}} A.
\]
This is called a \emph{dynamical gauge transformation} which needs to be gauge-fixed.

\subsubsection*{Localization of the interaction and split independence}
As for the scalar field, we need to localize the interaction in a compact space-time region. For the scalar field, we used a smooth compactly supported cutoff function $\lambda(x)$ which was equal $\lambda_0$ in the space-time region $\cR$ for which the algebra of interacting observables was constructed. This cut-off had the additional consequence that, for any two background solutions $\bphi$, $\bphi'$, the corresponding linearized wave operators $P_\bphi$, $P_{\bphi'}$, \cf \eqref{eq:bphi-lin-eom}, coincided outside of a compact space-time region (the support of~$\lambda$). This made it possible to define the flat connection $\fD_{\bvp}$, using the retarded variation $\delta^\ret$. 

Also for Yang-Mills theory, we use a smooth cutoff function $\lambda(x)$ to localize the interaction (see below). 
This cut-off, however, does not affect the linearized wave operator $P^\lin_{\bar \cA}$, defined in \eqref{eq:lin-YM-eq}.\footnote{Introducing a cut-off there would spoil gauge covariance.} Hence, the operators $P^\lin_{\bar \cA}$ in general do not coincide outside of a compact space-time region, which, however, is a prerequisite for the use of the retarded variation. Hence, we relax the condition that $\bar \cA$ is on-shell, i.e., a solution to \eqref{eq:YM_eom}, on the whole space-time. We proceed as follows: We choose a neighborhood $\cU$ of $\cR$ on which we require the backgrounds $\bar \cA$ to be on-shell, i.e.,
\beq
\label{eq:bg-eom}
\bar \nabla^\mu \bar F_{\mu \nu}(x) = 0, \qquad x \in \cU.
\eeq
Furthermore, we require all backgrounds $\bar \cA$ to coincide outside of a larger region $\cV \supset \cU$ with an arbitrary reference connection $\cA_0$. Consequently, the variations $\ba$ of the background are supported in $\cV$ and fulfil the linearized Yang-Mills equation \eqref{eq:lin-YM-eq} in $\cU$. In this way, one ensures that the retarded variation $\delta^\ret_\ba$ is well-defined.

Furthermore, one localizes the interaction by introducing a cut-off function $\lambda$, which is supposed to be supported in $\cU$ and equal to $1$ on a neighborhood of $\cR$.
The action is, then, defined as
\begin{multline}
\label{eq:YM_Action}
 S_{\YM} = - \frac{1}{4} \int \Big\{ \left( \bar \nabla_\mu A_\nu - \bar \nabla_\nu A_\mu + \lambda [A_\mu, A_\nu]_\g \right)^I \left( \bar \nabla^\mu A^\nu - \bar \nabla^\nu A^\mu + \lambda [A^\mu, A^\nu]_\g \right)^I  \\
   + 2 \bar F^I_{\mu \nu} [A^\mu, A^\nu]_\g^I \Big\} \vol,
\end{multline}
where summation over repeated indices $I$ is understood. In $\cR$, where $\lambda=1$ and the background $\bar{\cA}$ is on-shell, this is the Yang-Mills action \eqref{eq:YM-action} expanded around $\bar{\cA}$, with the constant term $-\frac{1}{4} \int \bar{F}^I_{\mu \nu} \bar F^{I \mu \nu} \vol$ omitted. Note that the full Yang-Mills action \eqref{eq:YM-action} would have a source term, i.e.,  a term linear in $A$, which however vanishes in $\cR$, as $\bar \cA$ is on-shell there.
The set-up of our localization prescription is summarized in Figure~\ref{fig:Localization}.

\begin{figure}
\begin{center}
\begin{tikzpicture}

\filldraw[color=black, fill=gray!25!white] (3,1.5) -- (9,1.5) -- (9, 2) -- (3, 2)  -- cycle;
\filldraw[color=black, fill=gray!45!white] (9, 1.5) -- (3, 1.5)  -- (3, 0.5) -- (9, 0.5) -- cycle;
\filldraw[color=black, fill=gray!70!white] (3,0.5) -- (9,0.5) -- (9,-0.5) -- (3, -0.5)  -- cycle;
\filldraw[color=black, fill=gray!45!white] (9,-0.5) -- (3, -0.5) -- (3, -1.5) -- (9, -1.5) -- cycle;
\filldraw[color=black, fill=gray!25!white] (3,-1.5) -- (9,-1.5) -- (9, -2) -- (3, -2) -- cycle;
\filldraw[black] (6, 2.30)  node {$\bar \cA= \cA_0$}; 
\filldraw[black] (6,1.75)  node {$\cV: \lambda=0$}; 
\filldraw[black] (6,1)  node {$\cU: \bar \nabla^\mu \bar F_{\mu \nu} = 0$}; 
\filldraw[black] (6,0)  node {$\cR: \lambda = 1, \bar \nabla^\mu \bar F_{\mu \nu} = 0 $}; 
\filldraw[black] (6,-1)  node {$\cU: \bar \nabla^\mu \bar F_{\mu \nu} = 0$}; 
\filldraw[black] (6,-1.75)  node {$\cV: \lambda=0$}; 
\filldraw[black] (6,-2.30)  node {$\bar \cA= \cA_0$}; 
\end{tikzpicture}
\caption{\label{fig:Localization} Different regions $\cR \subset \cU \subset \cV$ in our localization set-up. }
\end{center}
\end{figure}

Since $\bar{\nabla}$, $\bar{F}$ and $A$ transform covariantly under background gauge transformations, the action \eqref{eq:YM_Action} is invariant under background gauge transformations.
Analogously to \eqref{eq:ShiftSymmetryScalar}, the action is split independent in the sense that
\beq
\label{eq:ShiftSymmetryYM}
 \frac{\delta S_\YM}{\delta \bar \cA(x)} = \frac{\delta S_{\YM, \ia}}{\delta A(x)} \qquad x \in \cR.
\eeq
Here $S_{\YM, \ia}$ is the part of $S_\YM$ which is of degree higher than $2$ in $A$.
The restriction to $x \in \cR$ is due to the infra-red cut-off $\lambda$ of the interaction.

\subsubsection{BV-BRST formalism and background covariant gauge-fixing}
\label{sec:BRST}
In this section, we outline the straightforward generalization of the BV-BRST formalism \cite{Hollands:2007zg, becchi1976renormalization, BATALIN198127, Fredenhagen:2011mq}, to the case with non-trivial backgrounds. 
 
In order to perform gauge-fixing in the BV-BRST formalism, we need to augment the field variables with a set of ghosts and anti-fields, some of which are fermions, i.e., have an odd Grassmann parity.\footnote{The description of fermionic fields in terms of functionals, i.e., the fermionic generalization of \eqref{eq:def_functional}, is described in \cite{Rejzner:2011au}.}
The resulting gauge-fixed theory enjoys the BV-BRST symmetry $s$ as follows. 
Let us denote the set of all dynamical fields by $\Phi=(A_\mu^I, B^I, C^I, \bar{C}^I)$, where $C$ ($\bar C$) are called \emph{(anti-) ghosts} and $B$ is a Lagrange multiplier. One assigns mass dimensions $d_\Phi=(1, 2, 0, 2)$ and a \emph{ghost number} $g_\Phi=(0, 0, 1, -1)$ to the fields. 
The latter defines the Grassmann parity.
The BV-BRST operator $\gls{s}$, which increases the ghost number by 1, acts by
\[
{s} A_\mu^I = \bar \nabla_\mu C^I +  \lambda [A_\mu, C]_\g^I, \quad {s} C^I = - \tfrac{1}{2} \lambda [C, C]_\g^I, \quad {s} \bar{C}^I  = B^I,\quad {s} B^I  = 0.
\]
One also introduces \emph{anti-fields} $\Phi^\ddag = (A^{\ddag I \mu }, B^{\ddag I} , C^{\ddag I}, \bar{C}^{\ddag I})$, with mass dimensions $d_{\Phi^\ddag} = (3, 2, 4, 2)$ and ghost numbers $g_{\Phi^\ddag} = (-1, -1, -2, 0)$. They are interpreted as densities and act as classical, non-dynamical sources of BRST transformations of the fields, appearing in the action via
\[
S_{\text{sc}}= - \int \sum_{i} {s} \Phi^i \Phi_i^\ddag.
\]

To perform the gauge-fixing, we add a manifestly BV-BRST-invariant term $s \Psi$ to the action, where $\Psi$ is a \emph{gauge-fixing fermion} with ghost number $-1$ which does not contain anti-fields and we choose here to be
\beq
\label{eq:GaugeFixingFermion}
\gls{Psi} = \int \bar C^I \left( \bar \nabla^\mu A_\mu^I + \tfrac{1}{2} B^I \right) \vol.
\eeq
This is the so-called \emph{background covariant gauge-fixing}. It breaks dynamical gauge-invariance, while keeping the background gauge-invariance. In this respect, \eqref{eq:GaugeFixingFermion} is a useful gauge in practical calculations and is commonly employed in the background field formalism \cite{Abbott:1981ke, KlubergStern:1974xv,DeWitt:1967ub, tHooft:1975uxh, HONERKAMP1972269, Boulware:1980av}. 

The BV-BRST transformations of all fields and anti-fields can now be written as
\beq
\label{eq:s_as_antibracket}
s = (S, -),
\eeq
where $S$ is the extended and gauge-fixed action 
\beq \label{eq:gauge-fixed-action}
\gls{S-action}= S_{\YM} + S_{\rsc} + s \Psi,
\eeq 
and where $(-, -)$ is the so-called \emph{anti-bracket} defined by
\begin{equation*}
\gls{(F1,F2)} \defeq \int \left\{ \frac{\delta^R F_1}{\delta \Phi^i(x)}  \frac{\delta^L F_2}{\delta \Phi^\ddag_i(x)} -  \frac{\delta^R F_1}{\delta \Phi^\ddag_i(x)} \frac{\delta^L F_2}{\delta \Phi^i(x)} \right\},
\end{equation*}
\cf \cite{Rejzner:2011au} for a definition of left and right derivatives \wrt fields with Grassmann parity. In the following, field derivatives will be left derivatives, unless states otherwise.
The anti-bracket satisfies the graded Jacobi identity
\begin{align}\nn
 0 & = (-1)^{(\eps_1+1) (\eps_3+1)}  (F_1, ( F_2 , F_3 )) + (-1)^{(\eps_2+1) (\eps_1+1)}  (F_2, ( F_3 , F_1 )) \\ \label{eq:classical-Jacobi-id}
& +(-1)^{(\eps_3+1) (\eps_2+1)}  (F_3, ( F_1 , F_2 )),
\end{align}
and has the following graded symmetry
\begin{equation*}
( F_1 , F_2 ) = (-1)^{(\eps_1 +1)( \eps_2+1)+1} ( F_2 , F_1 ).
\end{equation*}
We remark that only on functionals supported in $\cR$, where $\bar{\cA}$ is on-shell and $\lambda = 1$, the operator $s$ coincides with the standard nilpotent BV-BRST differential and the gauge-fixed action fulfills the \emph{classical master equation},
\[
 ({S}, {S}) = 0,
\]
which expresses the BRST-invariance of $S$.

As usual, we split the action into a free and an interaction part,
\[
 S= S_0 + {S_\ia},
\]
where the free action $S_0$ is quadratic in $\Phi$ and $\Phi^\ddag$, and the compactly supported interaction ${S_\ia}$ contains the terms of degree higher than 2 in $\Phi$ and $\Phi^\ddag$. This, in turn, leads to the decomposition 
\[
{s} = {s}_0 + {s}_{\ia}
\]
of the BV-BRST differential.
The action of $s_{0}$ on all fields and anti-fields is given in Table \ref{table1}. Note that the requirement of the background connection being on-shell is necessary for the nilpotency of $s_0$. For instance, one can check by direct calculation that ${s}^2_0 A^{\ddag I}_\mu = [\bar{\nabla}^\nu \bar{F}_{\mu \nu}, C]_\g^I$, which vanishes only if $\bar{\nabla}^\nu \bar{F}_{\mu \nu}=0$. Hence, $s_0$ is only nilpotent when restricted to functionals localized in $\cU$, motivating our condition that $\supp \lambda \subset \cU$.

\begin{table}[h]
\begin{center} 
    \begin{tabular}{ | l | l | l | p{5cm} |}
    \hline
    ${s}_0 \Phi $ & ${s}_0 \Phi^{\ddag}$  \\ \hline \hline
    $s_0 A_\mu^I = \bar \nabla_\mu C^I$ & ${s}_0 A^{\ddag I}_\mu = \big( ( \bar P^\lin A)^I_\mu - \bar{\nabla}_\mu B^I\big) \vol$  \\ \hline
    ${s}_0 B^I  = 0$ & ${s}_0 B^{\ddag I}  = (B^I + \bar{\nabla}^\mu A^I_\mu) \vol - \bar{C}^{\ddag I}$  \\ \hline
    ${s}_0 C^I = 0$ & ${s}_0 C^{\ddag I} = - \bar{\nabla}^\rho \bar{\nabla}_\rho \bar{C}^I \vol - \bar{\nabla}^\mu A_\mu^\ddag$  \\ \hline
    ${s}_0 \bar{C}^I  = B^I$ & ${s}_0 \bar{C}^{\ddag I}  = \bar{\nabla}^\rho \bar{\nabla}_\rho {C}^I \vol $ \\ \hline
    \end{tabular}
    \caption{\label{table1} Free BRST transformations of fields $\Phi $ and anti-fields $\Phi^\ddag $.}
\end{center}
\end{table}

The gauge-fixed action $S$ is invariant under background gauge transformations since all the dynamical fields and anti-fields transform in the adjoint. However, it is no longer split independent, not even in $\cR$, since $\Psi$ destroys split independence as $\bar \cA$ and $A$ no longer appear in $\Psi$ in the form $\bar \cA + A$. 

\begin{proposition}
\label{prop:cD-S}
The gauge-fixed action \eqref{eq:gauge-fixed-action} satisfies
\beq
\label{eq:ShiftSymmetryYM_gf}
 \frac{\delta S}{\delta \bar \cA(x)} - \frac{\delta {S_\ia}}{\delta A(x)} = s \frac{\delta}{\delta \bar \cA(x)} \Psi \qquad x \in \cR.
\eeq
\end{proposition}
\begin{proof} For $x \in \cR$, we calculate
\begin{align*}
 \tfrac{\delta}{\delta \bar \cA(x)} S - \tfrac{\delta}{\delta A(x)} {S_\ia} & = \tfrac{\delta}{\delta \bar \cA(x)} (S_\YM + S_\rsc) - \tfrac{\delta}{\delta A(x)} (S_\YM + S_\rsc)_\ia + \tfrac{\delta}{\delta \bar \cA(x)} (S_\rsc, \Psi) - \tfrac{\delta}{\delta A(x)} (S_{\rsc, \ia}, \Psi) \\
 & = ( \tfrac{\delta}{\delta \bar \cA(x)} S_\rsc - \tfrac{\delta}{\delta A(x)} S_{\rsc, \ia}, \Psi) + (S_\rsc, \tfrac{\delta}{\delta \bar \cA(x)} \Psi) - (S_{\rsc, \ia}, \tfrac{\delta}{\delta A(x)} \Psi) \\
 & = s \tfrac{\delta}{\delta \bar \cA(x)} \Psi
\end{align*}
where we have used that \eqref{eq:ShiftSymmetryYM} also holds with $S_\YM$ replaced by $S_\rsc$ and that $\frac{\delta}{\delta A} \Psi$ is proportional to $\bar C$, on which $s_\ia$ vanishes.
\end{proof}

It is advantageous to also compute the action of $\cD_{\ba}$ on $S$, the former being defined, analogously to \eqref{eq:bg_ind_functional}, by
\[
\gls{cD-A} \defeq ( \delta_{\ba} - \delta_{\ba} ) \defeq  \skal{ ( \tfrac{\delta}{\delta \bar{\cA}} - \tfrac{\delta}{\delta A} ) -}{\ba} .
\]
\begin{corollary}
In $\cR$, i.e., when restricted to configurations supported in $\cR$, we have
\beq
\label{eq:cD_S}
 \cD_{\ba} S = s \cD_{\ba} \Psi.
\eeq
\end{corollary}
\begin{proof}
Using \eqref{eq:ShiftSymmetryYM_gf}, we compute
\[
 \cD_{\ba} S = s \bar \delta_{\ba} \Psi - \delta_{\ba} {S_0} = s \bar \delta_{\ba} \Psi - \delta_{\ba} S_{\YM, 0} -  \delta_{\ba} s_0 \Psi.
\]
The second term on the \rhs vanishes due to $\ba$ being, in $\cR$, a solution to the linearized equation of motion. The result then follows from $\delta_{\ba} s_0 \Psi = s \delta_{\ba} \Psi$, which holds for any $\Psi$ which is quadratic in fields and does not contain anti-fields.
\end{proof}

\subsubsection*{Local gauge covariance}
Our background data now consists of $(P \to M, g, \bar{\cA})$, i.e., a principal fibre bundle $P \to M$ with a fixed structure group $G$, the metric $g$, and a background connection $\bar{\cA}$ on $P$. To make the notion of local covariance precise, we define, following \cite{Zahn:2012dz}, morphisms $\chi: (P' \to M', g', \bar \cA') \to (P \to M, g, \bar \cA)$ as $G$ equivariant smooth maps $\chi: P' \to P$, which cover a causality preserving isometric embedding $\psi: M' \to M$, i.e., a morphism in the sense of the previous section, such that $\chi^* \bar{\cA} =\bar{\cA}'$. This covers the case of background gauge transformations where $\chi_g : P \to P$ is the natural action of a section $g$ of $P \times_{\Ad} G$ on $P$.
Locally covariant fields should then satisfy
\begin{equation}
\label{eq:local-gauge-cov}
\chi^* \mathcal{O} [g, \bar{\cA}, \Phi, \Phi^\ddag] = \mathcal{O}[\psi^*g, \chi^*\bar{\cA}, \chi^* \Phi, \chi^* \Phi^\ddag].
\end{equation}
By the Thomas replacement theorem \cite{Iyer:1994ys, Hollands:2007zg} such a field takes the form
\[
 \cO[g, \bar{\cA}, \Phi, \Phi^\ddag](x) = P\big( \bar{\nabla}_{(\alpha)} \Phi(x), \bar{\nabla}_{(\alpha)} \Phi^\ddag(x), g_{\mu \nu}(x), g^{\mu \nu}(x), \nabla_{(\alpha)} R_{\mu \nu \rho \sigma}(x), \bar{\nabla}^{(\alpha)} \bar{F}_{\mu \nu}(x)\big),
\]
where $P$ is a polynomial, $\alpha$ stands for multi-indices, $R_{\mu \nu \rho \sigma}$ is the Riemannian curvature of $g$, and $\bar{F}_{\mu \nu}$ is the curvature of $\bar{\cA}$.

\subsubsection*{Classical BV-BRST cohomology}

For the case of pure Yang-Mills theory, for semi-simple $G$, the cohomology ring $H(s)$ is generated by elements of the form
\beq
\label{eq: s-cohomology}
{\displaystyle \prod_{k}} r_{t_k}(g,  \bar{\nabla}^{(\alpha)} \bar{F}, \nabla^{(\alpha)} R) {\displaystyle \prod_{i}} p_{r_i}(C) {\displaystyle \prod_{j}} \Theta_{r_j}({D}^{(\alpha)} F),
\eeq
where $\alpha$ stands for multi-indices, $p_r$ and $\Theta_s$ are invariant polynomials of $\g$, $r_t$ is a local functional of the metric $g$, the background field strength $\bar{F}$, the Riemann tensor $R$ and their derivatives. $F$ is the full field strength, \cf \eqref{eq:full-curvature}.
This result for the case of trivial backgrounds, i.e., with $\bar{F}=0$, is proven in \cite{Hollands:2007zg, Barnich:2000zw}. The above expression is then obtained by the requirement of local covariance \eqref{eq:local-gauge-cov} in the presence of a non-trivial background connection. As there is no invariant polynomial of degree 1 on a semi-simple Lie algebra, the cohomology at ghost number 1, $H_1(s)$, is trivial.

Now restricting to sections of vector bundles associated with $P$ via the trivial representation of $G$, that is, those $\cO$ without a Lie algebra index, the cohomology ring  $H(s| \ud)$ is generated by linear combination of elements of the form \eqref{eq: s-cohomology} and elements of the form
\begin{equation}\label{eq: s|d-cohomology}
{\displaystyle \prod_{k}} r_{t_k}(g, \bar{\nabla}^{(\alpha)} \bar{F}, \nabla^{(\alpha)} R) {\displaystyle \prod_{i}} q_{r_i}(\bar F, C+A, A) {\displaystyle \prod_{j}}f_{s_j}(F),
\end{equation}
where
\[
 q_{r}(\bar F, C+A, A) = \int_0^1 \Tr \left( (A+C) \left[ \bar F + t ( \bar \ud A + A^2 ) + (t^2 - t) (A+C)^2 \right]^{m(r)-1} \right) \ud t
\]
are the Chern-Simons forms in the presence of a background connection \cite{ManesStoraZumino}. In this expression, $\bar \ud$ denotes the covariant differential, induced on sections of $\p \otimes \Omega$ by the Leibniz rule and $\bar \ud b = \bar \nabla_\mu b \ud x^\mu$ for $b$ a section of $\p$, and $m(r)$ are the degrees of the independent Casimir elements of $G$. The trace is in some representation of $\g$. Furthermore, $f_s$  are strictly gauge invariant monomials of $F$, and $r_t$ are closed forms. Again, the result \eqref{eq: s|d-cohomology} is a generalization of the well-known results in \cite{Hollands:2007zg, Barnich:2000zw} to the case with non-trivial background connection.

Elements of the cohomology class $H_0(s)$ at ghost number $0$ are in one-to-one correspondence with the gauge invariant observables of the original Yang-Mills theory, while those in the class $H_1^4(s | \ud)$ of $4$ forms at ghost number $1$ turn out to contain the gauge anomalies of the Yang-Mills theory, see e.g. \cite{Hollands:2007zg}.
\subsubsection*{The BRST charge}

Classically, the action of the BRST differential on fields is also generated by the Noether charge of the BRST symmetry via the graded Peierls bracket \cite{Peierls:1952cb, DeWitt:2004xz} $\{ -, -\}_{\bar \cA}$, i.e.,
\[
{s} = \{Q, -\}_{\bar{\cA}}.
\] 
The charge $Q$ is constructed as follows \cite{Hollands:2007zg}: One chooses a one-form $\gamma_\mu$, supported in $\cR$, such that
\[
 \int \gamma \wedge \alpha = \int_\Sigma \alpha,
\]
for a Cauchy surface $\Sigma$ contained in $\cR$ and any closed three form $\alpha$. One then sets
\beq
\label{eq:Q}
 Q = \int \gamma \wedge J,
\eeq
where $J$ is the Noether current of the BRST symmetry, which is a 3 form with ghost number $1$, and is conserved on-shell in $\cR$.

\subsubsection{Background independent local functionals} \label{section:class-BI-obs}
In the case of scalar field theory, we defined the background independent classical local functionals as those in the kernel of $ \cD_{\bvp} $, \cf \eqref{eq:bg_ind_functional}. However, as discussed above, the gauge invariant observables are defined to be equivalence classes of the BV-BRST cohomology. Therefore, the suitable operator whose kernel defines the background-independent classical local functionals, must be well-defined on BV-BRST cohomology (i.e., it must commute with $s$). 
However, in view of \eqref{eq:cD_S}, this is not the case for $\cD_{\ba}$.
We, therefore, define the following modified operator
\beq
\label{eq:tilde-nabla-def}
\gls{tilde-cD-A} \defeq \cD_{\ba} - (-,  \cD_{\ba} \Psi),
\eeq
which turns out to have the desired properties, as stated in the following theorem.
\begin{theorem} \label{thm:tilde-nabla}
The operator $\hat{\cD}_{\ba}$ defined in \eqref{eq:tilde-nabla-def}, satisfies, for $F_i$ with arbitrary support and $F$ supported in $\cR$,
\begin{align}\label{eq:nabla-tilde-anti-bracket}
 \hat{\cD}_{\ba} (F_1, F_2) &= (\hat{\cD}_{\ba} F_1, F_2)  +  (F_1, \hat{\cD}_{\ba} F_2), \\
\label{eq:[tilde-nabla,s]=0}
\left( \hat{\cD}_{\ba} \circ {s}  - {s} \circ \hat{\cD}_{\ba} \right) F &=0,  \\
\label{eq:tilde-nabla-curvature}
\left( [\hat{\cD}_{\ba} , \hat{\cD}_{\ba'} ] - \hat{\cD}_{\LB{\ba}{\ba'}} \right) F & = 0.
\end{align}
\end{theorem}
\begin{proof}
To prove \eqref{eq:nabla-tilde-anti-bracket}, we calculate
\begin{align*}
 \hat{\cD}_{\ba} (F_1, F_2) & =  \cD_{\ba} (F_1, F_2) - ( (F_1, F_2) ,  \cD_{\ba} \Psi)  \\
  &= (\cD_{\ba} F_1, F_2)  +  (F_1, \cD_{\ba} F_2)   - ( (F_1,\cD_{\ba} \Psi), F_2 ) - ( (F_1,  ( F_2, \cD_{\ba} \Psi)  )\\
  & = (\hat{\cD}_{\ba} F_1, F_2)  +  (F_1, \hat{\cD}_{\ba} F_2),
\end{align*}
where we have used the identity
\[
\cD_{\ba} (F_1, F_2) = (\cD_{\ba} F_1, F_2)  +  (F_1, \cD_{\ba} F_2), 
\]
and the Jacobi identity \eqref{eq:classical-Jacobi-id} for the anti-bracket.
To prove \eqref{eq:[tilde-nabla,s]=0}, we compute
\[
 \hat{\cD}_{\ba} (s F) = \hat{\cD}_{\ba} (S, F) = ( \hat{\cD}_{\ba} S, F) + ( S,  \hat{\cD}_{\ba} F) = s \hat{\cD}_{\ba} F,
\]
where we have used \eqref{eq:nabla-tilde-anti-bracket} and \eqref{eq:cD_S}.
To prove \eqref{eq:tilde-nabla-curvature}, we calculate
\[
\hat{\cD}_{\ba}  \hat{\cD}_{\ba'} F =   \cD_{\ba} \cD_{\ba'}  F  - (\cD_{\ba'} F, \cD_{\ba} \Psi)  - (\cD_{\ba}F, \cD_{\ba'} \Psi)  - (F, \cD_{\ba} \cD_{\ba'}  \Psi )+ ( (F, \cD_{\ba'} \Psi), \cD_{\ba} \Psi) 
\]
Therefore, we find
\begin{align*}
([\hat{\cD}_{\ba} , \hat{\cD}_{\ba'} ] - \hat{\cD}_{\LB{\ba}{\ba'}}) F & =   ( [\cD_{\ba} , \cD_{\ba'} ] - \cD_{\LB{\ba}{\ba'}} ) F -  ( F,  \{[\cD_{\ba}, \cD_{\ba'} ]  - \cD_{\LB{\ba}{\ba'}} \} \Psi )  \\
& \quad + ( (F, \cD_{\ba'} \Psi), \cD_{\ba} \Psi)  - ( (F, \cD_{\ba} \Psi), \cD_{\ba'} \Psi)  \\
& =  (F, (\cD_{\ba'} \Psi,  \cD_{\ba} \Psi)),
\end{align*}
where we have used
\[
 [\cD_{\ba} , \cD_{\ba'} ] - \cD_{\LB{\ba}{\ba'}}=0,
\]
and the Jacobi identity \eqref{eq:classical-Jacobi-id}. However, since $\Psi$ does not contain anti-fields, $ (\cD_{\ba'} \Psi,  \cD_{\ba} \Psi)=0$ and thus the curvature of $\hat{\cD}_{\ba}$ vanishes.
\end{proof}
\begin{remark}
The ``correction term'' $(-, \cD_{\ba} \Psi)$ in \eqref{eq:tilde-nabla-def} can also be motivated as follows.
Before introducing the gauge-fixing $\Psi$ in the action \eqref{eq:gauge-fixed-action}, the BV-BRST differential is given by $(S_{\YM} + S_{\rsc}, -)$ which is related to the gauge-fixed differential $s$ by
\beq
\label{eq:(hat-S, O)}
{s} = e^{(-, \Psi)} \circ (S_{\YM} + S_{\rsc}, -) \circ e^{-(-, \Psi)},
\eeq
where
\begin{equation*}
e^{(-,\Psi)} =  \id +  (-, \Psi) + \tfrac{1}{2!} \big( (-, \Psi), \Psi \big) + \tfrac{1}{3!} \big(((- ,\Psi), \Psi), \Psi \big) + \dots,
\end{equation*}
is a ``canonical transformation'' generated by $\Psi$ (in the cases of interest here, Yang-Mills theory and gravity, the series truncates, as $\Psi$ does not contain anti-fields). Consequently, the cohomologies of $ (S_{\YM} + S_{\rsc}, -) $ and $s$ turn out to be isomorphic under the map $F \mapsto e^{(-, \Psi)} F$.
In the non-gauge fixed theory, $\cD_{\ba}$ is the correct derivative operator, in the sense that it commutes with $(S_{\YM} + S_{\rsc}, -)$.
The operator $\hat{\cD}_{\ba}$ is then obtained by the same canonical transformation, applied to $\cD_{\ba}$:
\beq
\label{eq:tilde-nabla-O' = nabla-O}
\hat{\cD}_{\ba} = e^{(- , \Psi)} \circ \cD_{\ba} \circ e^{-(- , \Psi)}.
\eeq
Thus, in view of \eqref{eq:(hat-S, O)}, the correction term can be seen to naturally arise as a consequence of gauge-fixing.
\end{remark}

\begin{remark}
In view of \eqref{eq:tilde-nabla-curvature}, one may, similarly to Fedosov's approach, add the tangent vector fields $\ba$ to $\sS_\YM$ as a new non-dynamical fermionic field and define a differential $\hat \delta = \skal{\hat{\cD} -}{\ba}$ on $\ba$ independent functionals, and extend it naturally to $\ba$ dependent ones. By \eqref{eq:[tilde-nabla,s]=0}, $\hat \delta$ and $s$ then anticommute, so that one may define a new differential $\hat s = s + \hat \delta$, whose cohomology at grade 0 gives the gauge invariant, background independent, on-shell local functionals. Such an approach was pursued by several authors in the literature, \cf \cite{KlubergStern:1974xv, Grassi:1995wr, Ferrari:2000yp, Anselmi:2013kba, Becchi:1999ir} for example. We do not proceed in this way here, basically because in the quantized theory, the flatness of the analog of $\hat{\cD}$ will only hold on cohomology, see below.
\end{remark}

\subsection{Perturbative quantum Yang-Mills theory on a background $\bar{\cA}$ }\label{section:q-YM}
\label{sec:QuantumYM}

In this section, we outline the perturbative quantization of the gauge-fixed Yang-Mills theory, described in the previous section, i.e., we adapt \cite{Hollands:2007zg} to the case of non-trivial background gauge fields.

The construction of the free algebra $\gls{W_A}$ is similar to the scalar case, discussed in Section~\ref{section:phi-4}, 
 now with the differential operator
 \beq
\label{eq:bar_P}
  \bar P =
  \begin{pmatrix} 
   (\bar{P}^\lin)^{\ \nu}_\mu & - \bar \nabla_\mu & 0 & 0 \\
   \bar \nabla^\nu & 1 & 0 & 0 \\
   0 & 0 & 0 & \bar \nabla^\lambda \bar \nabla_\lambda \\
   0 & 0 & - \bar \nabla^\lambda \bar \nabla_\lambda & 0
   \end{pmatrix}
 \eeq
acting on $(A_\nu, B, C, \bar C)$. Here $\bar P^\lin$ was defined in \eqref{eq:lin-YM-eq}. The corresponding Hadamard two-point function is of the form
 \beq
 \label{eq:2ptYM}
  \omega =
  \begin{pmatrix} 
   {\omega_{\mathrm{v}}}^{\ \mu}_{\nu} & \bar \nabla_\nu \omega_{\mathrm{s}} & 0 & 0 \\
   - \bar \nabla^\nu {\omega_{\mathrm{v}}}^{\ \mu}_{\nu} & 0 & 0 & 0 \\
   0 & 0 & 0 & - \omega_{\mathrm{s}} \\
   0 & 0 & \omega_{\mathrm{s}} & 0
   \end{pmatrix},
 \eeq
where one assumes the vector and scalar two-point functions $\omega_{\mathrm{v}}$, $\omega_{\mathrm{s}}$ to be related by
\begin{align}
\label{eq:2ptConsistency}
 \bar \nabla^\nu \circ {\omega_{\mathrm{v}}}^{\ \mu}_\nu & = \omega_{\mathrm{s}} \circ \bar \nabla^\mu, &
 \bar \nabla_\nu \circ {\omega_{\mathrm{s}}} & = {\omega_{\mathrm{v}}}^{\ \mu}_\nu \circ \bar \nabla_\mu
\end{align} 
in $\cU$. The latter condition ensures that $s_0$ defines a graded derivation on $\bW_{\bar{\cA}}$, i.e., 
\[
{s}_0 (F_1 \star \dots \star F_n) = \sum_{k} (-1)^{\sum_{l<k} \eps_l}  F_1 \star \dots \star {s}_0 F_{k} \star \dots  \star F_n
\]
for $F_i$'s supported in $\cU$.
That one can construct Hadamard two-point functions $\omega_{\mathrm{v}}$, $\omega_{\mathrm{s}}$ fulfilling these properties was shown in \cite{Gerard:2014jba, WrochnaZahn}. 

As for scalar fields, the on-shell algebra is defined by dividing out the ideal $\bJ_{\bar \cA}$ generated by the equations of motion $s_0 \Phi^\ddagger_i = 0$. It is important to note that these in general contain anti-fields, \cf Table~\ref{table1}. These are being treated as sources, \cf \cite{hollands2005conservation}, for example. 

Time-ordered products on the algebra $\textbf{W}_{\bar{\cA}}$ are defined analogously to the scalar case to be a collection of maps graded symmetric linear maps
\begin{equation*}
\gls{T-A} : (\bW^\loc_{\bar \cA})^{\otimes n} \to \bW_{\bar \cA},
\end{equation*}
which satisfy the axioms mentioned below \eqref{eq:TOP} with obvious modifications to adapt to the gauge fields, and with the difference that local covariance is now defined with respect to the morphisms $\chi$. Time ordered products with one factor, i.e., Wick powers, are defined analogously to the scalar case, \cf \eqref{eq:DefWickPower}, with a Hadamard parametrix $H$ of the same form of the two-point functions, \cf \eqref{eq:2ptYM}. In particular, the vector and scalar parametrices ${H_{\mathrm{v}}}^{\ \mu}_{\nu}$ and $H_{\mathrm{s}}$ fulfill identities analogous to \eqref{eq:2ptConsistency}, up to smooth remainders, which in fact vanish in the coinciding point limit.\footnote{This can be shown for example using the methods developed in \cite{Zahn:2014uwa}.}

\subsubsection{Ward identities}
\label{subsubsec:WI}

A crucial aspect of quantized gauge theory is the interplay of gauge invariance and renormalization. It is encoded in the \emph{anomalous Ward identity} \cite{Hollands:2007zg}
\begin{equation}\label{AnWI}
s_0 T_{\bar{\cA}} ( \eti{F} )  = \ibar T_{\bar \cA}( \{ s_0 F + \tfrac{1}{2} (F, F) + A(\et{F}) \} \otimes \eti{F} ),
\end{equation}
valid for $F$ supported in $\cU$.\footnote{In \cite{Hollands:2007zg}, this was proven for a flat background connection without restrictions on the support of $F$. This proof can be straightforwardly generalized to general background connections. However, a crucial ingredient is that $s_0$ is a derivation and nilpotent, which is only true on functionals supported in $\cU$. This motivates the localization $\supp S_\ia \subset \cU$, which, by \eqref{eq:T_support} ensures that $\supp T^\ia_{\bar \cA}(\eti{F}) \subset \cU$ for $\supp F \subset \cR$ (those are the generators that we will be concerned with).} 
Here $A(e_\otimes^F)= \sum_{n \geq 1} \frac{1}{n!} A_n(F^{\otimes n})$ is the \emph{anomaly}, where each $A_n$ is a map
\begin{equation*}
A_n : (\bW^\loc_{\bar \cA})^{\otimes n} \to \bW^\loc_{\bar \cA},
\end{equation*}
with properties similar to $D_n$, \cf \eqref{eq:counter-terms}, that is, it is of order $O(\hbar)$, decreases the total $\Deg$ by $2 (n-1)$, is supported on the total diagonal, is local and covariant and graded symmetric and scales homogeneously under \eqref{eq:Scaling}. As proven in Lemmata~\ref{lemma:anomaly-field-ind} and \ref{lemma:anomaly-linear}, it is (anti-) field independent and vanishes if one of the arguments is a linear (anti-) field. In addition, each $A_n$ increases the ghost number by 1. 
Furthermore, it is subject to the \emph{consistency condition} \cite{Hollands:2007zg}
\beq
\label{eq:CC_0}
 s_0 A(\et{F}) + (F, A(\et{F})) + A( \{ s_0 F + \tfrac{1}{2} (F, F) + A(\et{F}) \} \otimes \et{F} ) = 0.
\eeq

\begin{remark}
In generating identities such as \eqref{AnWI} or \eqref{eq:CC_0}, we always assume $F$ to be Grassmann even. To handle Grassmann odd $F$, one proceeds by multiplying with Grassmann odd parameters and differentiating \wrt them (taking care about the order).
\end{remark}

As argued below, a crucial consistency requirement is the absence of \emph{gauge anomalies}, i.e.,
\begin{equation}
\label{eq:A=0}
A(e_\otimes^{{S_\ia}})=0.
\end{equation}
The consistency condition \eqref{eq:CC_0} is crucial for the removal of anomalies, i.e., for achieving \eqref{eq:A=0}. Let us indicate how this proceeds. Consider the expansion of $A(e_\otimes^{{S_\ia}})$ in powers of $\hbar$
\[
A(e_\otimes^{{S_\ia}}) = A^{(m)}(e_\otimes^{{S_\ia}})\hbar^m + A^{(m+1)}(e_\otimes^{{S_\ia}})\hbar^{m+1} + \dots,
\]
for some integer $m>0$. Now we write $A^{(m)}(e_\otimes^{{S_\ia}}) = \int_M \alpha$ as an integral of a local four-form $\alpha(x)$ with ghost number $1$ and mass dimension $4$ (this follows from the homogeneous scaling of the anomaly). The consistency condition \eqref{eq:CC_0} for $F = S_\ia$ implies that $\alpha(x) \in H^{4}_1(s| \ud )$. If the cohomology ring $H^{4}_1(s| \ud)$ is trivial, then
 \begin{equation*}
\alpha(x) = s \beta(x) + \ud \gamma(x),
\end{equation*}
for some fields $\beta$, $\gamma$, of ghost number $0$ and $1$, respectively.
Such an anomaly can be removed by passing
to another renormalization scheme, as follows. Let us write the interaction \eqref{eq:gauge-fixed-action} as ${S_\ia} = \int_M L_\ia$, and let $L_1$ be the term of degree 3 in fields and anti-fields (so that $\Deg(\ibar L_1) = 1$).
We now choose a new scheme ${T}'$ by setting the following local finite counter terms $D_n$:
\begin{equation}
\label{eq:D-n}
D_n^{(m)}(L_1(x_1) \otimes \dots \otimes L_1(x_n)) = - \hbar^m \beta(x_1) \delta(x_1, \dots ,x_n),
\end{equation}
where $D^{(m)}$ is the first non-trivial term in the $\hbar$-expansion of $D(e_\otimes^{{S_\ia}})$ and where 
$n= 2 (m-1) + \deg_\phi \beta$.
The anomalies $A'$ and $A$ in the schemes $T'$ and $T$ are related via \cite{Hollands:2007zg}
\begin{equation*}
{A'}^{(m)}(e_\otimes^{{S_\ia}}) = A^{(m)}(e_\otimes^{{S_\ia}}) + s D^{(m)}(e_\otimes^{{S_\ia}}),
\end{equation*}
and therefore with the choice \eqref{eq:D-n} the anomaly in the new scheme vanishes
\[
{A'}^{(m)}(e_\otimes^{{S_\ia}}) = \int_M \alpha' = \int_M \alpha - s \beta = \int_M \ud \gamma =0.
\]
Repeating the argument for higher order coefficients of $A$ in $\hbar$, we can fully remove the anomaly.

For the pure Yang-Mills case, as can be seen from \eqref{eq: s|d-cohomology}, $H^{4}_1(s| \ud)$ is actually non-trivial. However, one can argue \cite{Hollands:2007zg} that the parity property of the possible gauge anomaly is indeed not compatible with that of $A(e_\otimes^{{S_\ia}})$ and hence is absent, so that there exist a renormalization scheme in which \eqref{eq:A=0} holds. In the following we assume to work with such a scheme.

\subsubsection{Quantum BRST charge and the algebra of physical observables}

In analogy with the scalar field theory, we can now define the generating functional
\[
T^\ia_{\bar{\cA}} ( \eti{F}) = \sum_{n=0} \frac{\ibar^n}{n!} T^\ia_{\bar{\cA},n} (F^{\otimes n})
\]
of interacting time ordered products. These generate the interacting algebra $\bW^\ia_{\bar \cA}$. Due to the time-slice axiom \cite{Chilian:2008ye}, it suffices to consider $F$'s supported in $\cR$. However, the algebra $\bW^\ia_{\bar \cA}$ also contains gauge-variant and unphysical functionals. They can be represented only on a space with indefinite inner product. However, the algebra of physical and gauge invariant renormalized observables is defined to be  \cite{Hollands:2007zg, Kugo:1979gm}
\begin{equation*}
\gls{F-A} \defeq \frac{\Ker [Q^\ia_{\bar{\cA}}, -]_\star}{\Ran [Q^\ia_{\bar{\cA}}, -]_\star}, \qquad \text{at ghost number } 0
\end{equation*} 
in the interacting on-shell algebra $\bW^\ia_{\bar \cA} \mod \bJ_{\bar \cA}$. Here, $\gls{Q-quant}$ is the renormalized interacting quantum BRST charge, obtained by applying the definition \eqref{eq:InteractingField} to the local functional $Q$ defined in \eqref{eq:Q}. Equality in $\bF_{\bar \cA}$ is thus equality modulo equations of motion and $\Ran [Q^\ia_{\bar{\cA}}, -]_\star$, i.e.,
\[
 F \gls{eqF} G \quad \Leftrightarrow \quad F - G - [Q^\ia, H]_\star \in \bJ
\]
for some $H$.
Under certain conditions, $\bF_{\bar \cA}$ admits a Hilbert space representation \cite{dutsch1999local,WrochnaZahn}.

Whether such a construction of  $\bF_{\bar{\cA}}$ can be implemented turns out to be closely related to the issue of local gauge-symmetry preservation at the quantum level, which has the following manifestations:
\begin{enumerate}[(i)]
\item conservation of the renormalized interacting Noether current $J^\ia_{\bar{\cA}}$ of BRST symmetry,
\item nilpotency of $[Q^\ia_{\bar{\cA}}, -]_\star$ generated by BRST charge $Q^\ia_{\bar{\cA}}$ (obtained from $J^\ia_{\bar{\cA}}$),
\item invariance of renormalized operators $[Q^\ia_{\bar{\cA}}, \cO_{\bar{\cA}}^\ia]_\star=0$, for classically gauge invariant $\cO$.
\end{enumerate}

As proven in \cite{Hollands:2007zg, Tehrani:2017pdx}, for any theory with local gauge symmetry, the first two manifestations listed above 
hold in the absence of gauge anomalies, i.e., when \eqref{eq:A=0} holds.
Also, the last manifestation follows from the anomalous Ward identity \eqref{AnWI} if,  in addition to  \eqref{eq:A=0}, we have \cite{Hollands:2007zg, Tehrani:2017pdx}
\begin{equation*}
A(\cO \otimes \et{S_\ia})=0,
\end{equation*}
which turns out to be a consequence of the triviality of $H_1(s)$. 

 The key identity in the proof of the above statements is the following \emph{interacting anomalous Ward identity} \cite{Tehrani:2017pdx}, 
\begin{equation}\label{int-WI}
\big[Q^\ia_{\bar{\cA}} , T^\ia_{\bar{\cA}} ( \eti{F}) \big]_\star \eqos - T^\ia_{\bar{\cA}} ( \{ s F + \tfrac{1}{2} (F, F) + A^\ia(\et{F}) \} \otimes \eti{F} ), 
\end{equation}
which holds for all $F$ supported in $\cR$, under assumption \eqref{eq:A=0}.\footnote{The proof given in \cite{Tehrani:2017pdx} is for flat background connections. The generalization to the general case is straightforward. The restriction on the support of $F$ stems from the fact that the differential $s F = (S, F)$ is only nilpotent for $F$ supported in $\cR$, \cf Section~\ref{sec:BRST}.}
Here $\eqos$ means equal modulo the ideal $\bJ_{\bar \cA}$ of free equations of motion, defined analogously to \eqref{eq:ideal-J_0}, i.e.,
\[
F \gls{eqos} G \quad \Leftrightarrow \quad F - G \in \bJ_{\bar \cA},
\]
and
$A^\ia(\et{F})= \sum_{n \geq 1} \frac{1}{n!} A^\ia_{n}(F^{\otimes n})$ is the generating functional of \emph{interacting anomalies}, defined by
\[
\gls{hat-A-n}(F_1 \otimes \dots  \otimes F_n) \defeq A(F_1 \otimes \dots  \otimes F_n \otimes \et{S_\ia}).
\]
These are subject to the \emph{interacting consistency conditions}
\beq
\label{eq:CC_int}
 s A^\ia(\et{F}) + (F, A^\ia(\et{F})) + A^\ia( \{ s F + \tfrac{1}{2} (F, F) + A^\ia( \et{F} ) \} \otimes \et{F} ) = 0.
\eeq
At first order in $F$, this implies that the \emph{quantum BV-BRST operator} \cite{Tehrani:2017pdx} defined by
\begin{equation}\label{eq:hat-q}
\gls{q} F := {s} F + A^\ia_1(F),
\end{equation}
is nilpotent, i.e., $q^2 = 0$.
Using this notation, we may express \eqref{int-WI} at first order in $F$ as
\beq
\label{eq:Q_F}
 [ Q^\ia_{\bar \cA}, F^\ia_{\bar \cA} ]_\star \eqos i \hbar \left( q F \right)^\ia_{\bar \cA}.
\eeq
We also note that by \eqref{int-WI}, the gauge invariant generators of interacting time-ordered products are given by $T^\ia_{\bar \cA}(\eti{F})$, with $F$ fulfilling
\beq
\label{eq:s-nonlin-F=0}
 s F + \tfrac{1}{2} (F, F) + A^\ia(\et{F}) = 0.
\eeq
In particular, an interacting field $F^\ia_{\bar \cA} = T^\ia_{\bar \cA}(F)$ is gauge invariant if $q F = 0$. Furthermore, given $F$ of ghost number 0 and fulfilling $q F = 0$, one may supplement it with ``contact terms'' to $F' = F + C(\et{F})$ such that $F'$ fulfills \eqref{eq:s-nonlin-F=0} in the sense of power series in $F$ \cite{Frob:2018buw}.


\subsubsection{Perturbative agreement and the background dependence of the anomaly}
\label{sec:PA}

As for the scalar case, perturbative agreement is a crucial ingredient for background independence. For variations of the background connection, it means
\beq
\label{eq:PA_YM}
{\delta}^\ret_{\ba}  T( \eti{F}) = T( \ibar \bar{\delta}_{\ba} F  \otimes \eti{F}) + R( \eti{F}; \ibar \bar{\delta}_{\ba} S_0 ).
\eeq
In the following, we sketch the proof that this can indeed be fulfilled in pure Yang-Mills theories, on a proof in a simpler context given in \cite{Zahn:2013ywa}.\footnote{Perturbative agreement will in general not hold when the gauge fields couple to chiral fermions, due to the usual chiral anomalies \cf \cite{Zahn:2014uwa}.} We then explore the interplay of perturbative agreement and anomalies.

We first need to define the retarded variation, to make sense of the \lhs of \eqref{eq:PA_YM}. We recall the differential operator valued matrix $\bar P_{i j}$ defined by
\beq
\label{eq:Def_P_ij}
 \bar P_{ij} \Phi^j(x) \vol(x) = \frac{\delta {{S}_0}|_{\Phi^\ddag=0}}{\delta \Phi^i(x)}, 
\eeq
 \cf \eqref{eq:bar_P}, and denote the corresponding retarded/advanced propagator by $\Delta^{ij}_{\ret/\adv}$. It fulfills
\beq \label{retarded-prop}
\bar P_{ik} \Delta^{kj}_{\ret/\adv} = \delta_i^j \id = \Delta^{jk}_{\ret/\adv} \bar P_{k i}.
\eeq
Let us also introduce the (differential operator valued) matrix $K^i_{\ j}$ defined by 
\beq
\label{eq:Def_K_ij}
K^i_{\ j} = \frac{\delta (s_0 \Phi^i )}{\delta \Phi^j} = \begin{pmatrix} 0 & 0 & \bar \nabla_\nu & 0 \\ 0 & 0 & 0 & 0 \\ 0 & 0 & 0 & 0 \\ 0 & 1 & 0 & 0 \end{pmatrix},
\eeq
so that $K^i_{\ j} \Phi^j = s_0 \Phi^i$, and its formal adjoint $\hat K_i^{\ j}$ such that 
\beq
\label{eq:Def_hat_K_ij}
 S_{\rsc, 0} = - \int s_0 \Phi^i \Phi^\ddag_i = - \int K^i_{\ j} \Phi^j \Phi^\ddag_i = - \int \Phi^i \hat K_i^{\ j} \Phi^\ddag_j.
\eeq
Then
\beq
\label{eq:s0_Phi_ddag}
 s_0 \Phi^\ddag_i = \frac{\delta^R}{\delta \Phi^i} S_0 = (-1)^{\eps} \left( \bar P_{i j} \Phi^j \vol - \hat K_i^{\ j} \Phi^\ddag_j \right),
\eeq
with $\eps$ the Grassmann parity of $\Phi^i$. 

Analogously to the definition of the retarded wave operator in the scalar case, \cf \eqref{eq:retardedWaveOperator}, we now define\footnote{We refer to \cite{Zahn:2013ywa} for a treatment of fermionic fields.}
\begin{align*}
 r_{\bar \cA', \bar \cA} \Phi^i(x) & \defeq \Phi^i(x) + \int \Delta'^{i j}_\ret(x, y) \left( (\bar P - \bar P' )_{j k} \Phi^k(y) \vol(y) - (\hat K - \hat K' )_j^{\ k} \Phi^\ddagger_k(y) \right), \\
 r_{\bar \cA', \bar \cA} \Phi^\ddag_i(x) & \defeq \Phi^\ddag_i(x).
\end{align*}
It maps solutions to the free equations of motion $s_0 \Phi^\ddag_i = 0$ on the background $\bar \cA$ to solutions on the background $\bar \cA'$. It follows that the retarded M{\o}ller operator $\tau^\ret$, defined as in \eqref{eq:def_MollerMap}, is well-defined on the on-shell algebra. One also defines its infinitesimal version, the retarded variation $\delta^\ret_{a}$, as for the scalar case, \cf \eqref{eq:retarded-var-def}. 

A crucial ingredient in the proof that perturbative agreement can be fulfilled is the free current, obtained as the variation of the free part of the action \wrt the background connection, i.e.,
\beq\label{eq:back-current}
 j(a) \defeq \bar \delta_{a} S_0.
\eeq
Here we naturally extend the action to off-shell backgrounds, i.e., $a$ is an arbitrary section of $\p \otimes \Omega^1$, not subject to the linearized equations of motion.
When no sources are present, this current is classically covariantly conserved on-shell. In the present case, this is spoiled by the presence of anti-fields. One finds the off-shell identity
\beq
\label{eq:div_j}
 \bar \nabla_\mu j^{I \mu} = - (-1)^\eps [\Phi^i, s_0 \Phi^\ddag_i]_\g^I - [K^i_{\ j} \Phi^j, \Phi^\ddag_i]_\g^I,
\eeq
with $\eps$ the Grassmann parity of $\Phi^i$.

We now have all the necessary ingredients to prove that \eqref{eq:PA_YM} can be fulfilled.

\begin{proposition}
\label{prop:PA}
In space-time dimension $D \leq 4$, perturbative agreement \eqref{eq:PA_YM} can be fulfilled.
\end{proposition}

\begin{proof}
As shown in \cite{Zahn:2013ywa}, \cf also \cite{hollands2005conservation}, perturbative agreement \eqref{eq:PA_YM} can be fulfilled, by a redefinition of time-ordered products involving at least one factor of $j(a)$, provided that\footnote{This requirement can be seen as a stronger version of the Wess-Zumino consistency condition, \cf \cite{Schenkel:2016nyj}.}
\beq
\label{eq:def_E}
 E(a_1, a_2) \defeq \delta^\ret_{a_1} T_1(j(a_2)) - \delta^\ret_{a_2} T_1(j(a_1)) + \ibar [T_1(j (a_1)), T_1(j(a_2))]_\star = 0.
\eeq
This quantity is (anti-) field independent. 
It was also shown \cite{Zahn:2013ywa} that, for space-time dimension $D \leq 4$, \eqref{eq:def_E} holds on-shell,
provided that the divergence of the Wick ordered current vanishes on-shell,
\beq
\label{eq:CC_os}
 \bar \nabla_\mu T_1(j^{I \mu}(x)) \eqos 0.
\eeq
As we argue below, this is true when anti-fields are set to zero (i.e., when the ideal generated by $\Phi^\ddag_i$ is modded out).
Thus, \eqref{eq:def_E} holds when equations of motion $s_0 \Phi^\ddag_i$ and anti-fields $\Phi^\ddag_i$ are modded out. But as $E(a_1, a_2)$ is independent of (anti-) fields, \eqref{eq:def_E} then also holds off-shell, and so does perturbative agreement \eqref{eq:PA_YM}.

It remains to argue that \eqref{eq:CC_os} indeed holds when anti-fields are set to zero. The first term on the \rhs of \eqref{eq:div_j} then yields equations of motion $[\Phi^i, \bar P_{i j} \Phi^j]^I_\g$. To evaluate the corresponding Wick-ordered product, one has to apply $\bar P$ to the Hadamard parametrix $H$ and evaluate the limit of coinciding points. This can be done, for example using the methods developed in \cite{Zahn:2014uwa}. However, one can directly see that the result must vanish, as it is a locally and covariantly constructed section of $\p$ of mass dimension four. No such quantity exists in parity non-violating models for semi-simple gauge groups.
\end{proof}

\begin{theorem} \label{thm:anomaly-bg-variation}
If perturbative agreement \eqref{eq:PA_YM} holds, background variations of the anomaly satisfy
\[
 \bar{\delta}_{\ba} A( \et{F} ) = A(\bar{\delta}_{\ba}(S_0 +F) \otimes \et{F})
\]
for all $F$ supported in $\cU$.
\end{theorem}

\begin{proof}
As the anomaly is local and\footnote{This follows from the fact that the \lhs is a c-number (by Lemma~\ref{lemma:anomaly-field-ind} and Lemma~\ref{lemma:anomaly-linear}) of ghost number 1. But no such $c$-number exists.} 
\beq
\label{eq:A_1_delta_S_0}
 A_1(\bar \delta_{\ba} S_0) = 0,
\eeq
we may choose $\ba$ to be supported in the region $\cU' \supset \cU$ in which the background $\bar \cA$ is on-shell. As nilpotency of $s_0$ and the anomalous Ward identity \eqref{AnWI} also hold on functionals supported in $\cU'$, we may thus use perturbative agreement \eqref{eq:PA_YM} and \eqref{AnWI} to obtain
\begin{align}
s_0 ( {\delta}^\ret_{\ba} T( \eti{F}) )&=  \ibar T \left( \left\{ s_0 F + \tfrac{1}{2} (F, F) + A( \et{F}) \right\}  \otimes   \ibar \bar{\delta}_{\ba} F  \otimes  \eti{F} \right) \nn \\
& \quad + \ibar T \left( \{( F,  \bar{\delta}_{\ba}( S_0 + F) ) + s_0 \bar{\delta}_{\ba} F + A( \bar{\delta}_{\ba} (S_0 + F) \otimes \et{F}) \} \otimes \eti{F} \right) \nn \\
& \quad + \ibar R \left( \left\{ s_0 F + \tfrac{1}{2}(F, F) + A(\et{F})  \right\} \otimes \eti{F}; \ibar \bar{\delta}_{\ba} S_0  \right) \nn \\
\label{eq:s_0_delta_ret_T_F}
& \quad + \ibar R( \eti{F}; s_0 \bar \delta_{\ba} S_0),
\end{align}
where we have again used \eqref{eq:A_1_delta_S_0}.
Regarding the last term on the r.h.s., one computes
\[
 s_0 \bar \delta_\ba S_0 = \int A^I_\mu [ (\bar P^\lin \ba)^\mu, C ]^I \vol.
\]
In particular, this is supported outside of $\cU$. We may thus decompose as
\[
 s_0 \bar \delta_\ba S_0 = (s_0 \bar \delta_\ba S_0)_- + (s_0 \bar \delta_\ba S_0)_+,
\]
with $\supp (s_0 \bar \delta_\ba S_0)_\pm \subset J^\pm(\cU) \setminus \cU$. It follows that the last term in 
\eqref{eq:s_0_delta_ret_T_F} may be rewritten as a commutator,
\[
 R( \eti{F}; s_0 \delta_{\ba} S_0) = - [ T( (s_0 \bar \delta_\ba S_0)_- ), T(\eti{F}) ]_\star.
\]
On the other hand, we have
\begin{align*}
{\delta}^\ret_{\ba}( s_0 T( \eti{F})) &= \ibar  T \left(  \left\{ s_0 F + \tfrac{1}{2}(F, F) + A( \et{F})  \right\} \otimes \ibar \bar{\delta}_{\ba} F  \otimes \eti{F} \right) \\
& \quad + \ibar T \left(  \left\{ s_0 \bar{\delta}_{\ba} F + (F,  \bar{\delta}_{\ba} (S_0 + F)) + \bar{\delta}_{\ba} A(\et{F})  \right\} \otimes \eti{F} \right)  \\
& \quad + \ibar  R \left(  \left\{ s_0 F + \tfrac{1}{2}(F, F) + A(\et{F})  \right\} \otimes \eti{F};  \ibar \bar{\delta}_{\ba} S_0  \right).
\end{align*} 
We thus obtain
\begin{align}\label{eq:[delta^r,s_0]}
[ {\delta}^\ret_{\ba}, s_0] T(\eti{F}) & = \ibar  T \left( \{  \bar{\delta}_{\ba} A(\et{F}) - A(\bar{\delta}_{\ba}(S_0 +F) \otimes \et{F}) \} \otimes \eti{F} \right) + \ibar [ T( (s_0 \bar \delta_{\ba} S_0)_- ), T(\eti{F}) ]_\star .
\end{align}
In particular, $[ {\delta}^\ret_{\ba}, s_0]$ acts on linear (anti-) fields as
\begin{align*}
[ {\delta}^\ret_{\ba}, s_0] \Phi(x) & = \ibar  \{  \bar{\delta}_{\ba} A_1(\Phi(x))  - A_1(\bar{\delta}_{\ba} \Phi(x))  - A_2(\bar{\delta}_{\ba}S_0 \otimes \Phi(x)) \}  + \ibar [ T( ( s_0 \bar \delta_\ba S_0 )_- ), \Phi(x) ]_\star \\
 & = \ibar [ T( ( s_0 \bar \delta_\ba S_0 )_- ), \Phi(x) ]_\star,
\end{align*}
for $x \in \cU$,
since the anomaly of a linear (anti-) field vanishes, \cf Lemma~\ref{lemma:anomaly-linear}.\footnote{This can also be shown directly, using the definition of $s_0$ and $\delta^\ret_\ba$ on $\Phi^i$.} The action of both $s_0$ and $\delta^\ret_{\ba}$, and thus also of $[ {\delta}^\ret_{\ba}, s_0]$, on non-linear functionals is defined by their action on linear functionals, i.e.,\footnote{To be precise, $\delta^\ret_\ba$ also acts non-trivially on background fields, by $\delta^\ret_\ba \bar \cA = \ba$. However, as $s_0$ acts trivially on background fields, so does $[\delta^\ret_\ba, s_0]$.}
\begin{align*}
[ {\delta}^\ret_{\ba}, s_0] T(\eti{F}) & = \int \left\{ [ {\delta}^\ret_{\ba}, s_0] \Phi^i(x) \tfrac{\delta}{\delta \Phi^i(x)}T(\eti{F}) + [ {\delta}^\ret_{\ba}, s_0] \Phi_i^\ddag(x) \tfrac{\delta}{\delta \Phi_i^\ddag(x)}T(\eti{F}) \right\} \\
& = \int \ibar [ T( ( s_0 \bar \delta_\ba S_0 )_- ), \Phi^i(x) ]_\star \tfrac{\delta}{\delta \Phi^i(x)} T(\eti{F}) \vol(x).
\end{align*}
Comparing with \eqref{eq:[delta^r,s_0]} shows that we are finished if we can show that
\begin{align*}
 0 & = \ibar [ T( (s_0 \bar \delta_{\ba} S_0)_- ), - ]_\star - [\delta^\ret_\ba, s_0] - \\
 & =  \ibar [ T( (s_0 \bar \delta_{\ba} S_0)_- ), - ]_\star - \int \ibar [ T( ( s_0 \bar \delta_\ba S_0 )_- ), \Phi^i(x) ]_\star \tfrac{\delta}{\delta \Phi^i(x)} - \vol(x).
\end{align*}
The \rhs of this equation is of the form
\[
 \int W^{i j}(x,y) \tfrac{\delta}{\delta \Phi^i(x)} \tfrac{\delta}{\delta \Phi^j(y)} - \vol(x) \vol(y),
\]
with some smooth\footnote{Smoothness follows from the Hadamard property of the two-point function and \cite{HoermanderI}, Thm.~8.2.14.} kernel $W^{ij}$ which vanishes unless $\eps_i + \eps_j \mod 2 = 1$. It thus suffices to show that this vanishes when acting on $\Phi^i(x) \Phi^j(y)$ with $\eps_i + \eps_j \mod 2 = 1$. By this restriction, we have $T_2( \Phi^i(x) \otimes \Phi^j(y) ) = \Phi^i(x) \Phi^j(y)$ and plugging $F = \lambda_1 \Phi^i(x) + \lambda_2 \Phi^j(y)$ in \eqref{eq:[delta^r,s_0]} and considering the equation at $O(\lambda_1 \lambda_2)$, we indeed find that $W^{i j}$ must vanish, again by the absence of anomalies of linear fields, Lemma~\ref{lemma:anomaly-linear}.
\end{proof}

For the following considerations, it turns out to be convenient to introduce the notation
\beq
\label{eq:def-us}
 \us (\bar{\delta}_{\ba} \Psi ) \defeq \cD_{\ba} S -  {\delta}_{\ba} S_0 = \bar \delta_{\ba} S - \delta_{\ba} {S_\ia},
\eeq
even though outside of $\cR$, $\us$ does not need to be well-defined as an operator on local functionals. The important point is that in $\cR$, i.e., when restricted to configurations supported in $\cR$, $\us$ reduces to the BV-BRST differential $s$, \cf Proposition~\ref{prop:cD-S}.

\begin{corollary}
\label{cor:cDhatA}
If perturbative agreement \eqref{eq:PA_YM} holds, then, for $F$ supported in $\cU$,
\[
 \cD_{\ba} A^\ia(\et{F}) = A^\ia( \cD_{\ba} F \otimes \et{F}) + A^\ia( \us \bar \delta_{\ba} \Psi \otimes \et{F}),
\]
with $\us \bar \delta_{\ba} \Psi$ defined by \eqref{eq:def-us}. In particular, for $F$ supported in $\cR$, and $n \geq 1$,
\beq
\label{eq:cD-hatA}
 \cD_{\ba} A^\ia( F^{\otimes n} ) = n A^\ia( \cD_{\ba} F \otimes F^{\otimes (n-1)}) + A^\ia( s \bar \delta_{\ba} \Psi \otimes F^{\otimes n}).
\eeq
\end{corollary}
\begin{proof}
By field independence of the anomaly, Lemma~\ref{lemma:anomaly-field-ind}, we have
\[
 \delta_{\ba} A^\ia( \et{F} ) = A^\ia ( \delta_{\ba} F \otimes \et{F} ) + A^\ia(\delta_{\ba} {S_\ia} \otimes \et{F}).
\]
With Theorem~\ref{thm:anomaly-bg-variation}, we obtain
\[
 \cD_{\ba} A^\ia( \et{F} ) = A^\ia( \cD_{\ba} F \otimes \et{F}) + A^\ia( \{ \bar \delta_{\ba} (S_0 + {S_\ia}) - \delta_{\ba} {S_\ia} \} \otimes \et{F}),
\]
which proves the first claim.
The locality of the anomaly and the fact that on $\cR$, $\us \bar \delta_{\ba} \Psi = s \bar \delta_{\ba} \Psi$, then leads to \eqref{eq:cD-hatA}.
\end{proof}

\subsection{Background independence}\label{subsection:back-ind}
Having introduced the setting for the quantum Yang-Mills theory perturbatively constructed around each background $\bar{\cA}$, we now turn to the formulation of background independence. In analogy with the case of scalar field theory (Section \ref{section:BI-scalar}), we can identify the theories defined on different backgrounds via the retarded variation $ \delta^\ret_{\ba}$.
As shown in Section~\ref{sec:PA},
we can assume that perturbative agreement \eqref{eq:PA_YM} holds, and we will do so from now on.
Using this variation, we want to define a flat connection $\fD_{\ba}$ on the bundle
\[
\gls{W-bundle-coh} \defeq \bigsqcup_{\bar{\cA}} \bF_{\bar{\cA}} \rightarrow \sS_\YM,
\]
where $\sS_\YM$ is the manifold of background field configurations which are solutions to the Yang-Mills equation, \cf also the discussion following \eqref{eq:lin-YM-eq}. A connection is here defined in complete analogy to Definition~\ref{def:Connection}. The local algebras $\bF_{\bar \cA}(\cL)$ are then generated by $T^\ia_{\bar \cA}(\eti{F})$ with $F$ supported in $\cL$ and fulfilling \eqref{eq:s-nonlin-F=0}. We would also like to ensure that in the classical limit, it should reduce to the connection $\hat{\cD}_{\ba}$ on classical local functionals, in the sense that
\beq
\label{eq:fD_cD_A}
 \fD_{\ba} T^\ia_{\bar \cA}(\eti{F}) \eqF \ibar T^\ia_{\bar \cA}( \hat{\cD}_{\ba} F \otimes \eti{F}) + O(\hbar)
\eeq
for all $F$ fulfilling \eqref{eq:s-nonlin-F=0}.

We proceed by defining $\fD_{\ba}$ on the full bundle
\[
 \gls{W-bundle-A} \defeq \bigsqcup_{\bar{\cA}} \bW_{\bar{\cA}} \rightarrow \mathcal{S}_\YM
\]
and showing that it reduces to a connection on $\bF_\YM$, fulfilling the required properties. In particular, we have to ensure that
\begin{enumerate}[(i)]
\item it is well-defined on the on-shell algebra;
\item it is well-defined on $[Q^\ia_{\bar \cA}, -]_\star$ cohomology, i.e., it fulfills \eqref{eq:fD_well-defined} on-shell, ensuring that it maps kernel and image of $[Q^\ia_{\bar \cA}, -]_\star$ onto themselves;
\item it is a derivation, i.e., fulfills \eqref{eq:Leibniz};
\item it respects space-time localization in the sense defined in \eqref{eq:ConnectionLocalization}.
\end{enumerate}
A crucial requirement for the fulfillment of these properties will be the absence of a certain anomaly. We will later show that time-ordered products can indeed be defined accordingly.

\begin{remark}
\label{rem:CutoffIssue}
There is a subtlety regarding the definition of the bundles $\bF_\YM$ and $\bW_\YM$. We recall that the backgrounds $\bar \cA$ are only required to be on-shell in $\cU$ (and to coincide with an arbitrary reference connection $\cA_0$ outside of $\cV$). Hence, their behavior in $\cV \setminus \cU$ is arbitrary. A further requirement should thus be that the construction is independent of the choice of a representative, i.e., the connection $\fD_\ba$ should vanish for $\ba$ supported in $\cV \setminus \cU$, when applied to $T^\ia_{\bar \cA}(\eti{F})$ for $F$ supported in $\cR$. That this is indeed the case is checked below, \cf Remark~\ref{rem:fD}.
\end{remark}

To construct the desired connection $\fD$, it is useful to split the connection $\hat{\cD}$ on local functionals as 
\[
 \hat{\cD}_\ba = \left\{ \bar \delta_\ba - ( -, \bar \delta_\ba \Psi) \right\} - \left\{ \delta_\ba - ( -, \delta_\ba \Psi) \right\},
\]
where the two terms on the \rhs are obtained by applying the canonical gauge fixing transformation as in \eqref{eq:tilde-nabla-O' = nabla-O} separately to $\bar \delta_\ba$ and $\delta_\ba$. Hence, it is natural to see the first term on the \rhs as the gauge-fixed background variation, and replace it by the retarded variation. Our first tentative definition is thus
\begin{equation*}
\fD_{\ba}^0 \defeq \delta^\ret_{\ba} - {\delta}_{\ba} + (- , \delta_{\ba} \Psi).
\end{equation*}
That this is a natural starting point is evidenced by the following Lemma:
\begin{lemma}
The operator $\fD^0_\ba$ is well-defined on the on-shell algebra.
\end{lemma}

\begin{proof}
As the retarded variation is well-defined on the on-shell algebra, it remains to check for the last two terms. We have
\beq
\label{eq:delta_bar_C_as_antibracket}
 (- , \delta_{\ba} \Psi) = - \skal{\tfrac{\delta}{\delta \bar C^\ddag} -}{\bar \nabla^\mu \ba_\mu \vol},
\eeq
so that the last two terms are derivatives \wrt (anti-) fields.
Such a derivative is well-defined on the on-shell algebra if it acts in the direction of a solution to the free equations of motion, i.e., those obtained by $s_0 \Phi^\ddag_i$, \cf Table~\ref{table1}. The perturbation given by
\[
 (A, B, C, \bar C, A^\ddag, B^\ddag, C^\ddag, \bar C^\ddag) = (\ba, 0, 0, 0, 0, 0, 0, \bar \nabla^\mu \ba_\mu \vol)
\]
indeed fulfills that requirement, as $\ba$ is, in $\cU$, a solution to \eqref{eq:lin-YM-eq}.
\end{proof}

\subsubsection{Well-definedness of the connection on the quantum BRST cohomology}
\label{subsubsec:Welldefined}

Similarly to the case of scalar field theory (Proposition \ref{nabla-hbar}), $\fD_{\ba}^0$ acts as
\[
\fD_{\ba}^0 T^\ia_{\bar{\cA}}(\eti{F}) = \ibar T^\ia_{\bar{\cA}}( \cD^0_{\ba} F \otimes \eti{F} )  +  \ibar R^\ia_{\bar{\cA}}(\eti{F} ; (\cD^0_{\ba} S_\ia +  \bar{\delta}_{\ba} S_0)),
\]
where
\begin{equation*}
\cD^0_{\ba} \defeq \cD_{\ba} + (- , \delta_{\ba} \Psi).
\end{equation*}
We note that, by 
$({S_\ia}, {\delta}_{\ba} \Psi) =0$, we have 
\[
\cD^0_{\ba} {S_\ia} = \cD_{\ba} {S_\ia}.
\]
With the notation \eqref{eq:def-us}, we thus obtain
\begin{align}\label{eq:nabla-hbar-O}
\fD_{\ba}^0 T^\ia_{\bar{\cA}}(\eti{F}) = \ibar T^\ia_{\bar{\cA}}( \cD^0_{\ba} F )  +  \ibar R^\ia_{\bar{\cA}}(\eti{F} ; \us (\bar{\delta}_{\ba} \Psi )).
\end{align}
Note the presence of the second term on the \rhs of \eqref{eq:nabla-hbar-O} which is absent in the case of scalar field theory, cf. \eqref{eq:nablahbar-nabla}. It leads to a violation of the locality requirement \eqref{eq:ConnectionLocalization}. This term appears because the gauge-fixing fermion breaks the split independence of the action $S$, \cf \eqref{eq:ShiftSymmetryYM_gf}.

We first state a lemma which is crucial for the proof of the following theorem.

\begin{lemma}
\label{lemma:C=0}
For all $F$ supported in $\cR$, it holds 
\begin{multline}
\label{eq:Def_C}
 \cD^0_{\ba} \{ s F + \tfrac{1}{2} (F, F) + A^\ia(\et{F}) \} - s \cD^0_{\ba} F - (s \bar \delta_{\ba} \Psi, F) - (F, \cD^0_{\ba} F) \\ - A^\ia( \cD^0_{\ba} F \otimes \et{F}) - A^\ia( s \bar \delta_{\ba} \Psi \otimes \et{F}) = 0.
\end{multline}
\end{lemma}

\begin{proof}
As a consequence of \eqref{eq:[tilde-nabla,s]=0}, \eqref{eq:nabla-tilde-anti-bracket}, the graded Jacobi identity \eqref{eq:classical-Jacobi-id}, and \eqref{eq:cD-hatA}, the \lhs equals
\[
 ( s F, \bar \delta_{\ba} \Psi ) - s ( F, \bar \delta_{\ba} \Psi ) - (s \bar \delta_{\ba} \Psi, F) + (A^\ia(\et{F}), \delta_{\ba} \Psi) - A^\ia((F, \delta_{\ba} \Psi) \otimes \et{F}) .
\]
The first three terms cancel due to \eqref{eq:s_as_antibracket} and \eqref{eq:classical-Jacobi-id} and the last two terms due to Lemma~\ref{lemma:anomaly-field-ind}, taking into account \eqref{eq:delta_bar_C_as_antibracket} and the fact that $S_\ia$ is independent of $\bar C^\ddagger$.
\end{proof}

\begin{theorem}
\label{thm:fDQ}
Assuming
\begin{align}
\label{eq:A1(delta-Psi)=0}
 A^\ia_1(\bar{\delta}_{a} \Psi) = 0, \qquad \forall a, \ \supp a \subset \cR,
\end{align}
with $a$ not necessarily a solution to \eqref{eq:lin-YM-eq}, the operator 
\begin{align}
\gls{fD-A}& \defeq \fD_{\ba}^0 + \ibar [  (  \us ( \bar \delta_{\eta \ba} \Psi))^\ia_{\bar \cA}, - ]_\star \nn \\
\label{eq:def_fD_A}
 & = \delta^\ret_{\ba} - \delta_{\ba} - (-, \delta_{\ba} \Psi) + \ibar [  (  \us ( \bar \delta_{\eta \ba} \Psi))^\ia_{\bar \cA}, - ]_\star,
\end{align}
where $ \us ( \bar \delta_{\eta \ba} \Psi)$ is defined in \eqref{eq:def-us} and $\eta$ is a smooth non-negative function supported on $J^-(\cR)$ and equal to $1$ on $J^-(\cR) \setminus \cR$, is well-defined on the on-shell $[Q^\ia_{\bar{\cA}}, -]_\star$ cohomology, in the sense that
\beq
\label{eq:fD_welldefined}
\fD_{\ba} [ Q^\ia_{\bar{\cA}}, T^\ia_{\bar\cA}(\eti{F})]_\star  -   [ Q^\ia_{\bar{\cA}}, \fD_{\ba} T^\ia_{\bar\cA}(\eti{F})]_\star \eqos 0
\eeq
for all $F$ supported in $\cR$.
On this cohomology, it is independent of the choice of $\eta$. Furthermore, for $F$ fulfilling \eqref{eq:s-nonlin-F=0}, we have
\beq
\label{eq:fD_T_int}
 \fD_{\ba} T^\ia_{\bar \cA}(\eti{F}) \eqF \ibar T^\ia_{\bar \cA}( \{ \hat{\cD}_{\ba} F + A^\ia( \bar \delta_{\ba} \Psi \otimes \et{F}) \} \otimes \eti{F} ).
\eeq
In particular, $\fD_{\ba}$ is a connection on $\bF_{\YM}$ fulfilling \eqref{eq:fD_cD_A}.
\end{theorem}

\begin{remark}
\label{rem:fD}
The last term in the definition \eqref{eq:def_fD_A} can be motivated as follows: Assume that $\ba$ is supported outside of $\cU$. As discussed in Remark~\ref{rem:CutoffIssue}, the corresponding derivative $\fD_\ba$ should vanish on $T^\ia_{\bar \cA}(\eti{F})$ with $F$ localized in $\cR$. The first term on the \rhs of \eqref{eq:nabla-hbar-O} does indeed vanish (as the supports of $\ba$ and $F$ are disjoint), but the second one does not. However, due to causal factorization \eqref{eq:CausalFactorization_int} of interacting time-ordered products, it is cancelled by the commutator which is added in \eqref{eq:def_fD_A}. This is completely analogous to the unitary transformation (its generator in the present case) which compensates a change of the infra-red cut-off of the interaction in the so-called algebraic adiabatic limit, \cf \cite{Brunetti:1999jn}.
\end{remark}

\begin{proof}
We begin by proving the independence of the choice of $\eta$. The difference $\xi = \eta_1 - \eta_2$ of two admissible $\eta$s is supported in $\cR$, where $\us$ coincides with $s$. Hence, under the assumption \eqref{eq:A1(delta-Psi)=0} and using \eqref{eq:Q_F},
\[
 (  \us ( \bar \delta_{\eta_1 \ba} \Psi))^\ia_{\bar \cA} -  (  \us ( \bar \delta_{\eta_2 \ba} \Psi))^\ia_{\bar \cA} =  (  s ( \bar \delta_{\xi \ba} \Psi))^\ia_{\bar \cA} \eqos - \ibar [Q^\ia_{\bar{\cA}}, ( \bar \delta_{\xi \ba} \Psi)^\ia_{\bar \cA} ]_\star.
\]
But $[ [Q^\ia_{\bar{\cA}}, ( \bar \delta_{\xi \ba} \Psi)^\ia_{\bar \cA} ]_\star, -]_\star$ vanishes on $[Q^\ia_{\bar{\cA}}, -]_\star$ cohomology, as 
\[
 [[Q^\ia_{\bar{\cA}}, ( \bar \delta_{\xi \ba} \Psi)^\ia_{\bar \cA} ]_\star, - ]_\star = [Q^\ia_{\bar{\cA}}, [( \bar \delta_{\xi \ba} \Psi)^\ia_{\bar \cA}, - ]_\star]_\star + (-1)^\eps [( \bar \delta_{\xi \ba} \Psi)^\ia_{\bar \cA}, [Q^\ia_{\bar{\cA}}, - ]_\star]_\star,
\]
so that the action of $[ [Q^\ia_{\bar{\cA}}, ( \bar \delta_{\xi \ba} \Psi)^\ia_{\bar \cA} ]_\star, -]_\star$ on a $[Q^\ia_{\bar{\cA}}, -]_\star$ closed functional yields a $[Q^\ia_{\bar{\cA}}, -]_\star$ exact functional, i.e., a zero element in the cohomology.

We continue with proving \eqref{eq:fD_welldefined}. From equation \eqref{eq:nabla-hbar-O} and for all $F$, we have
\beq
\label{eq:fD_O_1st_step}
\fD_{\ba} T^\ia_{\bar\cA}( \eti{F}) = \ibar T^\ia_{\bar \cA}( \cD^0_{\ba} F \otimes \eti{F})  + \ibar R^\ia_{\bar{\cA}}( \eti{F} ; \us (\bar{\delta}_{\ba} \Psi )) + \ibar [  (\us (\bar \delta_{\eta \ba} \Psi))^\ia_{\bar \cA}, T^\ia_{\bar\cA}( \eti{F}) ]_\star.
\eeq
By the above, we may, without loss of generality, assume that $\supp \eta \cap J^+( \supp F ) = \emptyset$. We split
\[
  \us ( \bar \delta_{\ba} \Psi ) =  \us ( \bar \delta_{\eta \ba} \Psi ) +  s ( \bar \delta_{\chi \ba} \Psi ) +  \us ( \bar \delta_{\psi \ba} \Psi ),
\]
in the second term on the \rhs of \eqref{eq:fD_O_1st_step}, where $\eta$, $\chi$, and $\psi$ are smooth non-negative functions, summing up to $1$, with $\chi$ being supported inside $\cR$ and equal to $1$ in a neighborhood of $\supp F$, $\eta$ being supported in $J^-(\cR)$, and $\psi$ supported in $J^+(\cR)$. By causal factorization \eqref{eq:CausalFactorization_int} and \eqref{eq:R_1_int}, we have
\begin{align*}
 R^\ia_{\bar{\cA}}( \eti{F} ; \us (\bar{\delta}_{\ba} \Psi ) ) & = R^\ia_{\bar{\cA}}( \eti{F} ; \us (\bar{\delta}_{\chi \ba} \Psi ) ) + R^\ia_{\bar{\cA}}( \eti{F} ; \us (\bar{\delta}_{\eta \ba} \Psi ) ) \\
 & = R^\ia_{\bar{\cA}}( \eti{F} ; s \bar{\delta}_{\chi \ba} \Psi ) + T^\ia_{\bar \cA}( \eti{F} \otimes \us (\bar{\delta}_{\eta \ba} \Psi )) - (\us (\bar{\delta}_{\eta \ba} \Psi ))^\ia_{\bar \cA} \star T^\ia_{\bar \cA}( \eti{F}) \\
 & = R^\ia_{\bar{\cA}}( \eti{F} ; s \bar{\delta}_{\chi \ba} \Psi ) - [ (\us (\bar{\delta}_{\eta \ba} \Psi ))^\ia_{\bar \cA},  T^\ia_{\bar \cA}( \eti{F})]_\star.
\end{align*}
We thus obtain 
\beq
\label{eq:fD_T_int_calc}
\fD_{\ba} T^\ia_{\bar\cA}( \eti{F}) = \ibar T^\ia_{\bar \cA}( \cD^0_{\ba} F \otimes \eti{F})  + \ibar R^\ia_{\bar{\cA}}( \eti{F} ; s \bar{\delta}_{\chi \ba} \Psi ).
\eeq
With \eqref{int-WI}, we compute, using the assumption \eqref{eq:A1(delta-Psi)=0},
\begin{align*}
 [Q^\ia_{\bar \cA}, \fD_{\ba} T^\ia_{\bar\cA}( \eti{F})]_\star & \eqos - \ibar T^\ia_{\bar \cA}( \{ s F + \tfrac{1}{2} (F, F) + A^\ia(\et{F}) \} \otimes \cD^0_{\ba} F \otimes \eti{F}) \\
 & \quad - T^\ia_{\bar \cA}( \{ s \cD^0_{\ba} F + (F, \cD^0_{\ba} F) + A^\ia( \cD^0_{\ba} F \otimes \et{F}) \} \otimes \eti{F}) \\
 & \quad - \ibar R^\ia_{\bar \cA}( \{ s F + \tfrac{1}{2} (F, F) + A^\ia(\et{F}) \} \otimes \eti{F}; s \bar{\delta}_{\chi \ba} \Psi ) \\
 & \quad - T^\ia_{\bar \cA}( \{ (s \bar{\delta}_{\chi \ba} \Psi, F) + A^\ia(s \bar{\delta}_{\chi \ba} \Psi \otimes \et{F}) \} \otimes \eti{F})
\end{align*}
and
\begin{align*}
 \fD_{\ba} [Q^\ia_{\bar \cA}, T^\ia_{\bar\cA}( \eti{F})]_\star & \eqos - \fD_{\ba} \left( T^\ia_{\bar \cA}( \{ s F + \tfrac{1}{2} (F, F) + A^\ia(\et{F}) \} \otimes \eti{F}) \right) \\
 & \eqos - \ibar T^\ia_{\bar \cA}( \{ s F + \tfrac{1}{2} (F, F) + A^\ia(\et{F}) \} \otimes \cD^0_{\ba} F \otimes \eti{F}) \\
 & \quad - T^\ia_{\bar \cA}( \cD^0_{\ba} \{ s F + \tfrac{1}{2} (F, F) + A^\ia(\et{F}) \} \otimes \eti{F}) \\
 & \quad - \ibar R^\ia_{\bar \cA}( \{ s F + \tfrac{1}{2} (F, F) + A^\ia(\et{F}) \} \otimes \eti{F}; s \bar{\delta}_{\chi \ba} \Psi ).
\end{align*}
It follows that
\[
  \fD_{\ba} [ Q^\ia_{\bar{\cA}}, T^\ia_{\bar\cA}(\eti{F})]_\star  -   [ Q^\ia_{\bar{\cA}}, \fD_{\ba} T^\ia_{\bar\cA}(\eti{F})]_\star \eqos T^\ia_{\bar \cA}( C_{\ba}(F) \otimes \eti{F})
\]
with $C_{\ba}(F)$ the expression on the \lhs of \eqref{eq:Def_C}. Lemma~\ref{lemma:C=0} thus proves \eqref{eq:fD_welldefined}.

Finally, we note that using  \eqref{int-WI} and \eqref{eq:A1(delta-Psi)=0}, we may rewrite \eqref{eq:fD_T_int_calc} as
\begin{multline*}
 \fD_{\ba} T^\ia_{\bar\cA}(\eti{F}) \eqos \ibar T^\ia_{\bar \cA}( \{ \hat{\cD}_{\ba} F + A^\ia( \bar \delta_{\ba} \Psi \otimes \et{F}) \} \otimes \eti{F}) \\ 
 - \ibar^2 R^\ia_{\bar{\cA}}( \{ s F + \tfrac{1}{2} (F, F) + A^\ia( \et{F}) \} \otimes \eti{F} ; \bar{\delta}_{\chi \ba} \Psi ) - \ibar^2 [Q^\ia_{\bar \cA}, R^\ia_{\bar \cA}( \eti{F} ; \bar \delta_{\chi \ba} \Psi )]_\star,
\end{multline*}
which proves \eqref{eq:fD_T_int}.

As a direct consequence of \eqref{eq:fD_T_int}, $\fD_{\ba}$ fulfills \eqref{eq:fD_cD_A} and respects space-time localization in the sense defined in \eqref{eq:ConnectionLocalization}, and so defines a connection on $\bF_\YM$.
\end{proof}

\subsubsection{Flatness of the connection on the quantum BRST cohomology}
\label{subsubsec:Flatness}

Finally, we want to prove flatness of $\fD_{\ba}$.

\begin{theorem}\label{thm:hat-cD-flatness}
For all $F$ supported in $\cR$, fulfilling \eqref{eq:s-nonlin-F=0}, and under the assumption \eqref{eq:A1(delta-Psi)=0} 
we have
\[
 ([ \fD_{\ba} , \fD_{\ba'} ] - \fD_{\LB{\ba}{\ba'}}) T^\ia_{\bar \cA}(\eti{F}) \eqF 0. 
\] 
\end{theorem}
\begin{proof}
Using \eqref{eq:fD_T_int}, it suffices to prove
\[
 \ibar T^\ia_{\bar \cA}( D_{\ba, \ba'}(F) \otimes \eti{F} ) \eqF 0,
\]
with
\begin{align*}
 D_{\ba, \ba'}(F) & = ([ \hat{\cD}_{\ba} , \hat{\cD}_{\ba'} ] - \hat{\cD}_{\LB{\ba}{\ba'}}) F + \hat{\cD}_{\ba} A^\ia(\bar \delta_{\ba'} \Psi \otimes \et{F}) - \hat{\cD}_{\ba'} A^\ia(\bar \delta_{\ba} \Psi \otimes \et{F}) \\ 
 & + A^\ia(\bar \delta_{\ba} \Psi \otimes \hat{\cD}_{\ba'} F \otimes \et{F}) - A^\ia(\bar \delta_{\ba'} \Psi \otimes \hat{\cD}_{\ba} F \otimes \et{F}) - A^\ia(\bar \delta_{\LB{\ba}{\ba'}} \Psi \otimes \et{F}) \\
 & + A^\ia( A^\ia( \bar \delta_{\ba'} \Psi \otimes \et{F} ) \otimes \bar \delta_{\ba} \Psi \otimes \et{F}) - A^\ia( A^\ia( \bar \delta_{\ba} \Psi \otimes \et{F} ) \otimes \bar \delta_{\ba'} \Psi \otimes \et{F}).
\end{align*}
By \eqref{eq:cD-hatA}, we have
\begin{align*}
 & \cD_{\ba} A^\ia(\bar \delta_{\ba'} \Psi \otimes \et{F}) - \cD_{\ba'} A^\ia(\bar \delta_{\ba} \Psi \otimes \et{F}) - A^\ia(\bar \delta_{\LB{\ba}{\ba'}} \Psi \otimes \et{F}) \\
  & = A^\ia( \bar \delta_{\ba'} \Psi \otimes s \bar \delta_{\ba} \Psi \otimes \et{F}) - A^\ia( \bar \delta_{\ba} \Psi \otimes s \bar \delta_{\ba'} \Psi \otimes \et{F}) \\ 
  & \phantom{=} + A^\ia(\bar \delta_{\ba'} \Psi \otimes \cD_{\ba} F \otimes \et{F}) - A^\ia(\bar \delta_{\ba} \Psi \otimes \cD_{\ba'} F \otimes \et{F})  + A^\ia( \{ \cD_{\ba} \bar \delta_{\ba'} - \cD_{\ba'} \bar \delta_{\ba} - \bar \delta_{\LB{\ba}{\ba'}} \} \Psi \otimes \et{F}).
\end{align*}
The last term on the \rhs vanishes by the flatness of $\bar \delta$ and Lemma~\ref{lemma:anomaly-linear}. Thus, with \eqref{eq:tilde-nabla-curvature}, we have
\begin{align*}
 D_{\ba, \ba'}(F) & = A^\ia( \bar \delta_{\ba'} \Psi \otimes s \bar \delta_{\ba} \Psi \otimes \et{F}) - A^\ia( \bar \delta_{\ba} \Psi \otimes s \bar \delta_{\ba'} \Psi \otimes \et{F}) \\ 
 & - (A^\ia(\bar \delta_{\ba'} \Psi \otimes \et{F}), \cD_{\ba} \Psi) +  (A^\ia(\bar \delta_{\ba} \Psi \otimes \et{F}), \cD_{\ba'} \Psi) \\ 
 & - A^\ia(\bar \delta_{\ba} \Psi \otimes (F, \cD_{\ba'} \Psi) \otimes \et{F}) + A^\ia(\bar \delta_{\ba'} \Psi \otimes (F, \cD_{\ba} \Psi) \otimes \et{F}) \\
 & + A^\ia( A^\ia( \bar \delta_{\ba'} \Psi \otimes \et{F} ) \otimes \bar \delta_{\ba} \Psi \otimes \et{F}) - A^\ia( A^\ia( \bar \delta_{\ba} \Psi \otimes \et{F} ) \otimes \bar \delta_{\ba'} \Psi \otimes \et{F}).
\end{align*}
With the consistency condition \eqref{eq:CC_int}
this simplifies to
\begin{multline*}
 D_{\ba, \ba'}(F) = s A^\ia(\bar \delta_{\ba'} \Psi \otimes \bar \delta_{\ba} \Psi \otimes \et{F}) + (F, A^\ia(\bar \delta_{\ba'} \Psi \otimes \bar \delta_{\ba} \Psi \otimes \et{F})) \\ + A^\ia ( A^\ia(\bar \delta_{\ba'} \Psi \otimes \bar \delta_{\ba} \Psi \otimes \et{F}) \otimes \et{F}),
\end{multline*}
where we used Lemma~\ref{lemma:anomaly-field-ind}, taking into account \eqref{eq:delta_bar_C_as_antibracket} and the fact that $S_\ia$ is independent of $\bar C^\ddagger$, and that $\Psi$ does not contain anti-fields, so that $( \bar \delta_{\ba'} \Psi, \bar \delta_{\ba} \Psi) = 0$. With \eqref{int-WI}, we thus obtain
\[
 \ibar T^\ia_{\bar \cA}( D_{\ba, \ba'}(F) \otimes \eti{F} ) \eqos [Q^\ia_{\bar \cA}, T^\ia_{\bar \cA}( A^\ia(\bar \delta_{\ba} \Psi \otimes \bar \delta_{\ba'} \Psi \otimes \et{F}) \otimes \eti{F} )]_\star,
\]
which proves the statement.
\end{proof}

\subsubsection{Absence of obstructions to background independence}\label{section:YM-no anomaly}

Above, we found that condition  \eqref{eq:A1(delta-Psi)=0} is sufficient to ensure well-definedness and flatness of the connection $\fD_{\ba}$ on $[Q^\ia_{\bar \cA}, -]_\star$ cohomology. We now show that this can indeed be satisfied in pure Yang-Mills theory.

\begin{lemma}
Let $T$ and $T'$ be two renormalization schemes related via \eqref{eq:T-T'}, and let $A$ and $A'$ be the corresponding anomalies of the anomalous Ward identities \eqref{AnWI} in these schemes. 
Assuming the anomalies of the interaction $S_\ia$ vanish in both schemes, i.e., $A(\et{S_\ia}) = A'(\et{S_\ia})=0$, then for all  $G$ supported in $\cR$, it holds
\begin{align}
\label{(S+Z-hat-F)}
s Z_{S_\ia}G + (D_{S_\ia}, {Z}_{S_\ia} G) + A(Z_{S_\ia} G \otimes \et{S_\ia+D_{S_\ia} } ) = Z_{S_\ia} ( sG + A'^\ia_1(G)),
\end{align}
where $D_{S_\ia} := D( \et{S_\ia})$ and  ${Z}_{S_\ia} G := G + D(G \otimes \et{S_\ia})$.
\end{lemma}
\begin{proof}
From relation \eqref{eq:T-T'} it follows that 
\[
T' (G \otimes \eti{F} ) = T'( Z_F G \otimes \eti{(F + D_F)} ). 
\]
Using this identity, the anomalous Ward identity in the scheme $T'$ takes the form
\begin{align}\nonumber
s_0 T' ( \eti{F} )  &= \ibar T' ( \{ s_0 F + \tfrac{1}{2}(F, F) + A'(\et{F}) \} \otimes \eti{F} ) \\ \label{tildeT-AnWI-RHS}
& = \ibar T ( Z_F \{ s_0 F + \tfrac{1}{2}(F, F) + A'( \et{F} ) \} \otimes \eti{( F + D_F)} ).
\end{align} 
On the other hand using \eqref{eq:T-T'}, we can write the anomalous Ward identity as
\begin{align}\nonumber
s_0 T' ( \eti{F} )  & = s_0 T ( \eti{(F+ D_F)} ) \\ \label{tildeT-AnWI-LHS}
&= \ibar T ( \{ s_0 (F + D_F) + \tfrac{1}{2}( F +D_F, F + D_F )  + A(\et{F+D_F}) \} \otimes \eti{( F + D_F)} ).
\end{align}
Comparing \eqref{tildeT-AnWI-RHS} and \eqref{tildeT-AnWI-LHS}, we arrive at 
\begin{align*}
s_0 (F +D_F) + \tfrac{1}{2}( F +D_F, F + D_F )  + A(\et{F+ D_F }) =  Z_F ( s_0 F + \tfrac{1}{2}(F,F) + A'(\et{F}) ).
\end{align*}
Now \eqref{(S+Z-hat-F)} follows by replacing $F$ with $F + \tau G$, differentiating with respect to $\tau$, and setting $\tau=0$ and $F={S_\ia}$.
\end{proof}

In the following, we show that the a violation of condition \eqref{eq:A1(delta-Psi)=0} can be removed by a redefinition of time-ordered products. The strategy is as follows: Assume the the anomaly has been removed up to order $O(\hbar^{m-1})$, i.e.,
\beq
\label{eq:AnomalyExpansion}
 A^\ia_1(\bar{\delta}_{a} \Psi )  = \sum_{n\ge m}  \hbar^n A^{\ia (n)}_1(\bar{\delta}_{a} \Psi ),
\eeq
with $A^{\ia (n)}$ independent of $\hbar$. We denote by $A^{(m)}$ and $D^{(m)}$ the anomaly and the redefinition of time-ordered product at order $O(\hbar^m)$. From \eqref{(S+Z-hat-F)}, we conclude that
\beq
\label{eq:AnomalyRedefinitionExpanded}
 A'^{(m)}(G \otimes \et{S_\ia})  = {A}^{(m)}(G \otimes \et{S_\ia})  + s D^{(m)}(G \otimes \et{S_\ia}) - D^{(m)}(s G \otimes \et{S_\ia})+  (G, D^{(m)}(\et{S_\ia})).
\eeq
There are thus two possible strategies to remove the anomaly of $\bar \delta_a \Psi$ at order $O(\hbar^m)$: The first one would be to set
\beq
\label{eq:AnomalyRedefinition}
 D^{(m)}(s \bar \delta_a \Psi \otimes \et{S_\ia}) = A^{(m)}(\bar \delta_a \Psi \otimes \et{S_\ia}).
\eeq
However, such a definition must not spoil the absence of gauge anomalies or perturbative agreement. As discussed in the proof of Proposition~\ref{prop:PA}, achieving perturbative agreement proceeds by redefining time-ordered products involving at least one factor $j(a) = \bar \delta_a S_0$, \cf \eqref{eq:back-current}. Hence, redefinitions of such time-ordered products should not be allowed. Furthermore, due to field independence, a redefinition of a time-ordered product of the form $T(\delta_a S_i \otimes \eti{S_\ia})$, with the interaction $S_\ia = S_1 + S_2$ and $\Deg(\ibar S_i) = i$, would spoil the absence of gauge anomalies. However, by \eqref{eq:ShiftSymmetryYM_gf}, we have
\[
 s_0 \bar \delta_a \Psi = \bar \delta_a S_0 - \delta_a S_1.
\]
Hence, the time ordered products $T(s_0 \bar \delta_a \Psi \otimes \eti{S_\ia})$ must not be redefined. Thus, to implement \eqref{eq:AnomalyRedefinition}, one would have to redefine time ordered products of the form $T(s_\ia \bar \delta_a \Psi \otimes \eti{S_\ia})$. Concretely, one would set
\beq
\label{eq:AnomalyRemoval}
 D^{(m)}(s_\ia \bar{\delta}_a \Psi \otimes \et{S_\ia}|_{n-1}) = A_1^{(m)}(\bar{\delta}_a \Psi \otimes \et{S_\ia}|_n),
\eeq
for $n \geq 1$, with
\beq
\label{eq:ExpExpansion}
 \et{S_\ia}|_n = \sum_{k_1 + 2 k_2 = n} \frac{1}{k_1! k_2!} S_1^{\otimes k_1} \otimes S_2^{\otimes k_2}.
\eeq
Note the different number of interaction terms on the two sides of \eqref{eq:AnomalyRemoval}, which is enforced by the fact that $s_\ia \bar{\delta}_a \Psi$ is cubic in the fields, while $\bar{\delta}_a \Psi$ is only quadratic. This redefinition is still problematic. First, one has to show that the \rhs of \eqref{eq:AnomalyRemoval} vanishes for $n = 0$. Second, and more severe, are constraints from field independence. One can find $a'$ such that $\delta_{a'} s_\ia \bar{\delta}_a \Psi$ vanishes. Hence, if such a derivative $\delta_{a'}$ acts on the first variable of the functional on the l.h.s., one gets a functional that identically vanishes. Hence, all such derivatives only act on the $S_i$ factors. But there are less such factors on the \lhs than on the r.h.s., so that the redefinition \eqref{eq:AnomalyRemoval} might be inconsistent with field independence.

In order to circumvent these difficulties, we exploit the second (in fact related, \cf Remark~\ref{rem:RedefinitionSchemeRelation}) possibility to removing an anomaly based on \eqref{eq:AnomalyRedefinitionExpanded}. Namely, if $A^{(m)}(\bar \delta_a \Psi \otimes \et{S_\ia})$ happens to be $s$ exact, i.e., $A^{(m)}(\bar \delta_a \Psi \otimes \et{S_\ia}) = s H_a$, then we may set
\[
 D^{(m)}(\bar \delta_a \Psi \otimes \et{S_\ia}) = - H_a.
\]
Unfortunately, $A^{(m)}(\bar \delta_a \Psi \otimes \et{S_\ia})$ need not be $s$ exact. However, it turns out to be $s_0$ exact, which is sufficient to remove the anomaly order by order in the number of fields. To prove these statements, we collect a few lemmata.

\begin{lemma}
\label{lemma:H(s_0)}
The cohomology $H^k(s_0)$ is trivial at negative ghost number $k < 0$.
\end{lemma}

\begin{proof}
The statement was shown in \cite{Barnich:2000zw}, Thm.~7.1, for the full differential $s$ and the restricted algebra not containing $B$, $\bar C$, and their anti-fields. However, adding these does not change the statement, as they form trivial pairs and do not modify the cohomology. The proof given in \cite{Barnich:2000zw} only uses the triviality of the homology of the Koszul-Tate differential at positive anti-field number, and this also holds for its free part.
\end{proof}

\begin{lemma}
\label{lemma:F=sG}
Let $F = s G$ be of ghost number $0$ and $F_i \neq 0$ be the lowest order term of $F$ in an expansion in total (anti-) field number. Then there exists $G_i$ such that $F_i = s_0 G_i$.
\end{lemma}

\begin{proof}
Let $G_j \neq 0$ be the lowest order term in the (anti-) field number expansion of $G$. If $j = i$, we have found the sought for $G_i$. For $j < i$, we note that $s_0 G_j = 0$. By Lemma~\ref{lemma:H(s_0)}, there is $H_j$ such that $G_j = s_0 H_j$. Define $G^{(1)} = G - s H_j$. We still have $s G^{(1)} = F$, but now the lowest order term of $G^{(1)}$ occurs at $j^{(1)}> j$. We continue until $j^{(k)} = i$.
\end{proof}

\begin{lemma}
\label{lemma:s0Exact}
Let $\supp a \subset \cR$, $A_1^{\ia (m)}(\bar \delta_a \Psi) \neq 0$ be the lowest term in the $\hbar$ expansion of $A_1^\ia(\bar \delta_a \Psi)$, \cf \eqref{eq:AnomalyExpansion}, and $A^{\ia (m)}_1(\bar \delta_a \Psi)_i \neq 0$ be the lowest order term of $A^{\ia (m)}_1(\bar \delta_a \Psi)$ in an expansion in total (anti-) field number. Then there exists $G_{a i}$ such that $A^{\ia (m)}_1(\bar \delta_a \Psi)_i = s_0 G_{a i}$.
\end{lemma}

\begin{proof}
Expanding the consistency condition \eqref{eq:CC_int} in $\hbar$, one finds $A^{\ia (n)}_1(s \bar \delta_a \Psi) = 0$ for all $n < m$ and
\[
 s A^{\ia(m)}_1(\bar \delta_a \Psi) = A^{\ia(m)}_1(s \bar \delta_a \Psi). 
\]
By Corollary~\ref{cor:cDhatA} and the absence of gauge anomalies, \cf \eqref{eq:A=0}, we have $A^\ia_1(\us \bar \delta_{\bar a} \Psi) = 0$ for all $\bar a$ fulfilling the linearized field equation \eqref{eq:lin-YM-eq}. In particular, the local functional $A^\ia_1(s \bar \delta_a \Psi)$ vanishes when evaluated in configurations supported in a region $\cR' \subset \cR$ where $a$ is on-shell, i.e., where $\bar P^\lin a = 0$ holds. It follows that, for $a$ supported in $\cR$,
\[
 A^{\ia(m)}_1(s \bar \delta_a \Psi) = \int \Phi^{(m) I \mu} \bar P^\lin a^I_\mu \vol,
\]
with $\Phi^{(m)}$ a locally and covariantly constructed section of $\p \otimes \Omega^1(M)$ of ghost number $1$ and mass dimension $1$. Furthermore, again by expanding the consistency condition \eqref{eq:CC_int}, one finds $s \Phi^{(m)} = 0$. From the triviality of $H^1(s)$, it follows that $\Phi^{(m) I}_\mu = c^{(m)} s A^I_\mu$ with some coefficient $c^{(m)}$. Hence,
\beq
\label{eq:A_1_delta_Psi}
 A_1^{\ia(m)}(\bar \delta_a \Psi) = c^{(m)} \int A^{I \mu} \bar P^\lin a^I_\mu \vol + \int \Theta^{(m) I \mu} a^I_\mu \vol,
\eeq
with $\Theta^{(m)}$ a locally and covariantly constructed section of $\p \otimes \Omega^1(M)$ of ghost number $0$, mass dimension $3$, and in the kernel of $s$. By \eqref{eq: s-cohomology}, it must thus be of the form $\Theta^{(m) I \mu} = s \Sigma^{(m) I \mu} + \Xi^{(m) I \mu}$, with $\Xi^{(m)}$ a c-number. However, the only such c-number would be $\bar \nabla_\nu \bar F^{I \nu \mu}$, which vanishes in $\cR$, \cf \eqref{eq:bg-eom}. Noting that the first term on the \rhs of \eqref{eq:A_1_delta_Psi} can be rewritten as an element of the image of $s_0$ using $s_0 (A^{\ddag I}_\mu + \bar \nabla_\mu \bar C^I) = (\bar P^\lin A)^I_\mu$, and using Lemma~\ref{lemma:F=sG} on the second term, we obtain the desired statement.
\end{proof}

We are now ready to perform the necessary redefinitions.

\begin{theorem}
\label{thm:NoAnomaly}
There are renormalization schemes in which the condition \eqref{eq:A1(delta-Psi)=0} holds.
\end{theorem}

\begin{proof}
Let $A_1^{\ia (m)}(\bar \delta_a \Psi)$ be the lowest term in the $\hbar$ expansion \eqref{eq:AnomalyExpansion} and $A^{\ia (m)}_1(\bar \delta_a \Psi)_i \neq 0$ be the lowest order term of $A^{\ia (m)}_1(\bar \delta_a \Psi)$ in an expansion in total (anti-) field number. By Lemma~\ref{lemma:s0Exact}, there is a $G_{a i}$ such that $A^{\ia (m)}_1(\bar \delta_a \Psi)_i = s_0 G_{a i}$. Now we perform the redefinition
\beq
\label{eq:AnomalyRemovalFinal}
 D^{(m)}(\bar \delta_a \Psi \otimes \et{S_\ia}|_{2(m-1) + i}) = - G_{a i},
\eeq
where we used the notation \eqref{eq:ExpExpansion}. Both expressions are at the same order in the interaction, so there are no potential obstructions from field independence. Expanding \eqref{eq:AnomalyRedefinitionExpanded} in the total (anti-) field number, we see that the anomaly now occurs at a higher order in the total (anti-) field number. By power counting, the anomaly has a bounded total (anti-) field number, so the process terminates at some point, so that the anomaly at order $O(\hbar^m)$ is removed. Continuing at higher orders, one removes the anomaly to all orders.
\end{proof}

\begin{remark}
\label{rem:RedefinitionSchemeRelation}
The two possibilities for removing the anomaly, i.e., by either redefining time-ordered products involving $s_\ia \bar \delta_a \Psi$ or $\bar \delta_a \Psi$, are in fact related. This follows from field independence and the fact that $\frac{\delta}{\delta C} s_\ia \bar \delta_a \Psi$ involves the same Wick power as $\bar \delta_a \Psi$, namely $\bar C A$. Hence, the redefinition \eqref{eq:AnomalyRemovalFinal} implies a redefinition of the form \eqref{eq:AnomalyRemoval}, with a modified right hand side. One can thus see the approach chosen here as a means to rule out the potential clashes with field independence discussed below \eqref{eq:ExpExpansion}.
\end{remark}

\begin{remark}
\label{remark:YM-higher-dim}
The above arguments invoked power counting, and thus relied on power counting renormalizablity. Our method is thus not sufficient to rule out violation of background independence for example in Yang-Mills in higher dimensions.
\end{remark}

\begin{remark}
\label{rem:Abelian}
Let us consider the situation when the gauge group is not semi-simple, but contains abelian factors and possibly also matter fields. The proof of anomaly freedom given in \cite{Hollands:2007zg} does then not apply, but let us assume that there are no gauge anomalies. How are our considerations then affected? The dynamical fields corresponding to the abelian factors are free (apart from the possible coupling to matter), so the gauge fixing fermion $\Psi$ is independent of the abelian background connection. In particular, $\bar \delta_{a} \Psi = 0$ if the perturbation $a$ is only in the abelian background connection.
It follows that no further potential obstructions to achieving \eqref{eq:A1(delta-Psi)=0}
arise by including abelian factors.
\end{remark}

\subsubsection{Summary of assumptions}
\label{sec:YM_Summary}
Even though we discussed Yang-Mills theory here, the treatment of other gauge theories should be completely analogous, provided that a few conditions are met. Obviously, the theory should have no gauge anomaly, i.e., \eqref{eq:A=0} holds, and fulfill perturbative agreement \wrt changes in the background. Also the triviality of $H_1(s)$ was used. Apart from that, we used that
\begin{enumerate}[(i)]
\item $S_\ia$ does not contain $\bar C^\ddag$,
\item the gauge-fixing fermion is quadratic in fields,
\item does not contain anti-fields.
\end{enumerate}
If these conditions are met, then background independence holds, provided that the analog of \eqref{eq:A1(delta-Psi)=0} does.

\begin{remark}
Throughout, we also assumed compact Cauchy surfaces. This assumption is of technical nature only. It is relevant for the existence of the interacting BRST charge $Q^\ia_{\bar \cA}$, but as long as one is not interested in singling out the physical subspace in a Hilbert space representation, this charge is not needed. We only use it in the form $[Q^\ia_{\bar \cA}, - ]_\star$ of the on-shell interacting BRST differential. One could equally well work with the off-shell interacting BRST differential $\hat s$ recently constructed in \cite{Frob:2018buw} (which does not require compact Cauchy surfaces). For non-compact Cauchy surfaces, the construction of the cut-off functions needed for example in Thm.~\ref{thm:fDQ} or Thm.~\ref{thm:anomaly-bg-variation} becomes slightly more involved, but apart from the fact that the existence of Hadamard states for non-compact Cauchy surfaces has not been proven in full generality \cite{Gerard:2014jba}, our conclusions also hold for non-compact Cauchy surfaces (with the obvious replacements of $[Q^\ia_{\bar \cA}, - ]_\star$ by $\hat s$).
\end{remark}

\subsubsection{Renormalized background independent interacting fields}

Lemma~\ref{lemma:C=0} can be seen as a master equation for the compatibility of $\hat{\cD}_\ba = \cD^0_\ba - (-, \bar \delta_\ba \Psi)$ and $s$ in the renormalized setting. Let us explore some consequences.

We recall that for a local functional $F$ to give rise to a gauge invariant interacting field $T^\ia_{\bar \cA}(F)$, it must fulfill $q F = 0$, \cf \eqref{eq:hat-q} for the definition of $q$, corresponding to the linearization of \eqref{eq:s-nonlin-F=0}. As shown in \cite{Frob:2018buw},  for any field $\cO$ of ghost number $0$ which is classically gauge invariant, $s \cO = 0$, there is extension $\cO' = \cO + O(\hbar)$ such that $q \cO' = 0$. 
A further structure that naturally occurs at second order in $F$ is the \emph{quantum anti-bracket} \cite{Tehrani:2017pdx}
\[
 (F_1, F_2)_\hbar \defeq (F_1, F_2) + (-1)^{\eps_1} A^\ia_2(F_1 \otimes F_2).
\]
We may now define
\[
 \cD^\hbar_{\ba} \defeq \cD_{\ba} - (-, \cD_{\ba} \Psi)_\hbar = \hat{\cD}_\ba - A^\ia_2(\cD_\ba \Psi \otimes - ),
\]
which is equal to $\hat{\cD}_{\ba}$ up to quantum corrections. It follows from \eqref{eq:fD_T_int} that a functional $F$ giving rise to a background independent interacting field $T^\ia_{\bar \cA}(F)$ must fulfill
\[
 \cD^\hbar_{\ba} F \in \Ran q.
\]
A straightforward consequence of the consistency condition \eqref{eq:CC_int}, Lemma~\ref{lemma:C=0} and Theorem~\ref{thm:hat-cD-flatness} is the following:
\begin{corollary}
For all $F$ supported in $\cR$, we have, under the assumption \eqref{eq:A1(delta-Psi)=0},
\[
 q \cD^\hbar_{\ba} F - \cD^\hbar_{\ba} q F = 0.
\]
If furthermore $q F = 0 = q G$, then also
\begin{align*}
 \{ [\cD^\hbar_{\ba}, \cD^\hbar_{\ba'}] - \cD^\hbar_{\LB{\ba}{\ba'}} \} F & \in \Ran q, \\
 \cD^\hbar_{\ba} (F, G)_\hbar - ( \cD^\hbar_{\ba} F, G)_\hbar - (F, \cD^\hbar_{\ba} G)_\hbar & \in \Ran q.
\end{align*}
\end{corollary}

A natural question is now the following. Assume a field $\cO$ is given which is classically gauge invariant and background independent, i.e.,
\[ 
 s \cO = 0, \qquad \hat{\cD}_{\ba} \cO = 0,
\]
is there an extension $\cO' = \cO + O(\hbar)$ such that
\[
q \cO' = 0, \qquad \cD^\hbar_{\ba} \cO' \in \Ran q
\]
holds, so that ${\cO'}^\ia_{\bar \cA}$ is a gauge invariant, background independent field? Given an extension $\cO'$ such that $q \cO' = 0$, one may evaluate it on one background $\bar \cA'$ and then obtain local functionals on general backgrounds by parallel transport \wrt $\cD^\hbar$, at least locally on $\sS_\YM$ (using that $\cD^\hbar$ is flat on $q$ cohomology). However, it is not obvious whether one may choose the extension $\cO'$ such that this procedure results in a proper field in the sense defined in Section~\ref{section:phi-4}, i.e., is independent of $\bar \cA'$. We leave this as an interesting open problem.

\section{Perturbative quantum gravity}\label{section:pert-Grav}

Having treated background independence for Yang-Mills theory in full detail, we now turn to perturbative quantum gravity, with an emphasis on the differences to the Yang-Mills case. Quantum gravity in the sense of perturbation theory around generic backgrounds was recently formulated in \cite{Brunetti:2013maa}. Our setup differs in an important point, so this difference will also be highlighted.

The principal dynamical variable is the metric perturbation $h_{\mu \nu}$, i.e.,  the full metric is given by 
\[
g_{\mu \nu} = \bar g_{\mu \nu} + h_{\mu \nu},
\]
with $\bar g_{\mu \nu}$ the background metric. It is supplemented by ghosts $c^\mu$, antighosts $\bar c_\mu$, Lagrange multipliers $b_\mu$, which are (co-) vector fields and transform under the BRST transformation as 
\begin{align*}
s h_{\mu \nu} & = \nabla_\mu c_\nu + \nabla_\nu c_\mu, &
s c^\mu & = c^\nu \nabla_\nu c^\mu, &
s \bar c_\mu  & = i b_\mu, &
s b_\mu & = 0,
\end{align*}
with $\nabla$ the Levi-Civita derivative \wrt $g_{\mu \nu}$.
Correspondingly, the Einstein-Hilbert action is extended to
\[
 S_\EH+ S_\rsc = \int_M \left( R[\bar g + h] \vol[\bar g + h] - \Lie_c g_{\mu \nu} h^{\mu \nu \ddagger} - i b_\mu \bar c^{\mu \ddagger} - c^\nu \nabla_\nu c^\mu c_\mu^\ddagger \right),
\]
where the antifields $\Phi^\ddagger$ are interpreted as tensor-valued densities. The action is invariant under background gauge transformations, i.e.,  diffeomorphisms $\psi: M' \to M$ acting via pull-back on $\bar g$ and the dynamical fields.

As for Yang-Mills fields, the interaction terms are adiabatically cut-off. There is a slight complication \wrt the Yang-Mills case in that the cut-off function should be a function of \emph{covariant coordinates}, \cf below. That, however, does not change anything substantial, so this cut-off can be treated as for Yang-Mills fields. Hence, we ignore this subtlety in the following. In the region where the cut-off function is equal to one, the extended action has the shift symmetry
\[
 \frac{\delta ( S_\EH+ S_\rsc)}{\delta \bar g(x)} = \frac{\delta ( S_\EH+ S_\rsc)_\ia}{\delta h(x)}.
\]

To implement the harmonic (or de Donder) gauge, we employ the gauge fixing fermion \cite{Ichinose:1992np}
\[
 \Psi = i \int_M \left( \bar \nabla_\mu \bar c_\nu ( \bar g^{\mu \lambda} \bar g^{\rho \nu} - \tfrac{1}{2} \bar g^{\mu \nu} \bar g^{\lambda \rho} ) h_{\lambda \rho} - \tfrac{1}{2} b_\mu \bar c^\mu \right) \vol[\bar g],
\]
which is a covariant functional of the dynamical fields and the background metric $\bar g$. Here $\bar \nabla$ is the Levi-Civita derivative \wrt $\bar g_{\mu \nu}$. The gauge fixed action then becomes
\begin{multline*}
 S = S_{\rsc} + \int_M \left\{ R[g] \vol[g] - \left( \bar \nabla_\mu b_\nu ( \bar g^{\mu \lambda} \bar g^{\rho \nu} - \tfrac{1}{2} \bar g^{\mu \nu} \bar g^{\lambda \rho} ) h_{\lambda \rho} - \tfrac{1}{2} b_\mu b^\mu \right) \vol[\bar g] \right. \\
 \left. - i \left( 2 \bar \nabla^{(\mu} \bar c^{\nu)} ( \bar \nabla_{\mu} c_{\nu} + \tfrac{1}{2} c^\lambda \bar \nabla_\lambda h_{\mu \nu} + h_{\lambda \mu} \bar \nabla_\nu c^\lambda ) - \bar \nabla_\lambda \bar c^\lambda ( \bar \nabla_\rho c^\rho + \tfrac{1}{2} c^\rho \bar \nabla_\rho h + h_{\mu \nu} \bar \nabla^\nu c^\mu ) \right) \vol[\bar g] \right\}
\end{multline*}
with $h \defeq \bar g^{\mu \nu} h_{\mu \nu}$. It leads to hyperbolic equations of motion at the linearized level for $c^\mu$, $\bar c_\nu$ and $\gamma_{\mu \nu} \defeq h_{\mu \nu} - \frac{1}{2} \bar g_{\mu \nu} h$ (after eliminating $b_\nu$).


Local observables can be constructed as proposed in \cite{Brunetti:2013maa, Khavkine:2015fwa}, by what one might call \emph{covariant coordinates}. One chooses backgrounds $\bar g$ that are sufficiently generic\footnote{For similar constructions on non-generic backgrounds, we refer to \cite{Brunetti:2016hgw}, \cite{Frob:2017coq}.} to allow, in a neighborhood of $\bar g$, for $4$ curvature scalars to provide a coordinate system $X[g] : M \to U \subset \R^4$. By definition, these fulfil
\[
 X[\psi^* g] = X[g] \circ \psi
\]
for a diffeomorphism $\psi$. It follows that
\begin{align}
\label{eq:CovariantCoordinate}
 \psi^* \circ X^*[g] & = X^*[\psi^* g], &
 X[\psi^* g]_* & = X[g]_* \circ \psi_*.
\end{align}

Given a test tensor $t$ on $M$, and $T[g]$ a tensor covariantly constructed out of the metric, i.e.,  obtained by contractions of $g_{\mu \nu}$, $g^{\mu \nu}$, $\nabla_{(\lambda_1} \dots \nabla_{\lambda_r)} R_{\mu \nu \rho \sigma}$, one defines
\beq
\label{eq:GravityObservable}
 T_{\bar g}(t)(h) = \int_M t_{\mu_1 \dots \mu_k}^{\nu_1 \dots \nu_l} X[\bar g]^* \circ X[\bar g + h]_* ( \vol[\bar g + h] T[\bar g + h] )^{\mu_1 \dots \mu_k}_{\nu_1 \dots \nu_l}.
\eeq
From \eqref{eq:CovariantCoordinate}, it follows that the observable \eqref{eq:GravityObservable} transforms covariantly,
\[
 T_{\psi^* \bar g}(\psi^* t)(\psi^* h) = T_{\bar g}(t)(h),
\]
and is in the kernel of the BRST operator. We refer to \cite{Brunetti:2013maa, Khavkine:2015fwa} for the interpretation of these observables.

An adiabatic cut-off of the interaction terms, respecting covariance, can be implemented similarly. Let $L_\ia$ be the interaction Lagrangian density, obtained by Taylor expansion of the Lagrangian density in $(h, c, \bar c, b, h^\ddagger, c^\ddagger, \bar c^\ddagger)$ and keeping only the terms of order higher than two. Then a covariant cutoff can be implemented as 
\[
 {S_\ia} = \int_M \lambda X[\bar g]^* \circ X[\bar g + h]_* (L_\ia),
\]
with $\lambda$ a test function on the background, assumed to be equal to one in a neighborhood of the region $\cR$, \cf the set-up for the Yang-Mills case.

As for the case of pure Yang-Mills theory, there are no gauge anomalies and $H_1(s)$ is trivial \cite{Barnich:1994kj}, and there is also no obstruction to the fulfilment of perturbative agreement \cite{hollands2005conservation} for variations in the background metric. Also the conditions (i)--(iii) stated in Section~\ref{sec:YM_Summary} are met. It follows that it suffices to check the fulfillment of the analog of condition \eqref{eq:A1(delta-Psi)=0}, which is
\[
A^\ia_1(\bar{\delta}_{k} \Psi ) = 0, 
\]
where
\[
 \bar \delta_{k} \Psi  = \skal{\tfrac{\delta}{\delta \bar g_{\mu \nu }} \Psi}{k_{\mu \nu}}. 
\]
Due to power-counting non-renormalizability, the arguments invoked in Section~\ref{section:YM-no anomaly} to prove the fulfillment of  \eqref{eq:A1(delta-Psi)=0} can not be adapted to the present setting, \cf also Remark~\ref{remark:YM-higher-dim}. For example, even if $A_1^\ia(s \bar \delta_k \Psi) = 0$ holds, one can, using the covariant coordinates, still find non-trivial analogs of $\Theta$ in the proof of Lemma~\ref{lemma:s0Exact}, such as
\[
 A_1^\ia(\bar \delta_k \Psi) = \int_M k_{\mu \nu} \Theta^{\mu \nu}
\]
with
\[
 \Theta^{\mu \nu} =  X[\bar g]^* \circ X[\bar g + h]_* ( \vol[\bar g + h] T[\bar g + h] )^{\mu \nu}
\]
for any covariant symmetric tensor $T$. We leave open the question whether such obstructions to background independence occur in perturbative quantum gravity.

\begin{remark}
\label{rem:BFR}
In one respect, our setup severely deviates from the one employed in \cite{Brunetti:2013maa}. 
There, the gauge condition is that the four curvature scalars $X$ that are used as coordinates are harmonic. The corresponding Lagrange multipliers $b$ are then a collection of four scalars, and accordingly for the antighosts $\bar c$. It follows that the gauge fixed action is no longer covariant but explicitly depends on the choice of the coordinates $X$. It is in fact not even invariant under changing the coordinates to $Y = \psi \circ X$ using a diffeomorphism $\psi$ of $\R^4$, i.e.,  under relabelling the points in the chart. The advantage of this approach is that the gauge fixing fermion does not break the split independence. The downside is of course that in the end one has to show that covariance is still intact in the observable algebra. 
Furthermore, having given up covariance, renormalization schemes and thus also potential anomalies are much less constrained than in our approach.
\end{remark}


\subsection{Background independence as triviality of the relative Cauchy evolution} \label{section:RCE}

Finally, let us comment on a different criterion for background independence, which is used in \cite{Brunetti:2013maa} in the context of perturbative quantum gravity.
Based on ideas formulated in \cite{Brunetti:2006qj}, background independence is there defined as triviality of the \emph{interacting relative Cauchy evolution} $\beta$. We first discuss it in the example of the scalar field. One defines
\[
 \beta_{\bphi', \bphi} \defeq R_\bphi(-; \eti{S_\ia})^{-1} \circ \tau_{\bphi \bphi'}^\ret \circ R_{\bphi'}(- ; \eti{S_\ia}) \circ A_{\bphi'}(-; \eti{S_\ia})^{-1} \circ \left( \tau_{\bphi \bphi'}^\adv \right)^{-1} \circ A_{\bphi}(- ; \eti{S_\ia}).
\]
Here $\tau^\adv$ is the \emph{advanced M{\o}ller operator}, defined in complete analogy to the retarded one, \cf \eqref{eq:def_MollerMap}, and $A$ is the \emph{advanced product}\footnote{In this section, $A$ denotes the advanced product, not the anomaly.} defined as
\[
 A( \eti{F} ; \eti{G} ) \defeq T( \eti{F} \otimes \eti{G} ) \star T( \eti{G} )^{-1}.
\]
The inverses of retarded and advanced products appearing here are purely formal. However, the requirement that $\beta$ is trivial on-shell can be properly formulated as
\[
 T_{\bphi}( \eti{S_\ia} ) \star \left( \tau^\ret_{\bphi \bphi'} R_{\bphi'}( \eti{F} ; \eti{S_\ia}) \right) \eqos \left( \tau^\adv_{\bphi \bphi'} A_{\bphi'}( \eti{F}; \eti{S_\ia}) \right) \star T_\bphi( \eti{S_\ia}).
\]
The infinitesimal version of this is, using perturbative agreement,
\begin{align*}
 0 & \eqos T( \eti{S_\ia} ) \star \left( \delta^\ret_{\bvp} R( \eti{F}; \eti{S_\ia}) \right) - \left( \delta^\adv_{\bvp} A( \eti{F}; \eti{S_\ia}) \right) \star T( \eti{S_\ia}) \\
 & = - \ibar A(\bar \delta_{\bvp} S; \eti{S_\ia}) \star T( \eti{F} \otimes \eti{S_\ia}) + \ibar T( \eti{F} \otimes \eti{S_\ia}) \star R(\bar \delta_{\bvp} S; \eti{S_\ia}) \\
 & = - \ibar T( \eti{S_\ia}) \star [ R(\bar \delta_{\bvp} S; \eti{S_\ia}), R(\eti{F}; \eti{S_\ia}) ]_{\star}.
\end{align*}
Formally, i.e., putting aside cut-off issues, we have $\delta_{\bvp} S_0 = 0$, so that, with \eqref{eq:ShiftSymmetryScalar}, we may replace $\bar \delta_{\bvp} S$ by $\delta_{\bvp} S$ and conclude that the equation is indeed fulfilled, by the field equation, which follows from \eqref{eq:T10}.

In the case of Yang-Mills theory, the split independence of the action is broken by gauge fixing, \cf \eqref{eq:ShiftSymmetryYM_gf}, so that one then obtains, again ignoring cut-off issues, 
\begin{multline*}
 T( \eti{S_\ia}) \star \left( \delta^\ret_{\ba} R( \eti{F}; \eti{S_\ia}) \right) - \left( \delta^\adv_{\ba} A( \eti{F}; \eti{S_\ia}) \right) \star T(\eti{S_\ia}) \\
 \eqos - \ibar T(\eti{S_\ia}) \star [ R( s \bar \delta_{\ba} \Psi; \eti{S_\ia}), R( \eti{F} ; \eti{S_\ia})]_\star.
\end{multline*}
For $F$ fulfilling \eqref{eq:s-nonlin-F=0} and assuming \eqref{eq:A1(delta-Psi)=0} and $[Q^\ia_{\bar \cA}, T(\eti{S_\ia})]_\star \eqos 0$, the \rhs can be written as an element of $\Ran [Q^\ia_{\bar{\cA}}, -]_\star$, i.e., as a trivial element. Hence, assuming the absence of the anomaly \eqref{eq:A1(delta-Psi)=0}, one finds that the interacting relative Cauchy evolution is indeed trivial on the cohomology.

Two comments are in order:
\begin{itemize}
\item As discussed in Remark~\ref{rem:BFR}, in \cite{Brunetti:2013maa} the breaking of the split independence of the action is avoided by the use of a non-covariant gauge fixing. In particular, the relevance of the absence of the anomaly \eqref{eq:A1(delta-Psi)=0} was not noted there. The problems with such a non-covariant gauge fixing were discussed in Remark~\ref{rem:BFR}.
\item The significance of the criterion proposed in \cite{Brunetti:2013maa}, i.e., triviality of the interacting relative Cauchy evolution, seems unclear. Following the derivation above, one finds that, in the case of gravity, it is implied by the on-shell vanishing of the stress-energy tensor $\bar \delta_{\bar k} S$, or, equivalently, by the on-shell fulfillment of the equations of motion. However, it gives no information about how to relate observables defined on different backgrounds, i.e., does not answer our initial question, 
as evidenced by the fact that all derivatives of $F$ \wrt the background fields drop out in the above calculations. We therefore think that triviality of the interacting relative Cauchy evolution is not a sufficient criterion for background independence.
\end{itemize}

\section*{Acknowledgement}
We thank Klaus Fredenhagen, Markus Fr\"ob, Thomas-Paul Hack, Kasia Rejzner, Martin Reuter, Pedro Ribeiro, Gerd Rudolph, Micha{\l} Wrochna, and Stefan Hollands for helpful discussions and/or valuable hints to the literature (the latter also for several suggestions for improving the manuscript).
Some of these discussion took place at the workshop ``Foundational and structural aspects of gauge theories'' at the Mainz Institute for Theoretical Physics (MITP).
M.T.T.\ is grateful to the MITP for support and hospitality during this workshop.
This work is part of M.T.T.s PhD dissertation. He gratefully acknowledges financial
support by the Max Planck Institute for Mathematics in the Sciences and its International
Max Planck Research School (IMPRS). 
\appendix

\section{Lemmata on the anomaly} \label{app:lemmata}
\begin{lemma}
\label{lemma:anomaly-field-ind}
The anomaly $A(e_\otimes^F)$ is (anti-) field-independent, in the sense that 
\begin{align}
\label{eq:Anomaly-field-ind}
\tfrac{\delta}{\delta \Phi^i(x)} A( \et{F} ) = (-1)^\eps A (\tfrac{\delta }{\delta \Phi^i(x)} F \otimes \et{F} ),	
\end{align}
 where $\Phi^i = (A_\mu^I, B^I, C^I, \bar{C}^I)$ and $\eps$ is the Grassmann parity of $\Phi^i$, and analogously for $\Phi^\ddagger_i$.
\end{lemma}

\begin{proof}
From the anomalous Ward identity \eqref{AnWI} and field-independence \eqref{eq:FieldIndependence} of time-ordered products, we have
\begin{align*}
\tfrac{\delta}{\delta \Phi^i} s_0 T( \eti{F}) &= (-1)^\eps \ibar T( \{ s_0 F + \tfrac{1}{2} (F, F) + A( \et{F}) \} \otimes \ibar \tfrac{\delta}{\delta \Phi^i} F \otimes \eti{F})   \\
& \quad + \ibar  T( \{ \tfrac{\delta}{\delta \Phi^i} s_0 F +  (\tfrac{\delta}{\delta \Phi^i} F, F) + \tfrac{\delta}{\delta \Phi^i} A( \et{F} ) \} \otimes \eti{F}).
\end{align*}
On the other hand, we have
\begin{align*}
s_0 ( \tfrac{\delta}{\delta \Phi^i}  T( \eti{F} ) ) &=  s_0  T( \ibar \tfrac{\delta}{\delta \Phi^i} F \otimes \eti{F})  \\
& = (-1)^\eps \ibar T( \ibar \tfrac{\delta}{\delta \Phi^i} F \otimes \{ s_0 F + \tfrac{1}{2} (F, F) + A(\et{F}) \} \otimes \eti{F}) \\
& \quad + \ibar T( \{ s_0 \tfrac{\delta}{\delta \Phi^i} F + (F, \tfrac{\delta}{\delta \Phi^i} F ) + A( \tfrac{\delta}{\delta \Phi^i} F \otimes \et{F}) \} \otimes \eti{F} ).
\end{align*}
We thus obtain
\begin{align} \label{eq:anomaly-field-ind-proof}
 [\tfrac{\delta}{\delta \Phi^i}, s_0]  T(\eti{F}) =\ibar  T( [\tfrac{\delta}{\delta \Phi^i}, s_0]  F \otimes \eti{F}) + \ibar T( \{ \tfrac{\delta}{\delta \Phi^i} A(\et{F}) - (-1)^\eps A(\tfrac{\delta}{\delta \Phi^i} F \otimes \et{F}) \}  \otimes \eti{F} ),
\end{align}
where
\[
 [\tfrac{\delta}{\delta \Phi^i}, s_0] \defeq \tfrac{\delta}{\delta \Phi^i} \circ s_0 - (-1)^\eps s_0 \circ \tfrac{\delta}{\delta \Phi^i}.
\]
It thus remains to show that
\begin{align} \label{eq:anomaly-field-ind-proof-2}
 [\tfrac{\delta}{\delta \Phi^i}, s_0]  T(\eti{F}) =\ibar  T( [\tfrac{\delta}{\delta \Phi^i}, s_0]  F \otimes \eti{F}),
\end{align}
which together with \eqref{eq:anomaly-field-ind-proof} implies the claim. To prove this, we note that $\frac{\delta S_0}{\delta \Phi^i}$ is a linear expression in fields and anti-fields, hence, can be written $\frac{\delta S_0}{\delta \Phi^i}= a_{i j} \Phi^j + {b_i}^j \Phi^{\ddag}_j$ for some (differential operator valued) coefficients $a_{i j}, {b^i}_j$. Thus,
\beq
\label{eq:del_s0_commutator}
 [\tfrac{\delta}{\delta \Phi^i}, s_0]  = (\tfrac{\delta}{\delta \Phi^i} S_0, -)= a_{i j} \tfrac{\delta }{\delta \Phi^\ddag_j} - {b_i}^j \tfrac{\delta }{\delta \Phi^j}.
\eeq
Therefore, \eqref{eq:anomaly-field-ind-proof-2} follows from (anti-) field independence of time-ordered products. For anti-field independence, the proof proceeds analogously.
\end{proof}


\begin{lemma}\label{lemma:anomaly-linear}
  Let $\Phi^i = (A_\mu^I, B^I, C^I, \bar{C}^I)$. Then $A(\Phi^i(x) \otimes \et{F})=0$, for all $F$, and analogously for antifields.
\end{lemma}

\begin{proof}
To prove the claim, we use the single field axiom, i.e., \eqref{eq:T10}, which in the present situation reads
\beq
\label{eq:LinearField_YM}
T(\Phi^i(x) \otimes \eti{F}) = \Phi^i(x) \star T(\eti{F}) - \int \Delta^{ij}_\adv(x, y) T(  \tfrac{\delta}{\delta \Phi^j(y)} F \otimes \eti{F}),
\eeq
where $\Delta^{ij}_\adv(x, y)$ is the advanced propagator of the differential operator $\bar P_{i j}$ defined in \eqref{eq:Def_P_ij}. We also recall the definitions \eqref{eq:Def_K_ij}, \eqref{eq:Def_hat_K_ij} of the (differential operator valued) matrices $K^i_{\ j}$ and $\hat K_i^{\ j}$. From \eqref{eq:s0_Phi_ddag} and $s_0^2 \Phi^\ddag_i = 0$, it follows that
\[
 \bar P_{i j} K^j_{\ k} + (-1)^\eps \hat K_i^{\ j} \bar P_{j k} = 0,
\]
with $\eps$ the Grassmann parity of $\Phi^i$.
This implies, \cf \cite{Rejzner:2013ak},
\beq
\label{eq:s0_Delta_ret}
 K^i_{\ j} \Delta^{j k}_{\ret/\adv} + (-1)^\eps \Delta^{i k}_{\ret/\adv} \hat K_k^{\ j} = 0.
\eeq

From \eqref{eq:LinearField_YM} it follows that for the linear fields $ s_0 \Phi^i$, we have
\beq 
 \label{PPA-s0-phi}
T(s_0 \Phi^i \otimes e_\otimes^{\ibar F}) =   T( K^i_{\ j} \Phi^j \otimes e_\otimes^{\ibar F}) = s_0 \Phi^i \star T(e_\otimes^{\ibar F}) - T( K^i_{\ j} \Delta^{jk}_\adv \tfrac{\delta}{\delta \Phi^k} F \otimes e_\otimes^{\ibar F} ),
\eeq
where for simplicity we omitted the variable $x$ and 
$\Delta^{jk}_\adv \tfrac{\delta}{\delta \Phi^k}$ should be read as $\int \Delta^{jk}_\adv(x,y) \tfrac{\delta}{\delta \Phi^k(y)}$.
To prove the claim, we first apply $s_0$ on the left hand side of \eqref{eq:LinearField_YM} and find
\begin{align} \nonumber
&s_0 T(\Phi^i \otimes \eti{F} )\\ \nonumber
& = T( \{ s_0 \Phi^i + (F, \Phi^i ) + A(\Phi^i \otimes \et{F}) \} \otimes \eti{F} )  + (-1)^\eps \ibar  T ( \Phi^i \otimes \{ s_0 F + \tfrac{1}{2}(F, F) + A(\et{F}) \} \otimes \eti{F} )\\ \nonumber
& = T( \{ s_0 \Phi^i + (F, \Phi^i ) + A(\Phi^i \otimes \et{F} ) \} \otimes \eti{F}) +   (-1)^\eps  \ibar  \Phi^i \star T( \{ s_0 F + \tfrac{1}{2}(F, F) + A(\et{F}) \} \otimes \eti{F})\\ \nonumber
& \quad - (-1)^\eps T ( \Delta^{ij}_\adv \tfrac{\delta }{\delta \Phi^j} \{ s_0 F + \tfrac{1}{2}(F, F) + A(\et{F}) \} \otimes  \eti{F} ) \\ \nonumber
& \quad - \ibar T( \{ s_0 F + \tfrac{1}{2}(F, F) + A(\et{F}) \} \otimes \Delta^{ij}_\adv \tfrac{\delta}{\delta \Phi^j} F \otimes \eti{F} ) \\ \nonumber
& = T( \{ s_0 \Phi^i + (F, \Phi^i ) + A( \Phi^i \otimes \et{F} ) \} \otimes \eti{F} ) +  (-1)^\eps \Phi^i \star s_0 T( \eti{F})\\ \nonumber
& \quad - (-1)^\eps  T ( \{ \Delta^{ij}_\adv \tfrac{\delta}{\delta \Phi^j} s_0 F + ( \Delta^{ij}_\adv \tfrac{\delta}{\delta \Phi^j} F, F) + (-1)^\eps A( \Delta^{ij}_\adv \tfrac{\delta}{\delta \Phi^j} F \otimes \et{F}) \}  \otimes \eti{F} ) \\ \label{s0-T(phi)-LHS}
& \quad - \ibar  T(  \{ s_0 F + \tfrac{1}{2}(F, F) + A(\et{F}) \} \otimes  \Delta^{ij}_\adv \tfrac{\delta}{\delta \Phi^j} F \otimes  \eti{F} ),
\end{align}
with $\eps$ the Grassmann parity of $\Phi^i$.
In the first step we have used the anomalous Ward identity \eqref{AnWI} and the sign factor appears by commuting $\Phi^i$ and $s_0 F + \frac{1}{2}(F, F) + A(\et{F})$, which is fermionic. In the second step, we used \eqref{eq:LinearField_YM} to pull $\Phi^i$ out of the time-ordered product, and in the last step we have again used \eqref{AnWI}
and the field independence of $ A(\et{F})$, i.e., \eqref{eq:Anomaly-field-ind}. 
Now applying $s_0$ on the \rhs of \eqref{eq:LinearField_YM} we find
\begin{align} \nonumber
 & s_0 \left[ \Phi^i \star T(\eti{F}) - T( \Delta^{ij}_\adv \tfrac{\delta}{\delta \Phi^j} F \otimes \eti{F}) \right] \\ \nn
 & = s_0 \Phi^i \star T(\eti{F}) + (-1)^\eps \Phi^i \star s_0 T( \eti{F}) \\ \nonumber
& \quad -  T( \{ s_0 \Delta^{ij}_\adv \tfrac{\delta}{\delta \Phi^j} F + (F, \Delta^{ij}_\adv \tfrac{\delta}{\delta \Phi^j} F ) + A(\Delta^{ij}_\adv \tfrac{\delta}{\delta \Phi^j} F \otimes \et{F} ) \} \otimes \eti{F} ) \\ \label{s0-T(phi)-RHS}
& \quad -  (-1)^\eps \ibar T ( \Delta^{ij}_\adv \tfrac{\delta}{\delta \Phi^j} F \otimes \{ s_0 F + \tfrac{1}{2}(F, F) + A(\et{F}) \} \otimes  \eti{F} ).
\end{align}
Equating \eqref{s0-T(phi)-LHS} and \eqref{s0-T(phi)-RHS} we arrive at
\begin{align} \nonumber
 T(A(\Phi^i \otimes \et{F} ) \otimes \eti{F}) & = s_0 \Phi^i \star T(\eti{F}) - T ( \{ s_0 \Phi^i + (F, \Phi^i ) \}  \otimes \eti{F} ) \\ \label{hat{A}(phi)}
 & \quad   +  (-1)^\eps  T ( ( \Delta^{ij}_\adv \tfrac{\delta}{\delta \Phi^j} S_0, F )  \otimes \eti{F} ),
\end{align}
where we have used \eqref{eq:del_s0_commutator}.
Using \eqref{retarded-prop} and noting that $S_0 = {S_0}|_{\Phi^\ddag=0} - \int \Phi^j \hat K_j^{\ k} \Phi^\ddag_k$, we find
\[
  ( \Delta^{ij}_\adv \tfrac{\delta}{\delta \Phi^j} S_0, F) = ( \Delta^{ij}_\adv \bar P_{j k} \Phi^k \vol - \Delta^{ij}_\adv \hat K_j^{\ k} \Phi^\ddag_k, F) = (\Phi^i, F) + \Delta^{ij}_\adv \hat K_j^{\ k} \tfrac{\delta}{\delta \Phi^k} F
\]
Inserting this back into \eqref{hat{A}(phi)} and using \eqref{eq:s0_Delta_ret}, we obtain
\[
 T(A( \Phi^i \otimes \et{F}  )  \otimes \eti{F}) =   s_0 \Phi^i \star T(\eti{F})  - T(s_0 \Phi^i \otimes \eti{F}) - T ( K^i_{\ j} \Delta^{j k}_\adv \tfrac{\delta}{\delta \Phi^k} F  \otimes \eti{F} ) ,
\]
which vanishes by \eqref{PPA-s0-phi}. This proves the claim.

For an anti-field $\Phi^\ddag_i$, the second term on the \rhs of \eqref{eq:LinearField_YM} is absent, and so are the last two terms on the \rhs of \eqref{s0-T(phi)-LHS} and \eqref{s0-T(phi)-RHS}. The claim then follows from \eqref{eq:s0_Phi_ddag}, which entails
\begin{align*}
 T( s_0 \Phi^\ddag_i \otimes \eti{F}) & = s_0 \Phi^\ddag_i \star T(\eti{F}) - (-1)^\eps T( \vol \bar P_{i j} \Delta_\adv^{j k} \tfrac{\delta}{\delta \Phi^k} F \otimes \eti{F}) \\
 & = s_0 \Phi^\ddag_i \star T(\eti{F}) + (-1)^\eps T((\Phi^\ddag_i, F) \otimes \eti{F}),
\end{align*}
with $\eps$ the Grassmann parity of $\Phi^i$.
\end{proof}

 
\printnoidxglossary[sort=def]


\begin{thebibliography}{10}

\bibitem{Abbott:1981ke}
L.F. Abbott,
\newblock {Introduction to the Background Field Method}, Acta Phys. Polon. B13
  (1982) 33.

\bibitem{Brunetti:2013maa}
R. Brunetti, K. Fredenhagen and K. Rejzner,
\newblock {Quantum gravity from the point of view of locally covariant quantum
  field theory}, Commun. Math. Phys. 345 (2016) 741, [arXiv:1306.1058].

\bibitem{Becker:2014qya}
D. Becker and M. Reuter,
\newblock {En route to Background Independence: Broken split-symmetry, and how
  to restore it with bi-metric average actions}, Annals Phys. 350 (2014) 225,
  [arXiv:1404.4537].

\bibitem{Hollands:2001nf}
S. Hollands and R.M. Wald,
\newblock {Local Wick polynomials and time ordered products of quantum fields
  in curved space-time}, Commun. Math. Phys. 223 (2001) 289,
  [arXiv:gr-qc/0103074].

\bibitem{Hollands:unpublished}
S. Hollands,
\newblock Background independence in quantum field theory, unpublished notes
  (2011).

\bibitem{Hollands:talk}
S. Hollands,
\newblock Constructing quantum field theories with fedosov quantization, 2012,
\newblock Talk given at workshop Mathematical Aspects of Quantum Field Theory
  and Quantum Statistical Mechanics, Hamburg, July 2012. Available at
  https://www.lqp2.org/node/1492.

\bibitem{fedosov1994}
B.V. Fedosov,
\newblock A simple geometrical construction of deformation quantization, J.
  Differential Geom. 40 (1994) 213.

\bibitem{ReuterMetaplectic}
M. Reuter,
\newblock {Quantum mechanics as a gauge theory of metaplectic spinor fields},
  Int. J. Mod. Phys. A13 (1998) 3835, [arXiv:hep-th/9804036].

\bibitem{Witten:1993ed}
E. Witten,
\newblock {Conference on Highlights of Particle and Condensed Matter Physics
  (SALAMFEST) Trieste, Italy, March 8-12, 1993}, pp. 257--275, 1993,
  [arXiv:hep-th/9306122].

\bibitem{Sen:1993kb}
A. Sen and B. Zwiebach,
\newblock {Quantum background independence of closed string field theory},
  Nucl. Phys. B423 (1994) 580, [arXiv:hep-th/9311009].

\bibitem{Brunetti:2001dx}
R. Brunetti, K. Fredenhagen and R. Verch,
\newblock {The Generally covariant locality principle: A New paradigm for local
  quantum field theory}, Commun. Math. Phys. 237 (2003) 31,
  [arXiv:math-ph/0112041].

\bibitem{Zahn:2012dz}
J. Zahn,
\newblock {The renormalized locally covariant Dirac field}, Rev. Math. Phys. 26
  (2014) 1330012, [arXiv:1210.4031].

\bibitem{hollands2005conservation}
S. Hollands and R.M. Wald,
\newblock Conservation of the stress tensor in perturbative interacting quantum
  field theory in curved spacetimes, Rev. Math. Phys. 17 (2005) 227.

\bibitem{Brennecke:2007uj}
F. Brennecke and M. D{\"u}tsch,
\newblock {Removal of violations of the Master Ward Identity in perturbative
  QFT}, Rev. Math. Phys. 20 (2008) 119, [arXiv:0705.3160].

\bibitem{collini2016fedosov}
G. Collini,
\newblock Fedosov Quantization and Perturbative Quantum Field Theory,
\newblock {PhD} dissertation, Universit\"at Leipzig, 2016, [arXiv:1503.03754].

\bibitem{Drago2017}
N. Drago, T.P. Hack and N. Pinamonti,
\newblock The generalised principle of perturbative agreement and the thermal
  mass, Annales Henri Poincar{\'e} 18 (2017) 807.

\bibitem{Khavkine:2015fwa}
I. Khavkine,
\newblock {Local and gauge invariant observables in gravity}, Class. Quant.
  Grav. 32 (2015) 185019, [arXiv:1503.03754].

\bibitem{BergmannKomar}
P.G. Bergmann and A.B. Komar,
\newblock {Poisson brackets between locally defined observables in general
  relativity}, Phys. Rev. Lett. 4 (1960) 432.

\bibitem{Benini:2017zjv}
M. Benini, A. Schenkel and U. Schreiber,
\newblock {The Stack of Yang-Mills Fields on Lorentzian Manifolds}, Commun.
  Math. Phys. 359 (2018) 765, [arXiv:1704.01378].

\bibitem{arms_1981}
J.M. Arms,
\newblock The structure of the solution set for the yang-mills equations,
  Mathematical Proceedings of the Cambridge Philosophical Society 90 (1981)
  361.

\bibitem{KlubergStern:1974xv}
H. Kluberg-Stern and J.B. Zuber,
\newblock {Renormalization of Nonabelian Gauge Theories in a Background Field
  Gauge. 1. Green Functions}, Phys. Rev. D12 (1975) 482.

\bibitem{WaldGR}
R.M. Wald,
\newblock General Relativity (University of Chicago Press, 1984).

\bibitem{Chilian:2008ye}
B. Chilian and K. Fredenhagen,
\newblock {The Time slice axiom in perturbative quantum field theory on
  globally hyperbolic spacetimes}, Commun. Math. Phys. 287 (2009) 513,
  [arXiv:0802.1642].

\bibitem{Brunetti:1999jn}
R. Brunetti and K. Fredenhagen,
\newblock {Microlocal analysis and interacting quantum field theories:
  Renormalization on physical backgrounds}, Commun. Math. Phys. 208 (2000) 623,
  [arXiv:math-ph/9903028].

\bibitem{Hollands:2007zg}
S. Hollands,
\newblock {Renormalized Quantum Yang-Mills Fields in Curved Spacetime}, Rev.
  Math. Phys. 20 (2008) 1033, [arXiv:0705.3340].

\bibitem{Hollands:2014eia}
S. Hollands and R.M. Wald,
\newblock {Quantum fields in curved spacetime}, Phys. Rept. 574 (2015) 1,
  [arXiv:1401.2026].

\bibitem{MR3469848}
K. Rejzner,
\newblock Perturbative algebraic quantum field theory (Springer, 2016).

\bibitem{Hollands:2001fb}
S. Hollands and R.M. Wald,
\newblock {Existence of local covariant time ordered products of quantum fields
  in curved space-time}, Commun. Math. Phys. 231 (2002) 309,
  [arXiv:gr-qc/0111108].

\bibitem{dewitt1960radiation}
B.S. DeWitt and R.W. Brehme,
\newblock Radiation damping in a gravitational field, Annals of Physics 9
  (1960) 220.

\bibitem{epstein1973role}
H. Epstein and V. Glaser,
\newblock The role of locality in perturbation theory, Annales de l'IHP Physique th{\'e}orique A 19 (1973) 211.

\bibitem{Zahn:2013ywa}
J. Zahn,
\newblock {Locally covariant charged fields and background independence}, Rev.
  Math. Phys. 27 (2015) 1550017, [arXiv:1311.7661].

\bibitem{Tehrani:2017pdx}
M. Taslimi~Tehrani,
\newblock {Quantum BRST charge in gauge theories in curved space-time}, J.
  Math. Phys. 60 (2019) 012304, [arXiv:1703.04148].

\bibitem{ChruscielShatah}
P.T. Chrusciel and J. Shatah,
\newblock {Global existence of solutions of the Yang-Mills equations on
  globally hyperbolic four dimensional Lorentzian manifolds}, Asian J. Math. 1
  (1997) 530.

\bibitem{becchi1976renormalization}
C. Becchi, A. Rouet and R. Stora,
\newblock Renormalization of gauge theories, Annals of Physics 98 (1976) 287.

\bibitem{BATALIN198127}
I. Batalin and G. Vilkovisky,
\newblock Gauge algebra and quantization, Physics Letters B 102 (1981) 27 .

\bibitem{Fredenhagen:2011mq}
K. Fredenhagen and K. Rejzner,
\newblock {Batalin-Vilkovisky formalism in perturbative algebraic quantum field
  theory}, Commun. Math. Phys. 317 (2013) 697, [arXiv:1110.5232].

\bibitem{Rejzner:2011au}
K. Rejzner,
\newblock {Fermionic fields in the functional approach to classical field
  theory}, Rev. Math. Phys. 23 (2011) 1009, [arXiv:1101.5126].

\bibitem{DeWitt:1967ub}
B.S. DeWitt,
\newblock {Quantum Theory of Gravity. 2. The Manifestly Covariant Theory},
  Phys. Rev. 162 (1967) 1195.

\bibitem{tHooft:1975uxh}
G. 't~Hooft,
\newblock An algorithm for the poles at dimension four in the dimensional
  regularization procedure, Nuclear Physics B 62 (1973) 444 .

\bibitem{HONERKAMP1972269}
J. Honerkamp,
\newblock The question of invariant renormalizability of the massless
  yang-mills theory in a manifest covariant approach, Nuclear Physics B 48
  (1972) 269 .

\bibitem{Boulware:1980av}
D.G. Boulware,
\newblock {Gauge Dependence of the Effective Action}, Phys. Rev. D23 (1981)
  389.

\bibitem{Iyer:1994ys}
V. Iyer and R.M. Wald,
\newblock {Some properties of Noether charge and a proposal for dynamical black
  hole entropy}, Phys. Rev. D50 (1994) 846, [arXiv:gr-qc/9403028].

\bibitem{Barnich:2000zw}
G. Barnich, F. Brandt and M. Henneaux,
\newblock {Local BRST cohomology in gauge theories}, Phys. Rept. 338 (2000)
  439, [arXiv:hep-th/0002245].

\bibitem{ManesStoraZumino}
J. Ma{\~{n}}es, R. Stora and B. Zumino,
\newblock {Algebraic Study of Chiral Anomalies}, Commun. Math. Phys. 102 (1985)
  157.

\bibitem{Peierls:1952cb}
R.E. Peierls,
\newblock {The Commutation laws of relativistic field theory}, Proc. Roy. Soc.
  Lond. A214 (1952) 143.

\bibitem{DeWitt:2004xz}
B. DeWitt and C. DeWitt-Morette,
\newblock {From the Peierls bracket to the Feynman functional integral}, Annals
  Phys. 314 (2004) 448.

\bibitem{Grassi:1995wr}
P.A. Grassi,
\newblock {Stability and renormalization of Yang-Mills theory with background
  field method: A Regularization independent proof}, Nucl. Phys. B462 (1996)
  524, [arXiv:hep-th/9505101].

\bibitem{Ferrari:2000yp}
R. Ferrari, M. Picariello and A. Quadri,
\newblock {Algebraic aspects of the background field method}, Annals Phys. 294
  (2001) 165, [arXiv:hep-th/0012090].

\bibitem{Anselmi:2013kba}
D. Anselmi,
\newblock {Background field method, Batalin-Vilkovisky formalism and parametric
  completeness of renormalization}, Phys. Rev. D89 (2014) 045004,
  [arXiv:1311.2704].

\bibitem{Becchi:1999ir}
C. Becchi and R. Collina,
\newblock {Further comments on the background field method and gauge invariant
  effective actions}, Nucl. Phys. B562 (1999) 412, [arXiv:hep-th/9907092].

\bibitem{Gerard:2014jba}
C. {G\'erard} and M. Wrochna,
\newblock {Hadamard States for the Linearized Yang-Mills Equation on Curved
  Spacetime}, Commun. Math. Phys. 337 (2015) 253, [arXiv:1403.7153].

\bibitem{WrochnaZahn}
M. Wrochna and J. Zahn,
\newblock {Classical phase space and Hadamard states in the BRST formalism for
  gauge field theories on curved spacetime}, Rev. Math. Phys. 29 (2017)
  1750014, [arXiv:1407.8079].

\bibitem{Zahn:2014uwa}
J. Zahn,
\newblock {Locally covariant chiral fermions and anomalies}, Nucl. Phys. B890
  (2014) 1, [arXiv:1407.1994].

\bibitem{Kugo:1979gm}
T. Kugo and I. Ojima,
\newblock {Local Covariant Operator Formalism of Nonabelian Gauge Theories and
  Quark Confinement Problem}, Prog. Theor. Phys. Suppl. 66 (1979) 1.

\bibitem{dutsch1999local}
M. D{\"u}tsch and K. Fredenhagen,
\newblock {A local (perturbative) construction of observables in gauge
  theories: the example of QED}, Comm. Math. Phys. 203 (1999) 71.

\bibitem{Frob:2018buw}
M.B. Fr{\"o}b,
\newblock {Anomalies in time-ordered products and applications to the BV-BRST
  formulation of quantum gauge theories}, Commun. Math. Phys. 372 (2019) 281,
  [arXiv:1803.10235].

\bibitem{Schenkel:2016nyj}
A. Schenkel and J. Zahn,
\newblock {Global anomalies on Lorentzian space-times}, Ann. Henri Poincar{\'e}
  18 (2017) 2693, [arXiv:1609.06562].

\bibitem{HoermanderI}
L. H{\"o}rmander,
\newblock The analysis of linear partial differential operators. {I}
  (Springer-Verlag, Berlin, 2003).

\bibitem{Ichinose:1992np}
S. Ichinose,
\newblock {BRS symmetry on background field, Kallosh theorem and
  renormalization}, Nucl. Phys. B395 (1993) 433.

\bibitem{Brunetti:2016hgw}
R. Brunetti et~al.,
\newblock {Cosmological perturbation theory and quantum gravity}, JHEP 08
  (2016) 032, [arXiv:1605.02573].

\bibitem{Frob:2017coq}
M.B. Fr{\"o}b, T.P. Hack and A. Higuchi,
\newblock {Compactly supported linearised observables in single-field
  inflation}, JCAP 1707 (2017) 043, [arXiv:1703.01158].

\bibitem{Barnich:1994kj}
G. Barnich, F. Brandt and M. Henneaux,
\newblock {General solution of the Wess-Zumino consistency condition for
  Einstein gravity}, Phys. Rev. D51 (1995) 1435, [arXiv:hep-th/9409104].

\bibitem{Brunetti:2006qj}
R. Brunetti and K. Fredenhagen,
\newblock Towards a background independent formulation of perturbative quantum
  gravity,
\newblock Quantum gravity, pp. 151--159, Springer, 2006.

\bibitem{Rejzner:2013ak}
K. Rejzner,
\newblock {Remarks on Local Symmetry Invariance in Perturbative Algebraic
  Quantum Field Theory}, Ann. Henri Poincar{\'e} 16 (2015) 205,
  [arXiv:1301.7037].

\end{thebibliography}

\end{document}